\documentclass[a4paper,UKenglish,cleveref, autoref, thm-restate]{lipics-v2021}

\usepackage{amsmath}
\usepackage{bm}
\usepackage{mathtools}
\usepackage{dashrule}

\usepackage{graphicx,epsfig,color}

\usepackage{xypic}
\xyoption{rotate}
\input xy
\xyoption{all}
\usepackage{verbatim}
\usepackage{amssymb}
\usepackage{eucal}
\usepackage{enumerate}
\usepackage[shortlabels]{enumitem}

\usepackage{multicol}
\usepackage{manfnt}

\usepackage{amsthm}
\usepackage{amsfonts}
\usepackage{mathrsfs}
\usepackage[all,2cell]{xy}
\CompileMatrices
\UseAllTwocells
\SilentMatrices
\usepackage{float}
\usepackage{fancybox}
\usepackage{hyperref}
\usepackage{wrapfig}
\usepackage{wasysym}

\usepackage{mathbbol}

\usepackage{stmaryrd}

\DeclareMathOperator{\Mat}{\mathsf{Mat}}

\mathcode`:="003A  
\mathcode`;="003B  
\mathcode`?="003F  
\mathcode`|="026A  

\mathchardef\ls="213C    
\mathchardef\gr="213E    
\mathchardef\uparrow="0222  
\mathchardef\downarrow="0223  

\newcommand{\from}{\mathrel{:}}

\newcommand{\scriptgeq}{\scalebox{0.5}{\ensuremath \geq}}

\newenvironment{Iff-RL}{\textbf{($\Rightarrow$)} }{\bigskip}
\newenvironment{Iff-LR}{\textbf{($\Leftarrow$)} }{}

\def \: {\colon}
\newcommand{\scolon}{\,;\,}

\newcommand{\CSFG}{{\mathsf{SPP}}}

\newcommand{\nullvec}{\bullet}

\makeatletter
\DeclareRobustCommand
  \myvdots{\vbox{\baselineskip4\p@ \lineskiplimit\z@
    \hbox{.}\hbox{.}\hbox{.}}}
\makeatother

\def \N {\mathbb{N}}

\newcommand{\vect}[1]{\begin{pmatrix}#1\end{pmatrix}}

\def \SAIHP {\mathbb{SaIH^{\scriptgeq}}} 

\newcommand{\dsems}[1]{\langle\!\langle #1 \rangle\!\rangle} 
\newcommand{\dsemsO}{\langle\!\langle \cdot \rangle\!\rangle} 

\DeclareMathSymbol{\reversedExclMark}{\mathord}{operators}{"3C}

\newcommand\Seqcomp{\lower5pt\hbox{$\includegraphics[height=.8cm]{graffles/seqcomp.pdf}$}}
\newcommand\Parcomp{\lower5pt\hbox{$\includegraphics[height=.8cm]{graffles/tensor.pdf}$}}

\newcommand\Gmult{\lower5pt\hbox{$\includegraphics[width=20pt]{graffles/Gmult.pdf}$}}
\newcommand\Gcomult{\lower5pt\hbox{$\includegraphics[width=20pt]{graffles/Gcomult.pdf}$}}
\newcommand\Gunit{\lower5pt\hbox{$\includegraphics[width=16pt]{graffles/Gunit.pdf}$}}
\newcommand\Gcounit{\lower5pt\hbox{$\includegraphics[width=16pt]{graffles/Gcounit.pdf}$}}
\newcommand\twoGcounit{\lower5pt\hbox{$\includegraphics[width=16pt]{graffles/twoGcounit.pdf}$}}

\newcommand\ASepUno{\lower3pt\hbox{$\includegraphics[width=28pt]{graffles/ASep1.pdf}$}}
\newcommand\ASepDue{\lower3pt\hbox{$\includegraphics[width=24pt]{graffles/ASep2.pdf}$}}

\newcommand\BSepOne{\lower3pt\hbox{\includegraphics[width=30pt]{graffles/BlackSep1.pdf}}}
\newcommand\BlackX{\lower5pt\hbox{\includegraphics[width=30pt]{graffles/BlackX.pdf}}}

\newcommand{\onezero}{
\tikzset{x=1em, y=2.1ex}
\begin{tikzpicture}
	\begin{pgfonlayer}{nodelayer}
		\node [style=none] (0) at (-0.75, 0) {};
		\node [style=white] (1) at (0.75, 0) {};
		\node [style=none] (2) at (-0.75, 0.3) {};
		\node [style=none] (3) at (-0.75, -0.3) {};
	\end{pgfonlayer}
	\begin{pgfonlayer}{edgelayer}
		\draw (0.center) to (1);
		\draw (2.center) to (3.center);
	\end{pgfonlayer}
\end{tikzpicture}}
\tikzset{x=1em, y=1.5ex}
}

\newcommand{\xcirc}[1]{\begin{tikzpicture}
	\begin{pgfonlayer}{nodelayer}
		\node [style=square] (0) at (0, 0) {$x$};
		\node [style=none] (1) at (-0.75, 0.25) {};
		\node [style=none] (2) at (-0.75, -0.25) {};
		\node [style=none] (3) at (0.75, 0.25) {};
		\node [style=#1] (4) at (0.75, 0) {};
		\node [style=none] (5) at (0.75, -0.25) {};
		\node [style=none] (6) at (0.25, 0.25) {};
		\node [style=none] (7) at (0.25, 0) {};
		\node [style=none] (8) at (0.25, -0.25) {};
		\node [style=none] (9) at (-0.25, -0.25) {};
		\node [style=none] (10) at (-0.25, 0.25) {};
	\end{pgfonlayer}
	\begin{pgfonlayer}{edgelayer}
		\draw (10.center) to (1.center);
		\draw (2.center) to (9.center);
		\draw (6.center) to (3.center);
		\draw (7.center) to (4.center);
		\draw (8.center) to (5.center);
	\end{pgfonlayer}
\end{tikzpicture}
}

\def \Sem {\sem{\QA}} 

\newcommand{\op}{op}

\def \field {\mathsf{k}} 

\newcommand \IH[1]{\mathbb{IH}_{\scriptscriptstyle #1}}

\newcommand \IHP[1]{\mathbb{IH}_{\scriptscriptstyle #1}^{\scriptgeq}}
\newcommand \AIHP[1]{\mathbb{aIH}_{\scriptscriptstyle #1}^{\scriptgeq}}

\def \poi {\,\ensuremath{;}\,}

\newcommand\unitscalar{\lower4pt\hbox{$\includegraphics[width=22pt]{graffles/unitscalar.pdf}$}}
\newcommand\scalarminusone{\lower8pt\hbox{$\includegraphics[width=30pt]{graffles/scalarminusone.pdf}$}}
\newcommand\antipode{\lower3pt\hbox{$\includegraphics[width=22pt]{graffles/antipode.pdf}$}}
\newcommand\antipodeop{\lower3pt\hbox{$\includegraphics[width=22pt]{graffles/antipodeop.pdf}$}}
\newcommand\antipodesquare{\lower3pt\hbox{$\includegraphics[width=22pt]{graffles/antipodesquare.pdf}$}}
\newcommand\circuitAdots{\lower6pt\hbox{$\includegraphics[width=30pt]{graffles/circuitAdots.pdf}$}}

\newcommand\wcounitn{\lower5pt\hbox{$\includegraphics[width=25pt]{graffles/wcounitn.pdf}$}}
\newcommand\bcounitn{\lower5pt\hbox{$\includegraphics[width=25pt]{graffles/bcounitn.pdf}$}}
\newcommand\lccn{\lower5pt\hbox{$\includegraphics[width=25pt]{graffles/lccn.pdf}$}}
\newcommand\rccn{\lower5pt\hbox{$\includegraphics[width=25pt]{graffles/rccn.pdf}$}}
\newcommand\idncircuit{\lower5pt\hbox{$\includegraphics[width=25pt]{graffles/idncircuit.pdf}$}}
\newcommand\circuitrbcounits{\lower5pt\hbox{$\includegraphics[width=25pt]{graffles/circuitrbcounits.pdf}$}}
\newcommand\lccB{\lower5pt\hbox{$\includegraphics[width=25pt]{graffles/rccr.pdf}$}}
\newcommand\rccB{\lower5pt\hbox{$\includegraphics[width=25pt]{graffles/lccl.pdf}$}}
\newcommand\IdBcounitc{\lower5pt\hbox{$\includegraphics[width=20pt]{graffles/IdBcounit.pdf}$}}
\newcommand\BcounitId{\lower5pt\hbox{$\includegraphics[width=20pt]{graffles/BcounitId.pdf}$}}
\newcommand\symNetTwoOne{\lower7pt\hbox{$\includegraphics[width=25pt]{graffles/symNet21.pdf}$}}
\newcommand\nscalar{\!\!\lower5pt\hbox{$\includegraphics[width=35pt]{graffles/nscalar.pdf}$}\!\!}
\newcommand\Wmultstar{\!\lower5pt\hbox{$\includegraphics[width=20pt]{graffles/Wmultstar.pdf}$}\!}
\newcommand\scalarstar{\!\!\lower7pt\hbox{$\includegraphics[width=35pt]{graffles/scalarstar.pdf}$}\!\!}
\newcommand\twoBcounit{\!\!\lower8pt\hbox{$\includegraphics[width=20pt]{graffles/lunitsr.pdf}$}\!\!}

\newcommand\delay{\!\lower6pt\hbox{$\includegraphics[width=25pt]{graffles/delaycircuit.pdf}$}\!}
\newcommand\scalarp{\!\lower4pt\hbox{$\includegraphics[width=22pt]{graffles/scalarp.pdf}$}\!}
\newcommand\scalarpstar{\!\lower3pt\hbox{$\includegraphics[width=25pt]{graffles/scalarpstar.pdf}$}\!}
\newcommand\scalarpop{\!\lower4pt\hbox{$\includegraphics[width=22pt]{graffles/scalarpop.pdf}$}\!}
\newcommand\nscalarp{\!\lower3pt\hbox{$\includegraphics[width=30pt]{graffles/nscalarp.pdf}$}\!}
\newcommand\circuitfibr{\!\lower3pt\hbox{$\includegraphics[width=44pt]{graffles/circuitfibr.pdf}$}\!}
\newcommand\ncircuitX{\!\lower3pt\hbox{$\includegraphics[width=30pt]{graffles/ncircuitX.pdf}$}\!}
\newcommand\scalarpone{\!\lower3pt\hbox{$\includegraphics[width=22pt]{graffles/scalarpone.pdf}$}\!}
\newcommand\scalarptwoop{\!\lower3pt\hbox{$\includegraphics[width=22pt]{graffles/scalarptwoop.pdf}$}\!}

\newcommand{\Idnet}{
\tikzset{x=1em, y=2.1ex}
\begin{tikzpicture}
	\begin{pgfonlayer}{nodelayer}
		\node [style=none] (0) at (0.7, 0) {};
		\node [style=none] (1) at (-0.7, 0) {};
	\end{pgfonlayer}
	\begin{pgfonlayer}{edgelayer}
		\draw (0.center) to (1.center);
	\end{pgfonlayer}
\end{tikzpicture}
}
\tikzset{x=1em, y=1.5ex}
}

\newcommand{\Wunit}{
\tikzset{x=1em, y=2.1ex}
\begin{tikzpicture}[baseline=-.5ex, scale=.7]
	\begin{pgfonlayer}{nodelayer}
		\node [style=white] (0) at (0.25, 0) {};
		\node [style=none] (1) at (1.75, 0) {};
	\end{pgfonlayer}
	\begin{pgfonlayer}{edgelayer}
		\draw (0) to (1.center);
	\end{pgfonlayer}
\end{tikzpicture}
}
\tikzset{x=1em, y=1.5ex}
}
\newcommand{\Bcounit}{
\tikzset{x=1em, y=2.1ex}
\begin{tikzpicture}[baseline=-.5ex, scale=.7]
	\begin{pgfonlayer}{nodelayer}
		\node [style=black] (0) at (1.5, 0) {};
		\node [style=none] (1) at (0, 0) {};
	\end{pgfonlayer}
	\begin{pgfonlayer}{edgelayer}
		\draw (0) to (1.center);
	\end{pgfonlayer}
\end{tikzpicture}
}
\tikzset{x=1em, y=1.5ex}
}
 \newcommand{\Bcomult}{
\tikzset{x=1em, y=2.1ex}
\InputIfFileExists{./generators/copy.tikz}{}{\input{./tikz/./generators/copy.tikz}}
\tikzset{x=1em, y=1.5ex}
}
\newcommand{\Wmult}{
\tikzset{x=1em, y=2.1ex}
\InputIfFileExists{./generators/add.tikz}{}{\input{./tikz/./generators/add.tikz}}
\tikzset{x=1em, y=1.5ex}
}
\newcommand{\scalar}{
\tikzset{x=1em, y=2.1ex}
\begin{tikzpicture}
	\begin{pgfonlayer}{nodelayer}
		\node [style=reg] (0) at (-0.25, 0) {$k$};
		\node [style=none] (1) at (1, 0) {};
		\node [style=none] (2) at (-1.5, 0) {};
	\end{pgfonlayer}
	\begin{pgfonlayer}{edgelayer}
		\draw (2.center) to (0);
		\draw (0) to (1.center);
	\end{pgfonlayer}
\end{tikzpicture}
}
\tikzset{x=1em, y=1.5ex}
}
\newcommand{\circuitX}{
\tikzset{x=1em, y=2.1ex}
\begin{tikzpicture}
	\begin{pgfonlayer}{nodelayer}
		\node [style=reg] (0) at (-0.25, -0) {$x$};
		\node [style=none] (1) at (1, -0) {};
		\node [style=none] (2) at (-1.5, -0) {};
	\end{pgfonlayer}
	\begin{pgfonlayer}{edgelayer}
		\draw (2.center) to (0);
		\draw (0) to (1.center);
	\end{pgfonlayer}
\end{tikzpicture}}
\tikzset{x=1em, y=1.5ex}
}
\newcommand{\Bunit}{
\tikzset{x=1em, y=2.1ex}
\begin{tikzpicture}[baseline=-.5ex, scale=.7]
	\begin{pgfonlayer}{nodelayer}
		\node [style=black] (0) at (0.25, 0) {};
		\node [style=none] (1) at (1.75, 0) {};
	\end{pgfonlayer}
	\begin{pgfonlayer}{edgelayer}
		\draw (0) to (1.center);
	\end{pgfonlayer}
\end{tikzpicture}
}
\tikzset{x=1em, y=1.5ex}
}
\newcommand{\Wcounit}{
\tikzset{x=1em, y=2.1ex}
\begin{tikzpicture}
	\begin{pgfonlayer}{nodelayer}
		\node [style=white] (0) at (1.5, 0) {};
		\node [style=none] (1) at (0.5, 0) {};
	\end{pgfonlayer}
	\begin{pgfonlayer}{edgelayer}
		\draw (0) to (1.center);
	\end{pgfonlayer}
\end{tikzpicture}
}
\tikzset{x=1em, y=1.5ex}
}
\newcommand{\Wcomult}{
\tikzset{x=1em, y=2.1ex}
\InputIfFileExists{./generators/co-add.tikz}{}{\input{./tikz/./generators/co-add.tikz}}
\tikzset{x=1em, y=1.5ex}
}
\newcommand{\Bmult}{
\tikzset{x=1em, y=2.1ex}
\InputIfFileExists{./generators/co-copy.tikz}{}{\input{./tikz/./generators/co-copy.tikz}}
\tikzset{x=1em, y=1.5ex}
}

\tikzset{x=1em, y=1.5ex}
}

\tikzset{x=1em, y=1.5ex}
}

\newcommand{\IdnetT}{
\tikzset{x=1em, y=2.1ex}
}
\tikzset{x=1em, y=1.5ex}
}
\newcommand{\symNetT}{
\tikzset{x=1em, y=2.1ex}
\InputIfFileExists{./generators/sym.tikz}{}{\input{./tikz/./generators/sym.tikz}}
\tikzset{x=1em, y=1.5ex}
}
\newcommand{\WunitT}{
\tikzset{x=1em, y=2.1ex}
}
\tikzset{x=1em, y=1.5ex}
}
\newcommand{\BcounitT}{
\tikzset{x=1em, y=2.1ex}
}
\tikzset{x=1em, y=1.5ex}
}
 \newcommand{\BcomultT}{
\tikzset{x=1em, y=2.1ex}
\InputIfFileExists{./generators/copy.tikz}{}{\input{./tikz/./generators/copy.tikz}}
\tikzset{x=1em, y=1.5ex}
}
\newcommand{\WmultT}{
\tikzset{x=1em, y=2.1ex}
\InputIfFileExists{./generators/add.tikz}{}{\input{./tikz/./generators/add.tikz}}
\tikzset{x=1em, y=1.5ex}
}
\newcommand{\scalarT}{
\tikzset{x=1em, y=2.1ex}
\begin{tikzpicture}
	\begin{pgfonlayer}{nodelayer}
		\node [style=reg] (0) at (-0.25, 0) {$k$};
		\node [style=none] (1) at (1, 0) {};
		\node [style=none] (2) at (-1.5, 0) {};
	\end{pgfonlayer}
	\begin{pgfonlayer}{edgelayer}
		\draw (2.center) to (0);
		\draw (0) to (1.center);
	\end{pgfonlayer}
\end{tikzpicture}
}
\tikzset{x=1em, y=1.5ex}
}
\newcommand{\circuitXT}{
\tikzset{x=1em, y=2.1ex}
\begin{tikzpicture}
	\begin{pgfonlayer}{nodelayer}
		\node [style=reg] (0) at (-0.25, -0) {$x$};
		\node [style=none] (1) at (1, -0) {};
		\node [style=none] (2) at (-1.5, -0) {};
	\end{pgfonlayer}
	\begin{pgfonlayer}{edgelayer}
		\draw (2.center) to (0);
		\draw (0) to (1.center);
	\end{pgfonlayer}
\end{tikzpicture}}
\tikzset{x=1em, y=1.5ex}
}
\newcommand{\BunitT}{
\tikzset{x=1em, y=2.1ex}
\begin{tikzpicture}[baseline=-.5ex, scale=.7]
	\begin{pgfonlayer}{nodelayer}
		\node [style=black] (0) at (0.25, 0) {};
		\node [style=none] (1) at (1.75, 0) {};
	\end{pgfonlayer}
	\begin{pgfonlayer}{edgelayer}
		\draw (0) to (1.center);
	\end{pgfonlayer}
\end{tikzpicture}
}
\tikzset{x=1em, y=1.5ex}
}
\newcommand{\WcounitT}{
\tikzset{x=1em, y=2.1ex}
}
\tikzset{x=1em, y=1.5ex}
}
\newcommand{\WcomultT}{
\tikzset{x=1em, y=2.1ex}
\InputIfFileExists{./generators/co-add.tikz}{}{\input{./tikz/./generators/co-add.tikz}}
\tikzset{x=1em, y=1.5ex}
}
\newcommand{\BmultT}{
\tikzset{x=1em, y=2.1ex}
\InputIfFileExists{./generators/co-copy.tikz}{}{\input{./tikz/./generators/co-copy.tikz}}
\tikzset{x=1em, y=1.5ex}
}
\newcommand{\scalaropT}{
\tikzset{x=1em, y=2.1ex}
\begin{tikzpicture}
	\begin{pgfonlayer}{nodelayer}
		\node [style=coreg] (0) at (-0.25, 0) {$k$};
		\node [style=none] (1) at (1, 0) {};
		\node [style=none] (2) at (-1.5, 0) {};
	\end{pgfonlayer}
	\begin{pgfonlayer}{edgelayer}
		\draw (2.center) to (0);
		\draw (0) to (1.center);
	\end{pgfonlayer}
\end{tikzpicture}
}
\tikzset{x=1em, y=1.5ex}
}

\tikzset{x=1em, y=1.5ex}
}
\newcommand{\ZeronetT}{
\tikzset{x=1em, y=2.1ex}
\InputIfFileExists{./generators/empty-diag.tikz}{}{\input{./tikz/./generators/empty-diag.tikz}}
\tikzset{x=1em, y=1.5ex}
}

\newcommand{\greq}{
\tikzset{x=1em, y=2.1ex}
\begin{tikzpicture}
	\begin{pgfonlayer}{nodelayer}
		\node [style=reg] (0) at (-0.25, -0) {$\geq$};
		\node [style=none] (1) at (1, -0) {};
		\node [style=none] (2) at (-1.5, -0) {};
	\end{pgfonlayer}
	\begin{pgfonlayer}{edgelayer}
		\draw (2.center) to (0);
		\draw (0) to (1.center);
	\end{pgfonlayer}
\end{tikzpicture}}
\tikzset{x=1em, y=1.5ex}
}

\newcommand{\leqd}{
\tikzset{x=1em, y=2.1ex}
\begin{tikzpicture}
	\begin{pgfonlayer}{nodelayer}
		\node [style=coreg] (0) at (0, 0) {$\leq$};
		\node [style=none] (1) at (1.25, 0) {};
		\node [style=none] (2) at (-1.25, 0) {};
	\end{pgfonlayer}
	\begin{pgfonlayer}{edgelayer}
		\draw (2.center) to (0);
		\draw (0) to (1.center);
	\end{pgfonlayer}
\end{tikzpicture}
}
\tikzset{x=1em, y=1.5ex}
}

\newcommand{\cupd}{
\tikzset{x=1em, y=2.1ex}
\InputIfFileExists{cup.tikz}{}{\input{./tikz/cup.tikz}}
\tikzset{x=1em, y=1.5ex}
}
\newcommand{\capd}{
\tikzset{x=1em, y=2.1ex}
\InputIfFileExists{cap.tikz}{}{\input{./tikz/cap.tikz}}
\tikzset{x=1em, y=1.5ex}
}

\newcommand{\cupnd}{
\tikzset{x=1em, y=2.1ex}
\InputIfFileExists{cup-n.tikz}{}{\input{./tikz/cup-n.tikz}}
\tikzset{x=1em, y=1.5ex}
}
\newcommand{\capnd}{
\tikzset{x=1em, y=2.1ex}
\InputIfFileExists{cap-n.tikz}{}{\input{./tikz/cap-n.tikz}}
\tikzset{x=1em, y=1.5ex}
}

\newcommand{\cbox}{
\tikzset{x=1em, y=2.1ex}
\InputIfFileExists{cbox.tikz}{}{\input{./tikz/cbox.tikz}}
\tikzset{x=1em, y=1.5ex}
}
\newcommand{\cboxop}{
\tikzset{x=1em, y=2.1ex}
\InputIfFileExists{cboxop.tikz}{}{\input{./tikz/cboxop.tikz}}
\tikzset{x=1em, y=1.5ex}
}

\newcommand{\circuitminusone}{\!\lower5pt\hbox{$\includegraphics[width=20pt,height=15pt]{graffles/circuitminusone.pdf}$}\!}
\newcommand{\circuitminusoneop}{\!\lower5pt\hbox{$\includegraphics[width=20pt,height=15pt]{graffles/circuitminusoneop.pdf}$}\!}

\newcommand{\lbbd}{\left \llbracket}
\newcommand{\rbbd}{\right \rrbracket}
\newcommand{\dsem}[1]{\lbbd #1 \rbbd}
\newcommand{\dsemO}{\dsem{\cdot}}
\newcommand{\dsemp}[1]{\dsem{#1}}
\newcommand{\dsempO}{\dsem{\cdot}}

\newcommand{\dsempr}[1]{\dsem{#1}'}
\newcommand{\dsemprO}{\dsemp{\cdot}}

\newcommand{\lbbo}{\mathopen{\langle}}
\newcommand{\rbbo}{\mathclose{\rangle}}
\newcommand{\osem}[1]{\lbbo #1 \rbbo}

\newcommand\idzcircuit{\lower5pt\hbox{$\includegraphics[width=20pt]{graffles/idzcircuit.pdf}$}}
\newcommand{\circuitXspan}{\!\lower4pt\hbox{$\includegraphics[width=40pt]{graffles/circuitXspan.pdf}$}\!}
\newcommand{\circuitXcospan}{\!\lower5pt\hbox{$\includegraphics[width=40pt]{graffles/circuitXcospan.pdf}$}\!}
\newcommand\zeroscalar{\lower3pt\hbox{$\includegraphics[width=20pt]{graffles/zeroscalar.pdf}$}}
\newcommand\zeroscalarr{\lower3pt\hbox{$\includegraphics[width=25pt]{graffles/zeroscalar2.pdf}$}}
\newcommand\rationalcircuit{\lower5pt\hbox{$\includegraphics[width=35pt]{graffles/rationalcircuit.pdf}$}}
\newcommand\Wccl{\lower5pt\hbox{$\includegraphics[width=22pt]{graffles/Wccl.pdf}$}}
\newcommand{\circuitkkop}{\!\lower4pt\hbox{$\includegraphics[width=32pt]{graffles/circuitkkop.pdf}$}\!}

\usepackage{centernot}

\newcommand\circuitUnoMinusX{\lower4.5pt\hbox{$\includegraphics[width=32pt]{graffles/circuit1-x.pdf}$}}
\newcommand\circuitUnoMinusXSquare{\lower5pt\hbox{$\includegraphics[width=35pt]{graffles/circuit1-xsquare.pdf}$}}

\usepackage{prooftree}

\newcommand{\ruleLabel}[1]{#1}

\newcommand{\bnfSep}{\;\; | \;\;}

\makeatletter
\def\moverlay{\mathpalette\mov@rlay}
\def\mov@rlay#1#2{\leavevmode\vtop{%
\baselineskip\z@skip \lineskiplimit-\maxdimen
\ialign{\hfil$#1##$\hfil\cr#2\crcr}}}
\makeatother

 \newcommand{\derivationRule}[3]{{\prooftree{ #1}\justifies{ #2}\using\ruleLabel{#3}\endprooftree}}

\newcommand\twarr[2]{%
\mathrel{\mathop{\moverlay{\scriptstyle\xrightarrow{\,#1\,}\cr{\lower.2em\hbox{$\scriptstyle{}_{#2}$}}}}}}
\newcommand\twarrw[2]{%
\mathrel{\mathop{\moverlay{\scriptstyle\Longrightarrow\cr{\lower-.6em\hbox{$\scriptstyle{}_{#1}$}}
\cr{\lower.3em\hbox{$\scriptstyle{}_{#2}$}}}}}}

\newcommand{\dtrans}[2]{\hbox{$\;\twarr{#1}{#2}\;$}}
\newcommand{\dtransw}[2]{\raise1pt\hbox{$\;\twarrw{#1}{#2}\;$}}

\def \CD {\mathsf{Circ}}
\newcommand{\FC}{\mathsf{C}\overrightarrow{\hspace{-.1cm}\scriptstyle\mathsf{irc}}}
\newcommand{\FCop}{\mathsf{C}\overleftarrow{\hspace{-.1cm}\scriptstyle\mathsf{irc}}}

\def \Syn {\mathsf{Syn}}
\def \Sem {\mathsf{Sem}}
\def \PSigma {\mathsf{P}_\Sigma}
\def \PSigmaE {\mathsf{P}_{\Sigma,E}}
\def \IM {\mathsf{IM}}

\def \CDP {\mathsf{Circ^{\scriptgeq}}}

\usepackage{bbm}

\usepackage{tikz}
\usetikzlibrary{arrows,backgrounds,shapes.geometric,shapes.misc,matrix,arrows.meta,decorations.markings}

\pgfdeclarelayer{background}
\pgfdeclarelayer{edgelayer}
\pgfdeclarelayer{nodelayer}
\pgfsetlayers{background,edgelayer,nodelayer,main}
\tikzset{baseline=-0.5ex}
\definecolor{light-gray}{gray}{.7}
\tikzstyle{none}=[inner sep=0pt]
\tikzstyle{plain}=[inner sep=0pt]
\tikzstyle{black}=[circle, draw=black, fill=black, inner sep=0pt, minimum size=3.5pt]
\tikzstyle{black-faded}=[circle, draw=light-gray, fill=light-gray, inner sep=0pt, minimum size=4pt]
\tikzstyle{white}=[circle, draw=black, fill=white, inner sep=0pt, minimum size=3.5pt]
\tikzstyle{white-faded}=[circle, draw=light-gray, fill=white, inner sep=0pt, minimum size=4.5pt]
\tikzstyle{delay}=[fill=black, regular polygon, regular polygon sides=3,rotate=-90, scale=.55]
\tikzstyle{delay-op}=[fill=black, regular polygon, regular polygon sides=3,rotate=90, scale=.55]
\tikzstyle{reg}=[draw, fill=white, rounded rectangle, rounded rectangle left arc=none, minimum height=1em, minimum width=1em, node font={\scriptsize}]
\tikzstyle{coreg}=[draw, fill=white, rounded rectangle, rounded rectangle right arc=none, minimum height=1em, minimum width=1em, node font={\scriptsize}]
\tikzstyle{rn}=[circle, draw=red, fill=red, inner sep=0pt, minimum size=4pt]
\tikzstyle{place}=[circle, draw=black, fill=white, inner sep=0pt, minimum size=8pt]

\tikzstyle{medium box}=[fill=white, draw=black, shape=rectangle, minimum height=1cm, minimum width=0.75cm]
\tikzstyle{small box}=[fill=white, draw=black, shape=rectangle, minimum height=0.75cm, minimum width=0.5cm]
\tikzstyle{square}=[fill=white, draw=black, shape=rectangle, minimum height=1cm, minimum width=1cm]
\tikzstyle{triangle}=[fill=black, draw=black, shape=regular polygon, regular polygon sides=3, rotate=270, scale=0.6]
\tikzstyle{antipode}=[fill=black, draw=black, shape=rectangle]
\tikzstyle{triangleop}=[fill=black, draw=black, shape=regular polygon, regular polygon sides=3, rotate=90, scale=0.6]
\tikzstyle{edge}=[fill=lightgray, draw=white, shape=rectangle, node font={\scriptsize}, text=black]
\tikzstyle{label}=[fill=none, draw=none, shape=circle]
\tikzstyle{transition}=[fill=white, draw=black, shape=rectangle, minimum height=0.75cm, minimum width=0.75cm]

\tikzstyle{transition}=[rectangle,thick,draw=black!75,
  			  fill=black!20,minimum size=7pt]
\tikzstyle{place}=[circle, draw=black, fill=white, inner sep=0pt, minimum size=8pt]
\tikzstyle{vertex}=[circle, draw=black, fill=white, inner sep=0pt, minimum size=15pt]
\tikzstyle{arrow}=[->]

\newcommand{\tikzfig}[1]{
\tikzset{x=1em, y=2.1ex}
\InputIfFileExists{#1.tikz}{}{\input{./tikz/#1.tikz}}
\tikzset{x=1em, y=1.5ex}
}

\usepackage[most]{tcolorbox}
\usepackage{xstring}

\newcommand{\myeq}[1]{\stackrel{#1}{=}}
\newcommand{\mysubeq}[1]{\stackrel{#1}{\subseteq}}
\newcommand{\mysupeq}[1]{\stackrel{#1}{\supseteq}}

\newcommand{\One}{
\tikzset{x=1em, y=2.1ex}
\begin{tikzpicture}[baseline=-.5ex]
	\begin{pgfonlayer}{nodelayer}
		\node [style=none] (0) at (-0.75, -0) {};
		\node [style=none] (1) at (0.25, -0) {};
		\node [style=none] (2) at (-0.75, 0.25) {};
		\node [style=none] (3) at (-0.75, -0.25) {};
	\end{pgfonlayer}
	\begin{pgfonlayer}{edgelayer}
		\draw (0.center) to (1.center);
		\draw (3.center) to (2.center);
	\end{pgfonlayer}
\end{tikzpicture}}
\tikzset{x=1em, y=1.5ex}
}

\tikzset{x=1em, y=1.5ex}
}

\def \ACDP {\mathsf{ACirc}^{\scriptgeq}}

\newcommand \LinRel[1]{{\mathsf{LinRel}}_{\scriptscriptstyle #1}}
\newcommand{\RelX}[1]{\mathsf{Rel}_{\scriptstyle #1}}

\newcommand{\propCat}[1]{\mathsf{#1}}

\newcommand{\tcoscalar}[1]{
\begin{tikzpicture}[x=1em, y=2.1ex]
	\begin{pgfonlayer}{nodelayer}
		\node [style=coreg] (0) at (-0.25, 0) {$#1$};
		\node [style=none] (1) at (1.25, 0) {};
		\node [style=none] (2) at (-1.75, 0) {};
	\end{pgfonlayer}
	\begin{pgfonlayer}{edgelayer}
		\draw (2.center) to (0);
		\draw (0) to (1.center);
	\end{pgfonlayer}
\end{tikzpicture}}

\newcommand{\tantipode}{
\tikzset{x=1em, y=2.1ex}
\begin{tikzpicture}
	\begin{pgfonlayer}{nodelayer}
		\node [style=antipode] (3) at (0, 0) {};
		\node [style=none] (13) at (1, 0) {};
		\node [style=none] (14) at (-1, 0) {};
	\end{pgfonlayer}
	\begin{pgfonlayer}{edgelayer}
		\draw (14.center) to (3);
		\draw (3) to (13.center);
	\end{pgfonlayer}
\end{tikzpicture}
}
\tikzset{x=1em, y=1.5ex}
}

\newcommand{\tscalar}[1]{
\begin{tikzpicture}[x=1em, y=2.1ex]
	\begin{pgfonlayer}{nodelayer}
		\node [style=reg] (0) at (-0.25, 0) {$#1$};
		\node [style=none] (1) at (1.25, 0) {};
		\node [style=none] (2) at (-1.75, 0) {};
	\end{pgfonlayer}
	\begin{pgfonlayer}{edgelayer}
		\draw (2.center) to (0);
		\draw (0) to (1.center);
	\end{pgfonlayer}
\end{tikzpicture}}

\usepackage{wasysym}
\usepackage{tikz-cd}

\title{Diagrammatic Polyhedral Algebra} 

\author{Filippo Bonchi}{University of Pisa, Italy}{}{http://orcid.org/0000-0002-3433-723X}{Supported by the Ministero dell’Università e della Ricerca of Italy under Grant No. 201784YSZ5, PRIN2017 – ASPRA (Analysis of Program Analyses).}

\author{Alessandro Di Giorgio}{University of Pisa, Italy}{}{https://orcid.org/0000-0002-6428-6461}{Supported by the Ministero dell’Università e della Ricerca of Italy under Grant No. 201784YSZ5, PRIN2017 – ASPRA (Analysis of Program Analyses).}

\author{Pawe\l{} Soboci\'{n}ski}{Tallinn University of Technology, Estonia}{}{https://orcid.org/0000-0002-7992-9685}{Supported by the ESF funded Estonian IT Academy research measure (project 2014-2020.4.05.19-0001) and the Estonian Research Council grant PRG1210.}

\authorrunning{F. Bonchi, A. Di Giorgio and P. Soboci\'{n}ski} 

\Copyright{Filippo Bonchi, Alessandro Di Giorgio and Pawe\l{} Soboci\'{n}ski} 



\category{} 

\relatedversion{} 

\acknowledgements{The authors have benefited, at the early stage of this work, of enlightening comments and exciting discussions with Guillaume Boisseau. In particular, Guillaume proposed several simplifications to the axiomatisations, showed us the first proof of the generalised spider and the one for its polar.}

\nolinenumbers 

\EventEditors{John Q. Open and Joan R. Access}
\EventNoEds{2}
\EventLongTitle{42nd Conference on Very Important Topics (CVIT 2016)}
\EventShortTitle{CVIT 2016}
\EventAcronym{CVIT}
\EventYear{2016}
\EventDate{December 24--27, 2016}
\EventLocation{Little Whinging, United Kingdom}
\EventLogo{}
\SeriesVolume{42}
\ArticleNo{23}

\begin{document}

\maketitle

\begin{abstract}
 We extend the theory of Interacting Hopf algebras with an order primitive, and give a
 sound and complete axiomatisation of the prop of polyhedral cones. Next, we axiomatise an affine extension
 and prove soundness and completeness for the prop of polyhedra. 
\end{abstract}

%

\section{Introduction}
Engineers and scientists of different fields often rely on diagrammatic notations to model systems of various sorts but, to perform a rigorous analysis, diagrams usually need to be translated to more traditional mathematical language. Indeed diagrams have the advantage to be quite intuitive, highlight connectivity, distribution and communication topology of systems but they usually have an informal meaning and, even when equipped with a formal semantics, diagrams cannot be easily manipulated like standard mathematical expressions.
\emph{Compositional network theory} \cite{Baez2014,BaezCoya-propsnetworktheory} is a multidisciplinary research program studying diagrams as first class citizens: diagrammatic languages come equipped with a formal semantics, which has the key feature to be compositional; moreover diagrams can be 
manipulated like ordinary symbolic expressions if an appropriate equational theory--ideally characterising semantic equality--can be identified.
This approach has been shown effective in various settings like for instance, digital \cite{Ghica2016} and electrical circuits \cite{baez2015compositional,BaezCoya-propsnetworktheory}, quantum protocols \cite{CoeckeDuncanZX2011,Coecke2012}, linear dynamical systems \cite{BaezErbele-CategoriesInControl,ZanasiThesis}, Petri nets \cite{DBLP:journals/pacmpl/BonchiHPSZ19}, Bayesian networks \cite{JacobsZ18} and query languages \cite{DBLP:journals/corr/abs-2009-06836,GCQ}.

The common technical infrastructure is provided by \emph{string diagrams} \cite{Selinger2009}: arrows of a symmetric monoidal category freely generated by a monoidal signature. Intuitively, the signature is a set of generators and diagrams are simply obtained by composing in series (horizontally) and in parallel (vertically) generators plus some basic wires ($\IdnetT$ and $\symNetT$).  
The following set of generators is common to (most of) the aforementioned systems and, surprisingly enough, (almost) the same algebraic laws hold in the various settings. 
$$\BcounitT \qquad \BcomultT \qquad \WmultT \qquad \WunitT \qquad  
    \BunitT \qquad \BmultT \qquad \WcomultT \qquad \WcounitT  $$
It is convenient to give an intuition of the intended meaning of such generators by relying on the semantics from \cite{ZanasiThesis} that, amongst the aforementioned works, is the most relevant for the present paper: the copier $\BcomultT$ receives one value on the left and emits two copies on the right; the discharger $\BcounitT$ receives one value on the left and nothing on the right; the adder $\WmultT$ receives two values on the left and emits their sum on the right; the zero $\WunitT$ receives nothing on the left and constantly emits $0$ on the right. The behaviour of the remaining four generators is the same but left and right are swapped.
Here values are meant to be rational numbers. To deal with values from an arbitrary fields $\field$, one has to add a generator $\scalarT$ for each $k \in \field$; its intended meaning is the one of an amplifier: the value received on the left is multiplied by $k$ and emitted on the right.

This semantics has two crucial properties: first, it enjoys a sound and complete axiomatisation called the theory of Interacting Hopf Algebras ($\IH{}$); second, it can express exactly \emph{linear relations}, namely relations forming vector spaces over $\field$. In other words, diagrams modulo the laws of $\IH{}$ are in one to one correspondence with linear relations.

\medskip

In this paper we extend $\IH{}$ in order to express exactly relations that are \emph{polyhedra}, rather than mere vector spaces. Indeed, polyhedra allow the modeling of bounded spaces which are ubiquitous in computer science. For instance, in abstract interpretation \cite{DBLP:conf/popl/CousotC77} polyhedra represent bounded sets of possible values of variables; in concurrency theory and linear optimisation one always deals with systems having a bounded amounts of resources. 

To catch a glimpse of our result, consider the flow network \cite{10.5555/137406} in \eqref{eq:flownet}: edges are labeled with a positive real number representing their maximum capacity; the flow enters in the source (the node \texttt{s}) and exits from the sink (the node \texttt{t}).

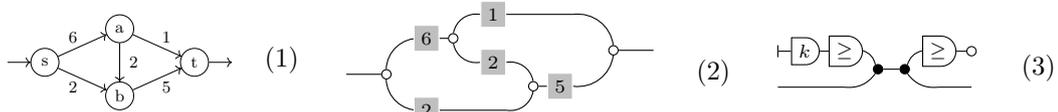
\noindent
\begin{minipage}{0.3\textwidth}
\begin{equation}\label{eq:flownet}
    \scalebox{0.7}{
\tikzset{x=1em, y=2.1ex}
\InputIfFileExists{flownet/example-net.tikz}{}{\input{./tikz/flownet/example-net.tikz}}
\tikzset{x=1em, y=1.5ex}
}
\end{equation}
\end{minipage}
\begin{minipage}{0.4\textwidth}
\begin{equation}\label{eq:flownetd}
    
\tikzset{x=1em, y=2.1ex}
\InputIfFileExists{flownet/example-enc.tikz}{}{\input{./tikz/flownet/example-enc.tikz}}
\tikzset{x=1em, y=1.5ex}

\end{equation}
\end{minipage}
\begin{minipage}{0.3\textwidth}
\begin{equation}\label{eq:flowneedge}
    
\tikzset{x=1em, y=2.1ex}
\InputIfFileExists{flownet/edge-1.tikz}{}{\input{./tikz/flownet/edge-1.tikz}}
\tikzset{x=1em, y=1.5ex}

\end{equation}
\end{minipage}

\noindent
The network in \eqref{eq:flownet} is represented within our diagrammatic language as in \eqref{eq:flownetd} where $
\tikzset{x=1em, y=2.1ex}
\begin{tikzpicture}
	\begin{pgfonlayer}{nodelayer}
		\node [style=none] (0) at (-1.5, 0) {};
		\node [style=edge] (1) at (0, 0) {$k$};
		\node [style=none] (2) at (1.5, 0) {};
	\end{pgfonlayer}
	\begin{pgfonlayer}{edgelayer}
		\draw (0.center) to (1.center);
		\draw (1.center) to (2.center);
	\end{pgfonlayer}
\end{tikzpicture}
}
\tikzset{x=1em, y=1.5ex}
$ is syntactic sugar for the diagram in \eqref{eq:flowneedge}. Here $\greq$ and $\One$ are the two novel generators that we need to add to Interacting Hopf Algebras to express exactly polyhedra: $\greq$ constrains the observation on the left to be greater or equal to the one on the right; $\One$ constantly emits $1$ on the right. Observe that \eqref{eq:flowneedge} forces the values on the left and on the right to be equal and to be in the interval $[0,k]$; the use of $\WmultT$ and $\WcomultT$ for the nodes forces the sum of the flows entering on the left to be equal to the sum of the flows leaving from the right. 

An important property of flow networks is the maximum flow that can enter in the source and arrive to the the sink. The sound and complete axiomatisation that we introduce allows to compute their maximum flow by mean of intuitive graphical manipulations: for instance, the diagram in \eqref{eq:flownetd} can be transformed in $
\tikzset{x=1em, y=2.1ex}
\begin{tikzpicture}
	\begin{pgfonlayer}{nodelayer}
		\node [style=none] (42) at (-1.5, 0) {};
		\node [style=edge] (94) at (0, 0) {$5$};
		\node [style=none] (98) at (1.5, 0) {};
	\end{pgfonlayer}
	\begin{pgfonlayer}{edgelayer}
		\draw (42.center) to (94);
		\draw (94) to (98.center);
	\end{pgfonlayer}
\end{tikzpicture}
}
\tikzset{x=1em, y=1.5ex}
$, meaning that its maximum flow is exactly $5$. We will come back to flow networks at the end of \S\ref{sec:poly} (Example \ref{ex:fnet}).

\medskip

The remainder of the paper is organised as follows. We recall  the basic categorical tools for string diagrams in \S\ref{sec:PROP} and the theory of Interacting Hopf Algebras in \S\ref{sec:IH}. In \S\ref{sec:polyc}, we extend the syntax of Interacting Hopf Algebras
with the generator $\greq$. On the semantic side, this allows to move from linear relations to \emph{polyhedral cones}, for which we give a sound and fully complete axiomatisation in terms of the diagrammatic syntax. The proof of completeness involves a diagrammatic account of Fourier-Motzkin elimination, two normal forms leading to the Weyl-Minkowski theorem, and a simple, inductive account of the notion of polar cone.

The results in \S\ref{sec:polyc} represent our main technical effort. Indeed, to pass from  polyhedral cones to polyhedra in \S\ref{sec:poly}, it is enough to add the generator $\One$, originally introduced in~\cite{BonchiPSZ19} to move from linear to affine relations, and one extra axiom. The proof substantially exploits the homogenization technique to reduce completeness for polyhedra to the just proved completeness for polyhedral cones. 

Finally, in \S\ref{app:state}, we conclude by showing a stateful extension of our diagrammatic calculus. By simply adding a register $\circuitXT$ we obtain a complete axiomatisation for stateful polyhedral processes: these are exactly all transition systems where both states and labels are vectors from some vector spaces and the underlying transition relation forms a polyhedron. Stateful polyhedral processes seem to be a sweet spot in terms of expressivity: on the one hand, they properly generalise signal flow graphs~\cite{mason1953feedback}, on the other, as illustrated in \S\ref{app:state}, they allow us to give a compositional account of continuous Petri nets~\cite{DavidAlla10}.

\section{Props and Symmetric Monoidal Theories}\label{sec:PROP}

The diagrammatic languages studied in network theory, e.g.,~\cite{Coecke2017,Piedeleu2021,Haydon2020,DBLP:journals/pacmpl/BonchiHPSZ19}, can be treated formally using the category theoretic notion of prop~\cite{MacLane1965,Lack2004a} (product and permutation category).
A \emph{prop} is a symmetric strict monoidal (ssm) category with objects natural numbers, where the monoidal product $\oplus$ on objects is addition. Morphisms between props are ssm functors that act as identity on objects.
The usual methodology is to use two props: $\Syn$, the arrows of which are the diagrammatic terms of the language, and $\Sem$, the arrows of which are the intended semantics. A morphism $\dsemO \colon \Syn \rightarrow \Sem$ assigns semantics to terms, with the functoriality of $\dsemO$ guaranteeing compositionality.

The syntactic prop $\Syn$ is usually freely generated from a \emph{monoidal signature} $\Sigma$, namely a set of generators $o\colon n  \to m$ with arity $n\in \N$ and coarity $m \in \N$. Intuitively, the arrows of $\Syn$ are diagrams wired up from the generators. A way of giving a concrete description is via $\Sigma$-terms. The set of \emph{$\Sigma$-terms} is obtained by composing generators in $\Sigma$, the identities $id_0\colon 0 \to 0$, $id_1\colon 1 \to 1$ and the symmetry $\sigma_{1,1} \colon 2 \to 2$ with $;$ and $\oplus$. This is a purely formal process: given $\Sigma$-terms $t\colon k \to l$, $u \colon l \to m$, $v \colon m \to n$, one constructs $\Sigma$-terms $t ; u \colon k \to m$ and $t \oplus v \colon k + n \to l + n$. Now, the \emph{prop freely generated by a signature} $\Sigma$, hereafter denoted by $\PSigma$, has as its arrows $n \to m$ the set of $\Sigma$-terms $n \to m$  modulo the laws of ssm categories.

There is a well-known, natural graphical representation for arrows of a freely generated prop as string diagrams, which we now sketch. 
A $\Sigma$-term $n \to m$ is pictured as a box with $n$ ordered wires on the left and $m$ on the right. Composition via $;$ and $\oplus$ are rendered graphically by horizontal and vertical juxtaposition of boxes, respectively.
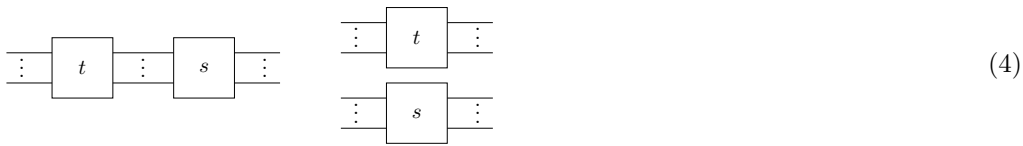
\begin{equation}\label{eq:horver}
        \scalebox{0.8}{
    \begin{tikzpicture}
        \begin{pgfonlayer}{nodelayer}
            \node [style=square] (0) at (-3.75, 0) {$t$};
            \node [style=square] (1) at (-1.75, 0) {$s$};
            \node [style=none] (2) at (-2, 0.25) {};
            \node [style=none] (3) at (-2, -0.25) {};
            \node [style=none] (4) at (-0.5, 0.25) {};
            \node [style=none] (5) at (-0.5, -0.25) {};
            \node [style=none] (6) at (-5, -0.25) {};
            \node [style=none] (7) at (-5, 0.25) {};
            \node [style=none] (8) at (-4, 0.25) {};
            \node [style=none] (9) at (-4, -0.25) {};
            \node [style=none] (10) at (-1.5, -0.25) {};
            \node [style=none] (11) at (-1.5, 0.25) {};
            \node [style=none] (12) at (-3.5, -0.25) {};
            \node [style=none] (13) at (-3.5, 0.25) {};
            \node [style=none] (15) at (-4.75, 0) {\myvdots};
            \node [style=none] (17) at (-2.75, 0) {\myvdots};
            \node [style=none] (19) at (-0.75, 0) {\myvdots};
            \node [style=square] (20) at (1.75, 0.5) {$t$};
            \node [style=none] (21) at (3, 0.75) {};
            \node [style=none] (22) at (3, 0.25) {};
            \node [style=none] (23) at (0.5, 0.25) {};
            \node [style=none] (24) at (0.5, 0.75) {};
            \node [style=none] (25) at (1.5, 0.75) {};
            \node [style=none] (26) at (1.5, 0.25) {};
            \node [style=none] (27) at (2, 0.25) {};
            \node [style=none] (28) at (2, 0.75) {};
            \node [style=none] (29) at (0.75, 0.5) {\myvdots};
            \node [style=none] (30) at (2.75, 0.5) {\myvdots};
            \node [style=square] (31) at (1.75, -0.75) {$s$};
            \node [style=none] (32) at (3, -0.5) {};
            \node [style=none] (33) at (3, -1) {};
            \node [style=none] (34) at (0.5, -1) {};
            \node [style=none] (35) at (0.5, -0.5) {};
            \node [style=none] (36) at (1.5, -0.5) {};
            \node [style=none] (37) at (1.5, -1) {};
            \node [style=none] (38) at (2, -1) {};
            \node [style=none] (39) at (2, -0.5) {};
            \node [style=none] (40) at (0.75, -0.75) {\myvdots};
            \node [style=none] (41) at (2.75, -0.75) {\myvdots};
        \end{pgfonlayer}
        \begin{pgfonlayer}{edgelayer}
            \draw (7.center) to (8.center);
            \draw (6.center) to (9.center);
            \draw (10.center) to (5.center);
            \draw (11.center) to (4.center);
            \draw (13.center) to (2.center);
            \draw (12.center) to (3.center);
            \draw (24.center) to (25.center);
            \draw (23.center) to (26.center);
            \draw (28.center) to (21.center);
            \draw (27.center) to (22.center);
            \draw (35.center) to (36.center);
            \draw (34.center) to (37.center);
            \draw (39.center) to (32.center);
            \draw (38.center) to (33.center);
        \end{pgfonlayer}
    \end{tikzpicture}}
\end{equation}
Moreover $id_1 \colon 1 \to 1$ is pictured as $\IdnetT$, the symmetry $\sigma_{1,1} \colon 1 + 1 \to 1+1$ as $\symNetT$, and the unit object for $\oplus$, that is, $id_0 \colon 0 \to 0$ as the empty diagram $\ZeronetT$.
Arbitrary identities $id_n$ and symmetries $\sigma_{n,m}$ are generated according to
~\eqref{eq:horver} and  drawn as \begin{tikzpicture}
	\begin{pgfonlayer}{nodelayer}
		\node [style=none] (10) at (-0.25, 0) {};
		\node [style=none] (11) at (0, 0.2) {$n$};
		\node [style=none] (14) at (0.25, 0) {};
	\end{pgfonlayer}
	\begin{pgfonlayer}{edgelayer}
		\draw (10.center) to (14.center);
	\end{pgfonlayer}
\end{tikzpicture}
 and \scalebox{0.8}{
\begin{tikzpicture}
	\begin{pgfonlayer}{nodelayer}
		\node [style=none] (10) at (-0.5, -0.25) {};
		\node [style=none] (11) at (-0.4, 0.4) {$n$};
		\node [style=none] (12) at (0, 0) {};
		\node [style=none] (13) at (0.5, -0.25) {};
		\node [style=none] (14) at (0.5, 0.25) {};
		\node [style=none] (15) at (-0.5, 0.25) {};
		\node [style=none] (20) at (0.4, -0.1) {$n$};
		\node [style=none] (21) at (0.4, 0.4) {$m$};
		\node [style=none] (22) at (-0.4, -0.1) {$m$};
	\end{pgfonlayer}
	\begin{pgfonlayer}{edgelayer}
		\draw [bend left] (15.center) to (12.center);
		\draw [bend right] (12.center) to (13.center);
		\draw [bend left] (12.center) to (14.center);
		\draw [bend left] (12.center) to (10.center);
	\end{pgfonlayer}
\end{tikzpicture}
}, respectively.

\medskip
Given a diagrammatic language $\Syn$ and a morphism $\dsemO \colon \Syn \to \Sem$,
a useful task is to identify a sound and (ideally) complete set of characterising equations $E$:
$\dsem{c}=\dsem{d}$ iff $c$ and $d$ are equal in $\stackrel{E}{=}$, the smallest congruence (w.r.t. $;$ and $\oplus$) containing $E$. Formally, the set $E$ consists of pairs $(t, t' \colon n \to m)$ of $\Sigma$-terms with the same arity and coarity. Then $\Sigma$ together with $E$ form a \emph{symmetric monoidal theory} (smt), providing a calculus of diagrammatic reasoning. Any smt $(\Sigma, E)$ yields a prop $\PSigmaE$, obtained by quotienting the $\PSigma$ by $\stackrel{E}{=}$.

Another issue is expressivity: one would like to characterise the image of $\Syn$ through $\dsemO$, namely a subprop $\IM$ of $\Sem$ consisting of exactly those arrows $d$ of $\Sem$ for which there exists some $c$ in $\Syn$ such that $\dsem{c}=d$. When this is possible and a sound and complete axiomatization is available, the semantics map $\dsemO$ factors as follows:
\[\xymatrix{{\Syn = \PSigma} \ar@{->>}[r]^q \ar@(ur,ul)[rrr]|{\dsemO} & \PSigmaE  \ar[r]^{\cong} & \IM \;\; \ar@{>->}[r]^{\iota } &\Sem }. \]
The morphism $q$ quotients $\PSigma$ by $\stackrel{E}{=}$, $\iota$ is the inclusion of $\IM$ in $\Sem$ and $\cong$ is an iso between $\PSigmaE$ and $\IM$. In this case we say that $(\Sigma,E)$ is the (symmetric monoidal) theory of $\IM$.

Let $\mathsf k$ be an ordered field. In this paper $\Sem$ is fixed to be the following prop.
\begin{definition}
    \label{def:rel}
    $\RelX{\mathsf k}$ is the prop where arrows $n \rightarrow m$ are relations $R\subseteq \mathsf{k}^n \times \mathsf{k}^m$.
    \begin{itemize}
        \item Composition is relational: given $R \colon n \rightarrow m$ and $S \colon m \rightarrow o$,
        \[R\mathrel{;}S=\{\,(u, v) \in \mathsf{k}^n \times \mathsf{k}^o \;\mid\; \exists w \in \mathsf{k}^m .\; (u, w) \in S \wedge (w, v) \in R \,\} \]
        \item The monoidal product is cartesian product: given $R \colon n \rightarrow m$ and $S \colon o \rightarrow p$,
        \[ R \oplus S = \{\,(\begin{pmatrix} u_1 \\ u_2 \end{pmatrix}, \begin{pmatrix} v_1 \\ v_2 \end{pmatrix}) \in \mathsf{k}^{n+o} \times \mathsf{k}^{m+p} \;\mid\; (u_1, v_1) \in R \wedge (v_1, v_2) \in S \,\} \]
        \item The symmetries $\sigma_{n,m} \colon n + m \rightarrow m + n$ are the relations
        $ \{\, (\begin{pmatrix} u \\ v \end{pmatrix}, \begin{pmatrix} v \\ u \end{pmatrix}) \;\mid\; u \in \mathsf{k}^n, v \in \mathsf{k}^m \,\} $
    \end{itemize}
\end{definition}
For $\IM$, we will consider the following three props. 
\begin{definition} \label{def:pc}
We define three sub-props of $\RelX{\mathsf k}$. Arrows $n\to m$
\begin{itemize}
    \item in $\LinRel{\field}$ are vector spaces $\{ (x,y) \in \mathsf{k}^{n}\times \mathsf{k}^{m} \mid A\begin{pmatrix*} x \\ y \end{pmatrix*} = 0 \}$ for some matrix $A$;
    \item in $\propCat{PC}_{\mathsf k}$ are polyhedral cones $\{ (x,y) \in \mathsf{k}^{n}\times \mathsf{k}^{m} \mid A\begin{pmatrix*} x \\ y \end{pmatrix*} \geq 0 \}$ for some matrix $A$;
\item in $\propCat{P}_{\mathsf k}$ are polyhedra $\{ (x,y) \in \field^n \times \field^m \mid A\begin{pmatrix*} x \\ y \end{pmatrix*} + b \geq 0 \}$ for some  matrix $A$ and $b \in \field^p$.
\end{itemize}
\noindent Identities, permutations, composition and monoidal product are defined as in $\RelX{\mathsf k}$.
\end{definition}
\begin{remark}
In Definition \ref{def:pc}, $A$ is a matrix with $n+m$ columns and $p$ rows, for some $p\in \N$. Observe that the matrix $A$ gives rise also to arrows $n' \to m'$ with $n',m'$ different from $n,m$ but, such that $n'+m' = n+m$. This is justified by the isomorphism of $\field^n \times \field^m$ and $\field^{n'}\times \field^{m'}$. Note therefore that the left and the right boundaries should not be confused with inputs and outputs. This is a common feature in diagrammatic approaches relying on a notion of relational composition which is unbiased.
\end{remark}

Showing that the above are well-defined---e.g. that the composition of polyhedral cones is a polyhedral cone---requires some well-known results, which are given in Appendix~\ref{app:standardproperties}.
%
In \S\ref{sec:IH}, we recall the theory of $\LinRel{\field}$, in \S\ref{sec:polyc} we identify the theory of $\propCat{PC}_{\mathsf k}$ and, in \S\ref{sec:poly}, that of  $\propCat{P}_{\mathsf k}$.

\subsection{Ordered Props and Symmetric Monoidal Inequality Theories}
As relations $R,S\colon n \to m$ in $\RelX{\mathsf k}$ carry the partial order of inclusion $\subseteq$, it is useful to be able to state when $\dsem{c} \subseteq \dsem{d}$ for some $c,d$ in $\Syn$ (see e.g.~\cite{BonchiHPS17} for motivating examples). In order to consider such inclusions, it is convenient to look at $\RelX{\mathsf k}$ as an ordered prop.

\begin{definition}
An \emph{ordered prop} is a prop enriched over the category of posets: a symmetric strict monoidal 2-category with objects the natural numbers, monoidal product on objects given by addition, where each set of arrows $n \to m$ is a poset, with composition and monoidal product monotonic.
Similarly, a \emph{pre-ordered prop} is a prop enriched over the category of pre-orders.
A morphism of (pre-)ordered props is an identity-on-objects  symmetric strict 2-functor.
\end{definition}
Just as SMTs yield props, \emph{Symmetric Monoidal Inequalities Theories}~\cite{BonchiHPS17} (SMITs) give rise to ordered props. A SMIT is a pair $(\Sigma, I)$ where $\Sigma$ is a signature and $I$ is a set of inequations: as for equations, the underlying data  is a pair  $(t, t' \colon n \to m)$ of $\Sigma$-terms with the same arity and coarity. Unlike equations, however, we understand this data as directed: $t\leq t'$.

To obtain the free ordered prop from an SMIT, first, we construct the free pre-ordered prop: arrows are $\Sigma$-terms. The homset orders, hereafter denoted by $\mysubeq{I}$, are determined by closing $I$ by reflexivity, transitivity, $\poi$ and $\oplus$: this is the smallest precongruence (w.r.t. $\poi$ and $\oplus$) containing $I$. Then, we obtain the free ordered prop by quotienting the free pre-ordered prop by the equivalence induced by $\mysubeq{I}$, i.e.\ quotienting wrt anti-symmetry.

Any prop can be regarded as an ordered prop with the discrete ordering. Moreover
any SMT $(\Sigma, E)$ gives rise to a canonical SMIT $(\Sigma, I)$ where each equation is replaced by two inequalities $I=E \cup E^{op}$ in the obvious way.  In the remainder of this paper, we always consider SMITs but, for the sake of readability, we generically refer to them just as \emph{theories}. Such theories consist of both inequalities and equations, that are generically called \emph{axioms}. Similarly, all props considered in the paper, their morphisms and isomorphisms are ordered.

\section{The theory of Linear relations}\label{sec:IH}
In this section, we recall from~\cite{Bonchi2014b,BaezErbele-CategoriesInControl,ZanasiThesis} the theory of Interacting Hopf algebras.
The signature consists of the following set of generators, where $k$ ranges over a fixed field $\mathsf k$.
\begin{align}
    &\BcounitT \bnfSep \BcomultT \bnfSep \scalarT
    \bnfSep \WmultT \bnfSep \WunitT \bnfSep  \label{eq:SFcalculusSyntax1} \\
    &\BunitT \bnfSep \BmultT \bnfSep \scalaropT
    \bnfSep  \WcomultT \bnfSep \WcounitT  \label{eq:SFcalculusSyntax2}
\end{align}
For each generator, its arity and coarity are given by the number of dangling wires on the left and, respectively, on the right. For instance $\BcounitT$ has arity $1$ and coarity $0$.
We call $\CD$ the prop freely generated by this signature and we refer to its arrows as circuits.
We use $\CD$ as the \emph{syntax} of our starting diagrammatic language. The semantics is given as the prop morphism ${\dsemO \colon \CD \rightarrow \RelX{\mathsf k}}$ defined for the generators in~\eqref{eq:SFcalculusSyntax1} as
%
\begin{equation}\label{eq:semIH}
\begin{array}{lll}
\dsem{\Bcomult} \ = \ \{ (x, \begin{pmatrix} x\\ x \end{pmatrix}) \mid x \in \mathsf{k} \}
&
\dsem{\Wmult} \ = \ \{ (\begin{pmatrix} x\\ y \end{pmatrix}, x+y) \mid x, y \in \mathsf{k} \} \\
\dsem{\Bcounit} \ = \ \{ (x, \nullvec) \mid x \in \mathsf{k} \}
&
\dsem{\Wunit}  \ = \  \{ (\nullvec, 0) \} \quad \dsem{\tscalar{k}}  \ = \  \{ (x, k\cdot x) \mid x \in \mathsf{k} \}
\end{array}
\end{equation}
%
and, symmetrically, for the generators in~\eqref{eq:SFcalculusSyntax2}. For instance, $\dsem{\scalaropT} = \{ (k\cdot x,  x) \mid x \in \mathsf{k} \}$.
The semantics of the identities, symmetries and compositions is given by the \emph{functoriality} of $\dsemO$, e.g., $\dsem{c \scolon d} = \dsem{c} \scolon \dsem{d}$. 
Above we used $\nullvec$ for the unique element of the vector space $\field^0$.

We call $\FC$ the prop freely generated from the generators in~\eqref{eq:SFcalculusSyntax1}
and $\FCop$ the one freely generated from~\eqref{eq:SFcalculusSyntax2}. The semantics of circuits in $\FC$ can be thought of as functions taking inputs on left ports and giving output on the right ports, with the intuition for the generators as given in the Introduction.
Symmetrically, the semantics of circuits in $\FCop$ are functions with inputs on the right ports and outputs on the left. The semantics of an arbitrary circuit in $\CD$ is, in general, a relation.

\begin{example}\label{ex:cc}
Two circuits will play a special role in our exposition: $\cupd$ and $\capd$. Using the definition of $\dsemO$, it is immediate to see that their semantics forces the two ports on the right (resp. left) to carry the same value.  \[\dsem{\cupd}=\{ (\nullvec,\begin{pmatrix} x\\ x \end{pmatrix}) \mid x \in \mathsf{k} \} \qquad \dsem{\capd}=\{ (\begin{pmatrix} x\\ x \end{pmatrix},\nullvec) \mid x \in \mathsf{k} \}\]
Using these diagrams (along with $\IdnetT$ and $\symNetT$) one defines for each $n\in \N$, $\cupnd \colon 0 \to n+n$ and $\capnd \colon n+n\to 0$ with semantics
$\{ (\begin{pmatrix} x\\ x \end{pmatrix},\nullvec) \mid x \in \mathsf{k}^n \}$ and $\{ (\nullvec, \begin{pmatrix} x\\ x \end{pmatrix}) \mid x \in \mathsf{k}^n \}$. These circuits give rise, modulo the axioms that we will illustrate later, to a self-dual compact closed structure. See~\cite[Sec.5.1]{BonchiSZ17} for full details. As for identities and symmetries, also for $\cupnd$ and $\capnd$ we will sometimes omit $n$ for readability. Given an arbitrary circuit $c\colon n\to m$, its \emph{opposite} circuit $c^{\op}\colon m \to n$ is defined as illustrated below. It is easy to see that $c^{\op}$ denotes the opposite relation of $\dsem{c}$, i.e., $\dsem{c^{\op}}=\{(y,x)\in \field^m \times \field^n \mid (x,y)\in \dsem{c}\}$.
\[ \left( \cbox \right)^{\op} \coloneqq \cboxop\]
\end{example}
As for $\cupnd$ above, one can define the $n$-version of each of the generators in
\eqref{eq:SFcalculusSyntax1} and~\eqref{eq:SFcalculusSyntax2} (as well as generators~\eqref{eq:>} and~\eqref{eq:perp} that we shall introduce later). For instance $\dsem{
\tikzset{x=1em, y=2.1ex}
\InputIfFileExists{add-n.tikz}{}{\input{./tikz/add-n.tikz}}
\tikzset{x=1em, y=1.5ex}
} = \{ (\begin{pmatrix} x\\ y \end{pmatrix}, x+y) \mid x, y \in \mathsf{k}^n\}$. When clear from the context, we will omit the $n$.

A sound and complete axiomatisation for semantic equality was developed in~\cite{Bonchi2014b,BaezErbele-CategoriesInControl,ZanasiThesis}, and in~\cite{BonchiHPS17} for inclusion. 
The above signature together with the axioms, recalled in Figure~\ref{fig:ih}, form the theory of Interacting Hopf Algebras.
The resulting prop is denoted by $\IH{\field}$.

\begin{remark}
Thanks to the compact closed structure, each of the axioms and laws that we prove in the text can be read both as $c \myeq{\IH{}} d$ and $c^{\op} \myeq{\IH{}} d^{\op}$. For example, by $\bullet\text{--}coas$ we also know that $
\tikzset{x=1em, y=2.1ex}
\InputIfFileExists{co-copy-associative.tikz}{}{\input{./tikz/co-copy-associative.tikz}}
\tikzset{x=1em, y=1.5ex}
 \myeq{\IH{}} 
\tikzset{x=1em, y=2.1ex}
\InputIfFileExists{co-copy-associative-1.tikz}{}{\input{./tikz/co-copy-associative-1.tikz}}
\tikzset{x=1em, y=1.5ex}
$.
\end{remark}

\begin{theorem} \label{co:complih}
    For all circuits $c, d$ in $\CD$, $\dsem{c} \subseteq \dsem{d}$ if and only if $c \stackrel{\IH{}}{\subseteq} d$.
\end{theorem}

We now come to expressivity: which relations in $\RelX{\mathsf k}$ are expressed by $\CD$? The answer is that $\CD$ captures exactly $\LinRel{\field}$ (see Definition~\ref{def:pc}). 
\begin{theorem}
    \label{pr:isoihlrel}
    $\IH{\field} \cong \LinRel{\field}$.
\end{theorem}
The above result means that $\IH{\field}$ is the theory of linear relations.
It is convenient to recall from~\cite{Lafont2003} a useful fact: circuits in $\FC$ express exactly  $\field$-matrices, as illustrated below:
\begin{example}\label{ex:matrix}
Consider the circuit $c \colon 3 \rightarrow 4$ below and its representation as a $4 \times 3$ matrix. Note that $A_{ij} = k$ whenever $k$ is the scalar encountered on the path from the $i$th port to the $j$th port. If there is no path, then $A_{ij}=0$. It is easy to check that $\dsem{c}= \{(x,y) \in \field^3\times \field^4 \mid y=Ax\}$.

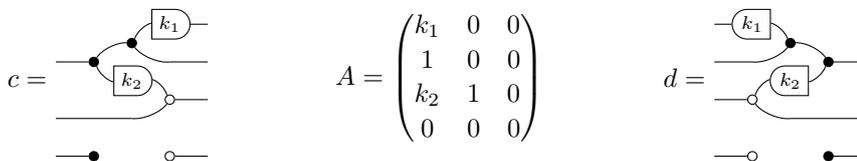
\begin{minipage}[t]{0.3\textwidth}
            \begin{equation*}
                c = \begin{tikzpicture}
                    \begin{pgfonlayer}{nodelayer}
                        \node [style=black] (0) at (-0.5, 0.25) {};
                        \node [style=black] (1) at (0, 0.5) {};
                        \node [style=reg] (2) at (0, 0) {$k_2$};
                        \node [style=none] (3) at (-1, 0.25) {};
                        \node [style=white] (4) at (0.5, -0.25) {};
                        \node [style=none] (5) at (0, -0.5) {};
                        \node [style=none] (6) at (-1, -0.5) {};
                        \node [style=none] (7) at (0.5, 0.25) {};
                        \node [style=reg] (8) at (0.5, 0.75) {$k_1$};
                        \node [style=none] (9) at (1, 0.75) {};
                        \node [style=none] (10) at (1, 0.25) {};
                        \node [style=none] (11) at (1, -0.25) {};
                        \node [style=none] (12) at (-1, -1) {};
                        \node [style=black] (13) at (-0.5, -1) {};
                        \node [style=none] (14) at (1, -1) {};
                        \node [style=white] (15) at (0.5, -1) {};
                    \end{pgfonlayer}
                    \begin{pgfonlayer}{edgelayer}
                        \draw (3.center) to (0);
                        \draw [bend left, looseness=0.75] (0) to (1);
                        \draw [bend right, looseness=0.75] (0) to (2);
                        \draw [bend right, looseness=0.75] (5.center) to (4);
                        \draw [bend right, looseness=0.75] (4) to (2);
                        \draw (6.center) to (5.center);
                        \draw [bend left, looseness=0.75] (1) to (8);
                        \draw [bend right, looseness=0.75] (1) to (7.center);
                        \draw (8) to (9.center);
                        \draw (7.center) to (10.center);
                        \draw (4) to (11.center);
                        \draw (12.center) to (13);
                        \draw (15) to (14.center);
                    \end{pgfonlayer}
                \end{tikzpicture}  
            \end{equation*}
        \end{minipage}
        \begin{minipage}[t]{0.3\textwidth}
            \begin{equation*}
                A = \begin{pmatrix}
                    k_1 & 0 & 0 \\
                    1 & 0 & 0 \\
                    k_2 & 1 & 0 \\
                    0 & 0 & 0
                \end{pmatrix}
            \end{equation*}
        \end{minipage}
        \begin{minipage}[t]{0.3\textwidth}
            \begin{equation*}
               d = \begin{tikzpicture}
	\begin{pgfonlayer}{nodelayer}
		\node [style=black] (0) at (0.5, 0.25) {};
		\node [style=black] (1) at (0, 0.5) {};
		\node [style=coreg] (2) at (0, 0) {$k_2$};
		\node [style=none] (3) at (1, 0.25) {};
		\node [style=white] (4) at (-0.5, -0.25) {};
		\node [style=none] (5) at (0, -0.5) {};
		\node [style=none] (6) at (1, -0.5) {};
		\node [style=none] (7) at (-0.5, 0.25) {};
		\node [style=coreg] (8) at (-0.5, 0.75) {$k_1$};
		\node [style=none] (9) at (-1, 0.75) {};
		\node [style=none] (10) at (-1, 0.25) {};
		\node [style=none] (11) at (-1, -0.25) {};
		\node [style=none] (12) at (1, -1) {};
		\node [style=black] (13) at (0.5, -1) {};
		\node [style=none] (14) at (-1, -1) {};
		\node [style=white] (15) at (-0.5, -1) {};
	\end{pgfonlayer}
	\begin{pgfonlayer}{edgelayer}
		\draw (3.center) to (0);
		\draw [bend right, looseness=0.75] (0) to (1);
		\draw [bend left, looseness=0.75] (0) to (2);
		\draw [bend left, looseness=0.75] (5.center) to (4);
		\draw [bend left, looseness=0.75] (4) to (2);
		\draw (6.center) to (5.center);
		\draw [bend right, looseness=0.75] (1) to (8);
		\draw [bend left, looseness=0.75] (1) to (7.center);
		\draw (8) to (9.center);
		\draw (7.center) to (10.center);
		\draw (4) to (11.center);
		\draw (12.center) to (13);
		\draw (15) to (14.center);
	\end{pgfonlayer}
\end{tikzpicture}
            \end{equation*}
        \end{minipage}

Dually, circuits in $\FCop$ are ``reversed'' matrices: inputs on the right and outputs on the left. For instance $d\colon 4 \to 3$ again encodes $A$, but its semantics is $\dsem{d}= \{(y,x) \in \field^4\times \field^3 \mid y=Ax\}$.
\end{example}

\begin{figure*}
\begin{alignat*}{3}
&\;
\tikzset{x=1em, y=2.1ex}
\InputIfFileExists{add-associative.tikz}{}{\input{./tikz/add-associative.tikz}}
\tikzset{x=1em, y=1.5ex}
 \myeq{\circ-as}  
\tikzset{x=1em, y=2.1ex}
\InputIfFileExists{add-associative-1.tikz}{}{\input{./tikz/add-associative-1.tikz}}
\tikzset{x=1em, y=1.5ex}

\qquad \quad &&\;
\tikzset{x=1em, y=2.1ex}
\InputIfFileExists{add-commutative.tikz}{}{\input{./tikz/add-commutative.tikz}}
\tikzset{x=1em, y=1.5ex}
\myeq{\circ-co}

\tikzset{x=1em, y=2.1ex}
\InputIfFileExists{add.tikz}{}{\input{./tikz/add.tikz}}
\tikzset{x=1em, y=1.5ex}
\qquad \quad &&\;
\tikzset{x=1em, y=2.1ex}
\InputIfFileExists{add-unital-left.tikz}{}{\input{./tikz/add-unital-left.tikz}}
\tikzset{x=1em, y=1.5ex}
\myeq{\circ-unl}
\tikzset{x=1em, y=2.1ex}
\InputIfFileExists{id.tikz}{}{\input{./tikz/id.tikz}}
\tikzset{x=1em, y=1.5ex}

\\
&
\tikzset{x=1em, y=2.1ex}
\InputIfFileExists{copy-associative.tikz}{}{\input{./tikz/copy-associative.tikz}}
\tikzset{x=1em, y=1.5ex}
\myeq{\bullet-coas} 
\tikzset{x=1em, y=2.1ex}
\InputIfFileExists{copy-associative-1.tikz}{}{\input{./tikz/copy-associative-1.tikz}}
\tikzset{x=1em, y=1.5ex}
\qquad \quad &&
\tikzset{x=1em, y=2.1ex}
\InputIfFileExists{copy-commutative.tikz}{}{\input{./tikz/copy-commutative.tikz}}
\tikzset{x=1em, y=1.5ex}
\myeq{\bullet-coco} 
\tikzset{x=1em, y=2.1ex}
\InputIfFileExists{copy.tikz}{}{\input{./tikz/copy.tikz}}
\tikzset{x=1em, y=1.5ex}
\qquad \quad &&
\tikzset{x=1em, y=2.1ex}
\InputIfFileExists{copy-unital-left.tikz}{}{\input{./tikz/copy-unital-left.tikz}}
\tikzset{x=1em, y=1.5ex}
\myeq{\bullet-counl}
\tikzset{x=1em, y=2.1ex}
\InputIfFileExists{id.tikz}{}{\input{./tikz/id.tikz}}
\tikzset{x=1em, y=1.5ex}

\end{alignat*}
\hrule
\vspace{0.7pt}
\begin{equation*}

\tikzset{x=1em, y=2.1ex}
\InputIfFileExists{add-copy-bimonoid.tikz}{}{\input{./tikz/add-copy-bimonoid.tikz}}
\tikzset{x=1em, y=1.5ex}
\myeq{\circ\bullet-bi}
\tikzset{x=1em, y=2.1ex}
\InputIfFileExists{add-copy-bimonoid-1.tikz}{}{\input{./tikz/add-copy-bimonoid-1.tikz}}
\tikzset{x=1em, y=1.5ex}
 \qquad 
\tikzset{x=1em, y=2.1ex}
\InputIfFileExists{add-copy-bimonoid-unit.tikz}{}{\input{./tikz/add-copy-bimonoid-unit.tikz}}
\tikzset{x=1em, y=1.5ex}
\myeq{\circ\bullet-biun} 
\tikzset{x=1em, y=2.1ex}
\InputIfFileExists{add-bimonoid-unit-1.tikz}{}{\input{./tikz/add-bimonoid-unit-1.tikz}}
\tikzset{x=1em, y=1.5ex}
 \qquad 
\tikzset{x=1em, y=2.1ex}
\InputIfFileExists{add-copy-bimonoid-counit.tikz}{}{\input{./tikz/add-copy-bimonoid-counit.tikz}}
\tikzset{x=1em, y=1.5ex}
\myeq{\bullet\circ-biun} 
\tikzset{x=1em, y=2.1ex}
\InputIfFileExists{add-copy-bimonoid-counit-1.tikz}{}{\input{./tikz/add-copy-bimonoid-counit-1.tikz}}
\tikzset{x=1em, y=1.5ex}
\qquad
\tikzset{x=1em, y=2.1ex}
\InputIfFileExists{bone-white-black.tikz}{}{\input{./tikz/bone-white-black.tikz}}
\tikzset{x=1em, y=1.5ex}
\myeq{\circ\bullet-bo}
\tikzset{x=1em, y=2.1ex}
\InputIfFileExists{empty-diag.tikz}{}{\input{./tikz/empty-diag.tikz}}
\tikzset{x=1em, y=1.5ex}

\end{equation*}
\hrule
\vspace{2.7pt}
\begin{equation*}
   
\tikzset{x=1em, y=2.1ex}
\InputIfFileExists{reals-add.tikz}{}{\input{./tikz/reals-add.tikz}}
\tikzset{x=1em, y=1.5ex}
\;\myeq{add}\;
\tikzset{x=1em, y=2.1ex}
\InputIfFileExists{reals-add-1.tikz}{}{\input{./tikz/reals-add-1.tikz}}
\tikzset{x=1em, y=1.5ex}
 \qquad 
\tikzset{x=1em, y=2.1ex}
\InputIfFileExists{zero.tikz}{}{\input{./tikz/zero.tikz}}
\tikzset{x=1em, y=1.5ex}
\;\myeq{zer}\;
\tikzset{x=1em, y=2.1ex}
\InputIfFileExists{reals-zero.tikz}{}{\input{./tikz/reals-zero.tikz}}
\tikzset{x=1em, y=1.5ex}
 \qquad
   
\tikzset{x=1em, y=2.1ex}
\InputIfFileExists{reals-copy.tikz}{}{\input{./tikz/reals-copy.tikz}}
\tikzset{x=1em, y=1.5ex}
\;\myeq{dup}\; 
\tikzset{x=1em, y=2.1ex}
\InputIfFileExists{reals-copy-1.tikz}{}{\input{./tikz/reals-copy-1.tikz}}
\tikzset{x=1em, y=1.5ex}
 \qquad 
\tikzset{x=1em, y=2.1ex}
\InputIfFileExists{reals-delete.tikz}{}{\input{./tikz/reals-delete.tikz}}
\tikzset{x=1em, y=1.5ex}
\;\myeq{del}\;
\tikzset{x=1em, y=2.1ex}
\InputIfFileExists{delete.tikz}{}{\input{./tikz/delete.tikz}}
\tikzset{x=1em, y=1.5ex}

\end{equation*}
\begin{equation*}
   
\tikzset{x=1em, y=2.1ex}
\InputIfFileExists{reals-multiplication.tikz}{}{\input{./tikz/reals-multiplication.tikz}}
\tikzset{x=1em, y=1.5ex}
\;\myeq{\times}\;
\tikzset{x=1em, y=2.1ex}
\InputIfFileExists{reals-multiplication-1.tikz}{}{\input{./tikz/reals-multiplication-1.tikz}}
\tikzset{x=1em, y=1.5ex}
 \qquad    
\tikzset{x=1em, y=2.1ex}
\InputIfFileExists{reals-sum.tikz}{}{\input{./tikz/reals-sum.tikz}}
\tikzset{x=1em, y=1.5ex}
\;\myeq{+}\;
\tikzset{x=1em, y=2.1ex}
\InputIfFileExists{reals-sum-1.tikz}{}{\input{./tikz/reals-sum-1.tikz}}
\tikzset{x=1em, y=1.5ex}
\qquad   
\tikzset{x=1em, y=2.1ex}
\InputIfFileExists{reals-scalar-zero.tikz}{}{\input{./tikz/reals-scalar-zero.tikz}}
\tikzset{x=1em, y=1.5ex}
\;\myeq{0}\;
\tikzset{x=1em, y=2.1ex}
\InputIfFileExists{reals-scalar-zero-1.tikz}{}{\input{./tikz/reals-scalar-zero-1.tikz}}
\tikzset{x=1em, y=1.5ex}

\end{equation*}
 \hrule
 \vspace{2.7pt}
 \begin{equation*}
     
\tikzset{x=1em, y=2.1ex}
\InputIfFileExists{scalar-division.tikz}{}{\input{./tikz/scalar-division.tikz}}
\tikzset{x=1em, y=1.5ex}
\;\myeq{r-inv}\; 
\tikzset{x=1em, y=2.1ex}
\InputIfFileExists{id.tikz}{}{\input{./tikz/id.tikz}}
\tikzset{x=1em, y=1.5ex}
 \qquad\quad 
\tikzset{x=1em, y=2.1ex}
\InputIfFileExists{id.tikz}{}{\input{./tikz/id.tikz}}
\tikzset{x=1em, y=1.5ex}
\;\myeq{r-coinv}\;
\tikzset{x=1em, y=2.1ex}
\InputIfFileExists{scalar-co-division.tikz}{}{\input{./tikz/scalar-co-division.tikz}}
\tikzset{x=1em, y=1.5ex}
\quad \text{ for } k\neq 0, k\in \field
 \end{equation*}
 \hrule
 \vspace{0.7pt}
 \begin{equation*}
 
\tikzset{x=1em, y=2.1ex}
\InputIfFileExists{copy-Frobenius-left.tikz}{}{\input{./tikz/copy-Frobenius-left.tikz}}
\tikzset{x=1em, y=1.5ex}
\myeq{\bullet-fr1} 
\tikzset{x=1em, y=2.1ex}
\InputIfFileExists{copy-Frobenius.tikz}{}{\input{./tikz/copy-Frobenius.tikz}}
\tikzset{x=1em, y=1.5ex}
\myeq{\bullet-fr2} 
\tikzset{x=1em, y=2.1ex}
\InputIfFileExists{copy-Frobenius-right.tikz}{}{\input{./tikz/copy-Frobenius-right.tikz}}
\tikzset{x=1em, y=1.5ex}
 \qquad 
\tikzset{x=1em, y=2.1ex}
\InputIfFileExists{copy-special.tikz}{}{\input{./tikz/copy-special.tikz}}
\tikzset{x=1em, y=1.5ex}
\myeq{\bullet-sp}
\tikzset{x=1em, y=2.1ex}
\InputIfFileExists{id.tikz}{}{\input{./tikz/id.tikz}}
\tikzset{x=1em, y=1.5ex}
\qquad 
\tikzset{x=1em, y=2.1ex}
\InputIfFileExists{bone-black.tikz}{}{\input{./tikz/bone-black.tikz}}
\tikzset{x=1em, y=1.5ex}
\myeq{\bullet-bo}
\tikzset{x=1em, y=2.1ex}
\InputIfFileExists{empty-diag.tikz}{}{\input{./tikz/empty-diag.tikz}}
\tikzset{x=1em, y=1.5ex}

 \end{equation*}
\vspace{0.7pt}
\hrule
\vspace{1.7pt}
\begin{equation*}

\tikzset{x=1em, y=2.1ex}
\InputIfFileExists{add-Frobenius-left.tikz}{}{\input{./tikz/add-Frobenius-left.tikz}}
\tikzset{x=1em, y=1.5ex}
\myeq{\circ-fr1} 
\tikzset{x=1em, y=2.1ex}
\InputIfFileExists{add-Frobenius.tikz}{}{\input{./tikz/add-Frobenius.tikz}}
\tikzset{x=1em, y=1.5ex}
\myeq{\circ-fr2} 
\tikzset{x=1em, y=2.1ex}
\InputIfFileExists{add-Frobenius-right.tikz}{}{\input{./tikz/add-Frobenius-right.tikz}}
\tikzset{x=1em, y=1.5ex}

\qquad 
\tikzset{x=1em, y=2.1ex}
\InputIfFileExists{add-special.tikz}{}{\input{./tikz/add-special.tikz}}
\tikzset{x=1em, y=1.5ex}
\myeq{\circ-sp}
\tikzset{x=1em, y=2.1ex}
\InputIfFileExists{id.tikz}{}{\input{./tikz/id.tikz}}
\tikzset{x=1em, y=1.5ex}
\qquad
\tikzset{x=1em, y=2.1ex}
\InputIfFileExists{bone-white.tikz}{}{\input{./tikz/bone-white.tikz}}
\tikzset{x=1em, y=1.5ex}
\myeq{\circ-bo}
\tikzset{x=1em, y=2.1ex}
\InputIfFileExists{empty-diag.tikz}{}{\input{./tikz/empty-diag.tikz}}
\tikzset{x=1em, y=1.5ex}

\end{equation*}
\vspace{0.7pt}
\hrule
\vspace{0.7pt}
\begin{equation*}
    
\tikzset{x=1em, y=2.1ex}
\InputIfFileExists{ccwhite.tikz}{}{\input{./tikz/ccwhite.tikz}}
\tikzset{x=1em, y=1.5ex}
 \myeq{cc-1} 
\tikzset{x=1em, y=2.1ex}
\InputIfFileExists{cc-black.tikz}{}{\input{./tikz/cc-black.tikz}}
\tikzset{x=1em, y=1.5ex}
 \qquad \quad 
\tikzset{x=1em, y=2.1ex}
\InputIfFileExists{ccwhiteop.tikz}{}{\input{./tikz/ccwhiteop.tikz}}
\tikzset{x=1em, y=1.5ex}
 \myeq{cc-2} 
\tikzset{x=1em, y=2.1ex}
\InputIfFileExists{cc-blackop.tikz}{}{\input{./tikz/cc-blackop.tikz}}
\tikzset{x=1em, y=1.5ex}
 
    \qquad\qquad\qquad\WcounitT \mysubeq{\circ\bullet-inc} \BcounitT
\end{equation*}
\caption{Axioms of Interacting Hopf Algebras ($\IH{\field}$).\label{fig:ih}}
\end{figure*}
\begin{figure}
\input{figures/polyc/ihp-axioms}
\caption{Axioms of $\IHP{\field}$}\label{fig:axiom:ihp}
\end{figure}

\section{The Theory of Polyhedral cones}\label{sec:polyc}
Hereafter, we assume $\field$ to be an \emph{ordered field}, namely a field equipped with a total order $\leq$ such that for all $i,j,k \in \field$: (a) if $i \leq j$, then $i + k \leq j + k$; (b) if $0 \leq i$ and $0 \leq j$, then $0 \leq ij$.

We extend the signature in~\eqref{eq:SFcalculusSyntax1} and~\eqref{eq:SFcalculusSyntax2} with the following generator
\begin{equation}\label{eq:>}
\greq
\end{equation}
and denote the resulting free prop $\CDP$.
The morphism $\dsempO \colon \CDP \rightarrow \RelX{\mathsf k}$ behaves as $\eqref{eq:semIH}$ for the generators in~\eqref{eq:SFcalculusSyntax1} and~\eqref{eq:SFcalculusSyntax2}, whereas for $\greq$, it is defined as hinted by our syntax: $\dsemp{\greq}=\{\, (x, y) \mid x, y \in \mathsf{k}, x \geq y \,\}$.

\begin{example}
Let \input{figures/polyc/pnf-def1} be a diagram in $\FC$ denoting some matrix $A$ (see Example \ref{ex:matrix}). Consider the following circuit in $\CDP$
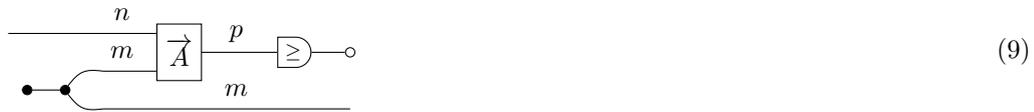
\begin{equation}\label{eq:polyhedralcnf}
\input{figures/poly/Ihpgeneric}
\end{equation}
It easy to check that its semantics is the relation $C= \{ (x,y) \in \mathsf{k}^{n}\times \mathsf{k}^{m} \mid A\begin{pmatrix*} x \\ y \end{pmatrix*} \geq 0 \}$. Thus~\eqref{eq:polyhedralcnf} denotes a \emph{polyhedral cone}, i.e., the set of solutions of some system of linear inequations.
\end{example}

We denote by $\IHP{\mathsf k}$ the prop generated by the theory consisting of this signature (namely,~\eqref{eq:SFcalculusSyntax1}, \eqref{eq:SFcalculusSyntax2} and~\eqref{eq:>}), the axioms of $\IH{\mathsf k}$ and the axioms in Figure \ref{fig:axiom:ihp}, where $\leqd$ is just $\greq^{\op}$ (see Example~\ref{ex:cc}). The first two rows of axioms describe the interactions of $\greq$ with the generators in~\eqref{eq:SFcalculusSyntax1}.
The third row asserts that $\geq$ is antisymmetric and satisfies an appropriate spider condition.
In the last axiom, the right-to-left inclusion states that for all
$k,l\in\field$, there exists an upper bound $u$, i.e.\ $u\geq k$, $u\geq l$.
The left-to-right inclusion is redundant.

The axioms are perhaps surprising: e.g.\ reflexivity and transitivity are not included. As a taster for working with the diagrammatic calculus, we prove these properties below.
\begin{remark}
We will often use the alternative antipode notation $\tantipode$ for the scalar $\tscalar{-1}$.
\end{remark}
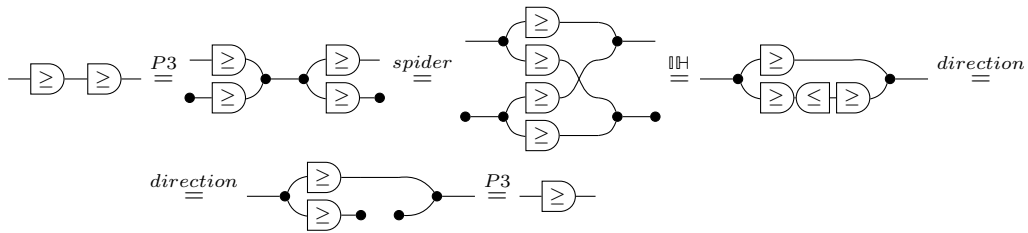
\input{figures/polyc/tran-proof}
The derivation above proves transitivity. The following derivation 
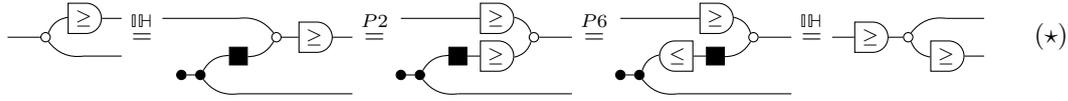
\input{figures/polyc/lemma-reflx-proof}
is used below to show that $\geq$ is reflexive:
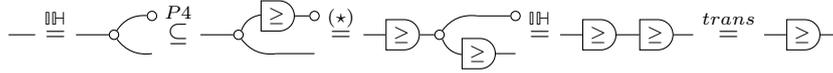
\input{figures/polyc/reflx-proof}
\begin{remark}
When we annotate equalities with $\IH{}$, we are making use of multiple unmentioned derived laws presented in related works. These can be seen to hold also by appealing to Theorem~\ref{co:complih}.
\end{remark}

Routine computations confirm that all the axioms are sound.
To prove completeness, we give diagrammatic proofs of several well-known results.

\begin{proposition}[Fourier-Motzkin elimination]\label{prop:FMElimitation}
For each arrow $A \colon n \rightarrow m$ of $\FC$, there exists an arrow $B \colon n-1 \rightarrow l$ of $\FC$ such that
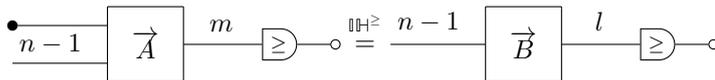
\input{figures/polyc/fm1}
\end{proposition}
The proof, in Appendix \ref{app:FMelimination}, 
mimics the Fourier-Motzkin elimination, an algorithm for eliminating variables from a system of linear inequations (i.e. projecting a polyhedral cone).

The next step is a normal form theorem. A circuit $c \colon n \rightarrow m$ of $\CDP$ is said to be in \emph{polyhedral normal form} if there is an arrow \input{figures/polyc/pnf-def1} of $\FC$, such that $c=\eqref{eq:polyhedralcnf}$.
\begin{theorem}[First Normal Form]
    \label{th:polynf}
    For each arrow $c \colon n \rightarrow m$ of $\CDP$, there is another arrow $d \colon n \rightarrow m$ of $\CDP$ in polyhedral normal form such that $c \myeq{\IHP{}} d$.
\end{theorem}
The proof is by induction on the structure of $\CDP$.
The only challenging case is sequential composition, which uses the Fourier-Moztkin elimination. Details are in Appendix \ref{app:polynf}.

Diagrams in $\CDP$ enjoy a second normal form: an arrow $c \colon n \rightarrow m$ of $\CDP$ is said to be in \textit{finitely generated normal form} if there is an arrow \input{figures/polyc/fnf-def1} of $\FCop$, such that
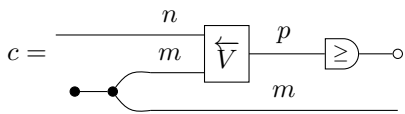
\[ c =  \input{figures/polyc/fnf-def2} \]
As recalled in Example \ref{ex:matrix}, the semantics of the circuit in $\FCop$ is $\dsem{\input{figures/polyc/fnf-def1}}=\{(u,z)\in \field^{n+m}\times \field^{p} \mid u=Vz\}$ for some $(n+m)\times p$ matrix $V$. Then, it is easy to check that $\dsemp{c} = \{ (x,y) \in \mathsf{k}^{n}\times \mathsf{k}^{m} \mid \exists z\in \field^{p} \text{ s.t. } \begin{pmatrix*} x \\ y \end{pmatrix*} = Vz, z \geq 0 \}$. The matrix $V$ can be regarded as a set of column vectors $\{ v_1, \dots, v_p \}$ and $\dsemp{c}$ as the \emph{conic combination} of those vectors, defined as $\mathsf{cone}(V) = \{ z_1 v_1 + \ldots + z_p v_p \mid z_i \in \mathsf{k}, \,z_i\geq 0\}$. Sets of vectors generated in this way are known as \emph{finitely generated cones}, which justifies the name of the normal form.

To prove the existence of this normal form, we introduce the polar operator, an important construction in convex analysis which, in our approach, has a simple inductive definition.

\begin{definition}
    The \textit{polar} operator $\cdot^{\circ} \colon \CDP \rightarrow \CDP$ is the functor inductively defined as:
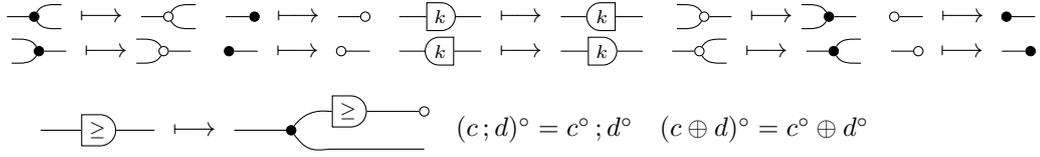
\[
\begin{array}{ccccc}
\Bcomult \ \longmapsto \ \Wcomult
&
\Bcounit \ \longmapsto \ \Wcounit
&
\tscalar{k} \ \longmapsto \ \tcoscalar{k}
&
\Wmult \ \longmapsto \Bmult
&
\Wunit \ \longmapsto \ \Bunit
\\
\Bmult \ \longmapsto \Wmult
&
\Bunit \ \longmapsto \ \Wunit
&
\tcoscalar{k} \ \longmapsto \ \tscalar{k}
&
\Wcomult \ \longmapsto \ \Bcomult
&
\Wcounit \ \longmapsto \ \Bcounit
\end{array}
\]
\[
\begin{array}{ccc}
\input{figures/polyc/polarpolar-geq} & (c \scolon d)^{\circ} = c^\circ \scolon d^\circ  & (c \oplus d)^{\circ} = c^\circ \oplus d^\circ
\end{array}
\]
\end{definition}
The polar operator enjoys the following useful properties.
\begin{proposition} \label{th:polarpolar}
For all arrows $c,d \colon n \rightarrow m$ in $\CDP$, it holds that
\begin{enumerate}
    \item if $c \stackrel{\IHP{}}{\subseteq} d$ then $(d)^\circ \stackrel{\IHP{}}{\subseteq} (c)^\circ$;
    \item $(c^\circ)^\circ \myeq{\IHP{}} c$;
    \item if $c$ is an arrow of $\FC$, then $c^\circ$ is an arrow of $\FCop$.
\end{enumerate}
\end{proposition}

\begin{proposition}
    \label{th:polarfingen}
    For each arrow $c \colon n \rightarrow m$ of $\CDP$ in polyhedral normal form there is an arrow $d \colon n \rightarrow m$ of $\CDP$ in finitely generated normal form, such that $(c)^\circ \myeq{\IHP{}} d$.
\end{proposition}
\begin{proof} 
    \input{figures/polyc/polarfingen-proof}
\end{proof}

\begin{theorem}[Second Normal Form]
    \label{co:fingennf}
    For each arrow $c \colon n \rightarrow m$ of $\CDP$, there is an arrow $d \colon n \rightarrow m$ of $\CDP$ in finitely generated normal form such that $c \myeq{\IHP{}} d$.
\end{theorem}
\begin{proof}
    By Theorem \ref{th:polynf}, there exists an arrow $p \colon n \rightarrow m$ of $\CDP$ in polyhedral normal form, such that $c^\circ \myeq{\IHP{}} p$. By Proposition \ref{th:polarpolar}.1, $p^\circ \myeq{\IHP{}} (c^\circ)^\circ$ and, by Proposition \ref{th:polarpolar}.2 $, (c^\circ)^\circ \myeq{\IHP{}} c$, thus $p^\circ \myeq{\IHP{}} c$. Since $p$ is in polyhedral normal form, by Proposition \ref{th:polarfingen}, there exists a circuit $d$ in finitely generated normal form such that $d \myeq{\IHP{}} p^\circ \myeq{\IHP{}} c$.
\end{proof}
An immediate consequence of the two normal form theorems is the well-known Weyl-Minkowski theorem,
which states that every polyhedral cone is finitely generated and, vice-versa, every finitely generated cone is polyhedral.
It is worth emphasising that neither the polyhedral nor the finitely generated normal form are unique: different matrices may give rise to the same cone. However, with the finitely generated normal form, proving completeness requires only a few more lemmas, which are given in Appendix \ref{app:completenessproofs}.

\begin{theorem}[Completeness]\label{th:compl2}
    For all circuits $c, d \in \CDP$, if $\dsemp{c} \subseteq \dsemp{d}$ then $c \mysubeq{\IHP{}} d$.
\end{theorem}

We now come to the problem of expressivity: what is the image of $\CDP$ through $\dsemO$? It turns out that $\CDP$ denotes exactly the arrows of $\propCat{PC}_{\mathsf k}$ (see Definition \ref{def:pc}).


%
\begin{proposition}[Expressivity] \label{th:ihpiso}
    For each arrow $C \colon n \rightarrow m$ in $\propCat{PC}_{\mathsf k}$ there exists a circuit $c \colon n \rightarrow m$ of $\CDP$, such that $C = \dsemp{c}$. Vice-versa,
    for each circuit $c \colon n \rightarrow m$ of $\CDP$ there exists an arrow $C \colon n \rightarrow m$ of $\propCat{PC}_{\mathsf k}$, such that $\dsemp{c} = C$.

\end{proposition}

By Theorem \ref{th:compl2} and Proposition \ref{th:ihpiso} it follows that
\begin{corollary}
$\IHP{\field} \cong \propCat{PC}_{\mathsf k}$
\end{corollary}
The above allows us to conclude that $\IHP{\field}$ is the theory of polyhedral cones.



\section{The theory of Polyhedra}\label{sec:poly}

In~\cite{BonchiPSZ19}, the signature of $\CD$ was extended with an additional generator
\begin{equation}\label{eq:perp}
\One
\end{equation}
with semantics $\dsem{\One}=\{(\nullvec, 1)\}$. The three leftmost equations in Figure~\ref{fig:aihpaxioms} provide a complete axiomatisation for semantic equality. In terms of expressivity, the resulting calculus expresses exactly \emph{affine spaces}, namely sets of solutions of $Ax +b = 0$ for some matrix $A$ and vector $b$.


Here we extend $\CDP$ with~\eqref{eq:perp}. Let $\ACDP$ be the prop freely generated by~\eqref{eq:SFcalculusSyntax1}, \eqref{eq:SFcalculusSyntax2}, \eqref{eq:>} and~\eqref{eq:perp}. The circuits of $\ACDP$ can denote \emph{polyhedra}, namely sets 
$P = \{\, x \in \field^n \mid Ax +b \geq 0 \,\}$. Observe that the empty set $\emptyset$ is a polyhedron, but not a polyhedral cone.

\begin{example}Let \input{figures/polyc/pnf-def1} and \input{figures/poly/b-diag} be circuits in $\FC$ denoting, respectively some matrix $A$ and some vector $b$. Consider the following circuit in $\ACDP$.
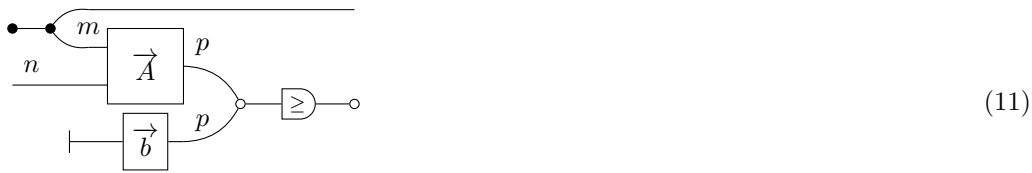
\begin{equation}\label{eq:polynf}
\input{figures/poly/Nf1-general}
\end{equation}
It is easy to check that its semantics is the relation
$P=\{\, (x,y) \in \field^n \times \field^m \mid A\begin{pmatrix*} y \\ x \end{pmatrix*} + b \geq 0 \,\}$. Another useful circuit is $\onezero$: $\dsem{\onezero}= \{(\nullvec, 1)\} ; \{(0, \nullvec)\} = \emptyset$. Intuitively, it behaves as a logical false, since for any relation $R$ in $\RelX{\field}$, $R \oplus \emptyset = \emptyset = \emptyset \oplus R$.
\end{example}

\begin{figure}[t]
    \input{figures/poly/aihp-axioms}
    \caption{Axioms of $\AIHP{\field}$}
    \label{fig:aihpaxioms}
\end{figure}
In order to obtain a complete axiomatisation, it is enough to add to the three axioms in~\cite{BonchiPSZ19}, only one axiom: $AP1$ in Figure~\ref{fig:aihpaxioms}. Intuitively, $AP1$ states that $1\geq 0$. The prop freely generated by~\eqref{eq:SFcalculusSyntax1}, \eqref{eq:SFcalculusSyntax2}, \eqref{eq:>}, \eqref{eq:perp} and the axioms in Figures~\ref{fig:ih}, \ref{fig:axiom:ihp} and~\ref{fig:aihpaxioms} is denoted by $\AIHP{\mathsf k}$. 

With these axioms, any circuit in $\ACDP$ can be shown equivalent to one of the form~\eqref{eq:polynf}. This is shown using the first normal form (Theorem~\ref{th:polynf}) for $\CDP$ and the following lemma.

\begin{lemma}\label{lm:nf2}
For any  $c \colon n \rightarrow m$ of $\ACDP$, there exists $c' \colon n+1 \rightarrow m$ of $\CDP$ such that
\input{figures/poly/nf2-def}
\end{lemma}

\begin{theorem}\label{th:nf}
For all $c$ of $\ACDP$, there exist $d$ in the form of~\eqref{eq:polynf} such that $c \myeq{\AIHP{}} d$.
\end{theorem}
To prove completeness, the notion of homogenization is pivotal. The \emph{homogenization} of a polyhedron $P= \{ x \in \field^n \mid Ax +b \geq 0 \}$ is the cone $P^H = \{ (x, y) \in \field^{n+1} \mid Ax + by \geq 0, y \geq 0 \}$.
Diagrammatically, this amounts to replace the $\One$ in~\eqref{eq:polynf} with $
\tikzset{x=1em, y=2.1ex}
\InputIfFileExists{polargeq.tikz}{}{\input{./tikz/polargeq.tikz}}
\tikzset{x=1em, y=1.5ex}
$, obtaining the following diagram
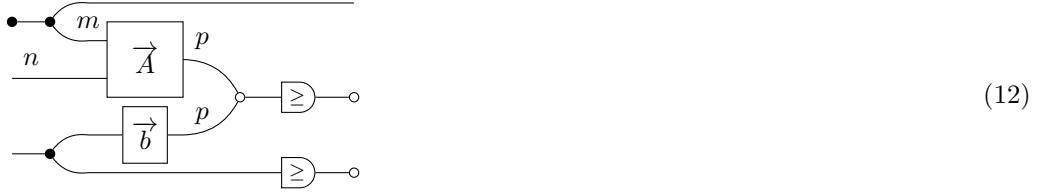
\begin{equation}
    \input{figures/poly/homog}
\end{equation}

\begin{lemma}\label{th:polyinc}
Let $P_1, P_2 \subseteq \field^n$ be two non-empty polyhedra. Then,
$P_1 \subseteq P_2$ iff $P_1^H \subseteq P_2^H$.
\end{lemma}
Using Theorem~\ref{th:nf} and Lemma~\ref{th:polyinc}, we can reduce completeness for \emph{non-empty} polyhedra to completeness of polyhedral cones.
\begin{theorem}\label{th:compl}
Let $c, d \colon n \rightarrow m$ in $\ACDP$ denote non-empty polyhedra. If
$\dsem{c} \subseteq \dsem{d}$, then $c \stackrel{\AIHP{}}{\subseteq} d$.
\end{theorem}
Completeness for empty polyhedra requires a few additional lemmas, given in Appendix~\ref{app:poly}.
%
\begin{theorem}\label{th:acompl2}
For all circuits $c$ in $\ACDP$, if $\dsem{c}= \emptyset$ then $c \myeq{\AIHP{}} 
\tikzset{x=1em, y=2.1ex}
\InputIfFileExists{one-false-disconnect.tikz}{}{\input{./tikz/one-false-disconnect.tikz}}
\tikzset{x=1em, y=1.5ex}
$ 
\end{theorem}

\begin{corollary}
[Completeness]\label{co:compl}
For all circuits $c,d$ in $\ACDP$, if
$\dsem{c} \subseteq \dsem{d}$ then $c \mysubeq{\AIHP{}} d$.
\end{corollary}

Finally, we characterise the semantic image of circuits in $\ACDP$.
%

\begin{proposition}[Expressivity]\label{th:aihpiso}
    For each arrow $P \colon n \rightarrow m$ in $\propCat{P}_{\mathsf k}$ there exist a circuit $c \colon n \rightarrow m$ in $\ACDP$, such that $P = \dsemp{c}$. Vice-versa, for each circuit $c \colon n \rightarrow m$ of $\ACDP$ there exists an arrow $P \colon n \rightarrow m$ of $\propCat{P}_{\mathsf k}$, such that $\dsemp{c} = P$.
\end{proposition}
Indeed, $\AIHP{\field}$ is the theory of polyhedra:
\begin{corollary}\label{cor:isopol}
$\AIHP{\field} \cong \propCat{P}_{\mathsf k}$
\end{corollary}


%

\begin{example}[Flow networks]\label{ex:fnet}
Consider again flow networks, previously mentioned in the Introduction: edges with capacity $k$ can be expressed in $\ACDP$ by the diagram in \eqref{eq:flowneedge} hereafter referred as $
\tikzset{x=1em, y=2.1ex}
\begin{tikzpicture}
	\begin{pgfonlayer}{nodelayer}
		\node [style=none] (0) at (-1.5, 0) {};
		\node [style=edge] (1) at (0, 0) {$k$};
		\node [style=none] (2) at (1.5, 0) {};
	\end{pgfonlayer}
	\begin{pgfonlayer}{edgelayer}
		\draw (0.center) to (1.center);
		\draw (1.center) to (2.center);
	\end{pgfonlayer}
\end{tikzpicture}
}
\tikzset{x=1em, y=1.5ex}
$. Observe that $\dsem{
\tikzset{x=1em, y=2.1ex}
\begin{tikzpicture}
	\begin{pgfonlayer}{nodelayer}
		\node [style=none] (0) at (-1.5, 0) {};
		\node [style=edge] (1) at (0, 0) {$k$};
		\node [style=none] (2) at (1.5, 0) {};
	\end{pgfonlayer}
	\begin{pgfonlayer}{edgelayer}
		\draw (0.center) to (1.center);
		\draw (1.center) to (2.center);
	\end{pgfonlayer}
\end{tikzpicture}
}
\tikzset{x=1em, y=1.5ex}
}=\{(x,x) \mid 0 \leq x\leq k\}$ is exactly the expected meaning of an edge in a flow network. Nodes with $n$ incoming edges and $m$ outgoing edges can be encoded by the diagram $\;
\tikzset{x=1em, y=2.1ex}
\InputIfFileExists{flownet/vertex.tikz}{}{\input{./tikz/flownet/vertex.tikz}}
\tikzset{x=1em, y=1.5ex}
$. Again the semantics is the expected one: the total incoming flow must be equal to the total outgoing flow. For an example of the encoding, check the flow network in \eqref{eq:flownet} and the corresponding diagram in \eqref{eq:flownetd}. 

The axioms in $\AIHP{\field}$ can be exploited to compute the maximum flow of a network. By using the following two derived laws (proved in Appendix \ref{app:flow}) 

\begin{minipage}{0.4\textwidth}
\begin{equation}\label{eq:dlaw2}
    
\tikzset{x=1em, y=2.1ex}
\InputIfFileExists{flownet/law2-1.tikz}{}{\input{./tikz/flownet/law2-1.tikz}}
\tikzset{x=1em, y=1.5ex}
 \myeq{\AIHP{}} 
\tikzset{x=1em, y=2.1ex}
\begin{tikzpicture}
	\begin{pgfonlayer}{nodelayer}
		\node [style=none] (0) at (-1.5, 0) {};
		\node [style=edge] (1) at (0, 0) {$k+l$};
		\node [style=none] (2) at (1.5, 0) {};
	\end{pgfonlayer}
	\begin{pgfonlayer}{edgelayer}
		\draw (0.center) to (1.center);
		\draw (1.center) to (2.center);
	\end{pgfonlayer}
\end{tikzpicture}
}
\tikzset{x=1em, y=1.5ex}

\end{equation}
\end{minipage}
\begin{minipage}{0.5\textwidth}
\begin{equation}\label{eq:dlaw1}
 
\tikzset{x=1em, y=2.1ex}
\InputIfFileExists{flownet/law1-1.tikz}{}{\input{./tikz/flownet/law1-1.tikz}}
\tikzset{x=1em, y=1.5ex}
 \myeq{\AIHP{}} 
\tikzset{x=1em, y=2.1ex}
\InputIfFileExists{flownet/law1-2.tikz}{}{\input{./tikz/flownet/law1-2.tikz}}
\tikzset{x=1em, y=1.5ex}

 \quad (k+q \leq l)
\end{equation}
\end{minipage}


\noindent
one can transform the diagram in \eqref{eq:flownetd} into $
\tikzset{x=1em, y=2.1ex}
}
\tikzset{x=1em, y=1.5ex}
$
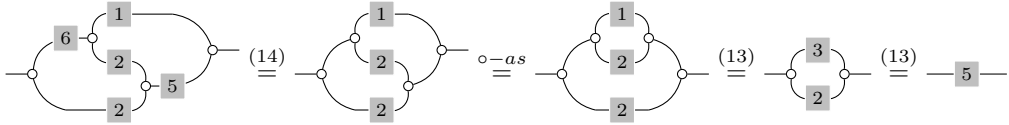
\begin{align*}
    &
\tikzset{x=1em, y=2.1ex}
\begin{tikzpicture}
	\begin{pgfonlayer}{nodelayer}
		\node [style=edge] (29) at (1.75, 2.5) {$1$};
		\node [style=white] (41) at (-1.5, 0) {};
		\node [style=none] (42) at (-2.5, 0) {};
		\node [style=edge] (47) at (1.75, 0.5) {$2$};
		\node [style=white] (67) at (0.75, 1.5) {};
		\node [style=white] (69) at (5.25, 1) {};
		\node [style=white] (73) at (2.75, -0.5) {};
		\node [style=edge] (74) at (1.75, -1.5) {$2$};
		\node [style=none] (91) at (3.75, 2.5) {};
		\node [style=edge] (92) at (3.75, -0.5) {$5$};
		\node [style=edge] (94) at (-0.25, 1.5) {$6$};
		\node [style=none] (95) at (-0.25, -1.5) {};
		\node [style=none] (96) at (6.25, 1) {};
	\end{pgfonlayer}
	\begin{pgfonlayer}{edgelayer}
		\draw (42.center) to (41);
		\draw [bend left=45] (67) to (29);
		\draw [bend right=45] (67) to (47);
		\draw [bend right=45] (74) to (73);
		\draw (94) to (67);
		\draw (73) to (92);
		\draw (69) to (96.center);
		\draw [bend left] (41) to (94);
		\draw [bend right] (41) to (95.center);
		\draw [bend right] (92) to (69);
		\draw [bend right] (69) to (91.center);
		\draw [bend right=45] (73) to (47);
		\draw (29) to (91.center);
		\draw (95.center) to (74);
	\end{pgfonlayer}
\end{tikzpicture}
}
\tikzset{x=1em, y=1.5ex}
 \myeq{\eqref{eq:dlaw1}} 
\tikzset{x=1em, y=2.1ex}
\begin{tikzpicture}
	\begin{pgfonlayer}{nodelayer}
		\node [style=edge] (29) at (1.75, 2.5) {$1$};
		\node [style=white] (41) at (-0.5, 0) {};
		\node [style=none] (42) at (-1.5, 0) {};
		\node [style=edge] (47) at (1.75, 0.5) {$2$};
		\node [style=white] (67) at (0.75, 1.5) {};
		\node [style=white] (69) at (4, 1) {};
		\node [style=white] (73) at (2.75, -0.5) {};
		\node [style=edge] (74) at (1.75, -1.5) {$2$};
		\node [style=none] (91) at (2.75, 2.5) {};
		\node [style=none] (95) at (0.75, -1.5) {};
		\node [style=none] (96) at (5, 1) {};
	\end{pgfonlayer}
	\begin{pgfonlayer}{edgelayer}
		\draw (42.center) to (41);
		\draw [bend left=45] (67) to (29);
		\draw [bend right=45] (67) to (47);
		\draw [bend right=45] (74) to (73);
		\draw (95.center) to (74);
		\draw (29) to (91.center);
		\draw (69) to (96.center);
		\draw [bend right] (41) to (95.center);
		\draw [bend right] (69) to (91.center);
		\draw [bend right=45] (73) to (47);
		\draw [bend right] (73) to (69);
		\draw [bend left] (41) to (67);
	\end{pgfonlayer}
\end{tikzpicture}
}
\tikzset{x=1em, y=1.5ex}
 \myeq{\circ-as} 
\tikzset{x=1em, y=2.1ex}
\begin{tikzpicture}
	\begin{pgfonlayer}{nodelayer}
		\node [style=edge] (29) at (1.75, 2.5) {$1$};
		\node [style=white] (41) at (-0.5, 0) {};
		\node [style=none] (42) at (-1.5, 0) {};
		\node [style=edge] (47) at (1.75, 0.5) {$2$};
		\node [style=white] (67) at (0.75, 1.5) {};
		\node [style=edge] (74) at (1.75, -1.5) {$2$};
		\node [style=none] (95) at (0.75, -1.5) {};
		\node [style=white] (96) at (4, 0) {};
		\node [style=none] (97) at (5, 0) {};
		\node [style=white] (98) at (2.75, 1.5) {};
		\node [style=none] (99) at (2.75, -1.5) {};
	\end{pgfonlayer}
	\begin{pgfonlayer}{edgelayer}
		\draw (42.center) to (41);
		\draw [bend left=45] (67) to (29);
		\draw [bend right=45] (67) to (47);
		\draw (95.center) to (74);
		\draw [bend right] (41) to (95.center);
		\draw [bend left] (41) to (67);
		\draw (97.center) to (96);
		\draw [bend left] (96) to (99.center);
		\draw [bend right] (96) to (98);
		\draw (74) to (99.center);
		\draw [bend right=45] (47) to (98);
		\draw [bend right=45] (98) to (29);
	\end{pgfonlayer}
\end{tikzpicture}
}
\tikzset{x=1em, y=1.5ex}
 \myeq{\eqref{eq:dlaw2}} 
\tikzset{x=1em, y=2.1ex}
\begin{tikzpicture}
	\begin{pgfonlayer}{nodelayer}
		\node [style=edge] (29) at (0, 1) {$3$};
		\node [style=edge] (47) at (0, -1) {$2$};
		\node [style=white] (67) at (-1, 0) {};
		\node [style=white] (98) at (1, 0) {};
		\node [style=none] (99) at (-2, 0) {};
		\node [style=none] (100) at (2, 0) {};
	\end{pgfonlayer}
	\begin{pgfonlayer}{edgelayer}
		\draw [bend left=45] (67) to (29);
		\draw [bend right=45] (67) to (47);
		\draw [bend right=45] (47) to (98);
		\draw [bend right=45] (98) to (29);
		\draw (67) to (99.center);
		\draw (98) to (100.center);
	\end{pgfonlayer}
\end{tikzpicture}
}
\tikzset{x=1em, y=1.5ex}
 \myeq{\eqref{eq:dlaw2}} 
\tikzset{x=1em, y=2.1ex}
}
\tikzset{x=1em, y=1.5ex}

\end{align*}
meaning that $\dsem{\eqref{eq:flownetd}}=\{(x,x) \mid 0\leq x \leq 5\}$, i.e., the maximum flow of  \eqref{eq:flownetd} is exactly 5.
\end{example}

\newcommand{\pre}[1]{{^\circ{#1}}}
\newcommand{\post}[1]{{#1^\circ}}

\newcommand{\markingone}{v_1}
\newcommand{\markingtwo}{v_2}
\newcommand{\firingstep}{f}
\newcommand{\irr}[1]{\overline{#1}}

\newcommand{\Pl}{\mathsf{Pl}}
\newcommand{\encodingPetri}[1]{\mathcal{E}(#1)}
\newcommand{\Circ}{\mathsf{Circ}} 
\newcommand{\CircP}{\Circ_{p}} 
\tikzset{ha/.style={mat,rounded rectangle,rounded rectangle left arc=none}}
\tikzset{mat/.style={draw,fill=white,rectangle,node font=\scriptsize}}

\tikzstyle{none}=[inner sep=0pt]
\tikzstyle{plain}=[inner sep=0pt]
\tikzstyle{black}=[circle, draw=black, fill=black, inner sep=0pt, minimum size=4pt]
\tikzstyle{black-faded}=[circle, draw=light-gray, fill=light-gray, inner sep=0pt, minimum size=4pt]
\tikzstyle{white}=[circle, draw=black, fill=white, inner sep=0pt, minimum size=4.5pt]
\tikzstyle{white-faded}=[circle, draw=light-gray, fill=white, inner sep=0pt, minimum size=4.5pt]
\tikzstyle{delay}=[fill=black, regular polygon, regular polygon sides=3,rotate=-90, scale=.55]
\tikzstyle{delay-op}=[fill=black, regular polygon, regular polygon sides=3,rotate=90, scale=.55]
\tikzstyle{reg}=[draw, fill=white, rounded rectangle, rounded rectangle left arc=none, minimum height=1.2em, minimum width=1.4em, node font={\scriptsize}]
\tikzstyle{coreg}=[draw, fill=white, rounded rectangle, rounded rectangle right arc=none, minimum height=1.2em, minimum width=1.4em, node font={\scriptsize}]
\tikzstyle{rn}=[circle, draw=red, fill=red, inner sep=0pt, minimum size=4pt]

\tikzstyle{pl}=[circle,thick,draw=black!75,fill=white,minimum size=9pt]
\tikzstyle{port}=[circle, fill,inner sep=1.2pt]

\newcommand{\place}{
\begin{tikzpicture}
	\begin{pgfonlayer}{nodelayer}
		\node [style=place] (0) at (0, 0) {};
		\node [style=none] (1) at (-1.5, 0) {};
		\node [style=none] (2) at (1.2, 0) {};
	\end{pgfonlayer}
	\begin{pgfonlayer}{edgelayer}
		\draw [style=arrow] (1.center) to (0);
		\draw (0) to (2.center);
	\end{pgfonlayer}
\end{tikzpicture}
}

\newcommand\register[1][$x$]{
  \tikz {
    \node[ha] (ha) {#1};
    \draw (ha.west) -- ++(-0.75, 0);
    \draw (ha.east) -- ++(0.75, 0);
  }
}

\newcommand{\sym}{
  \tikz {
    \draw (0,  0.5) .. controls (0.5,  0.5) and (0.5, -0.5) .. (1, -0.5);
    \draw (0, -0.5) .. controls (0.5, -0.5) and (0.5,  0.5) .. (1,  0.5);
  }
}

\newcommand{\Petri}{\mathsf{Petri}} %
\newcommand{\RC}{\mathsf{Rc}} 
\newcommand{\isoRC}{I} 
\newcommand{\RCS}{\RC_s} 
\newcommand{\CircS}{\Circ_s} 
\newcommand{\AddRel}{\mathsf{AddRel}}
\newcommand{\isoRCS}{\isoRC_s} 
\newcommand{\posreals}{{\mathbb{R}_{+}}}
\newcommand{\semantics}[1]{[\![ #1 ]\!]} 

\let\tinymatrix\smallmatrix
\let\endtinymatrix\endsmallmatrix
\patchcmd{\tinymatrix}{\scriptstyle}{\scriptscriptstyle}{}{}
\patchcmd{\tinymatrix}{\scriptstyle}{\scriptscriptstyle}{}{}
\patchcmd{\tinymatrix}{\vcenter}{\vtop}{}{}
\patchcmd{\tinymatrix}{\bgroup}{\bgroup\scriptsize}{}{}

\newcommand{\diagstate}[4]{
\tikzset{x=1em, y=2.1ex}
\begin{tikzpicture}
	\begin{pgfonlayer}{nodelayer}
		\node [style=none] (0) at (-0.75, 0.5) {};
		\node [style=none] (1) at (-0.75, 1) {};
		\node [style=none] (2) at (-0.75, -0.5) {};
		\node [style=none] (3) at (0.75, -0.5) {};
		\node [style=none] (4) at (-0.75, -1) {};
		\node [style=none] (5) at (0.75, 0.5) {};
		\node [style=none] (6) at (2.5, -0.5) {};
		\node [style=none] (7) at (0.75, -0.5) {};
		\node [style=none] (8) at (0.75, -1) {};
		\node [style=none] (9) at (0.75, 1) {};
		\node [style=none] (10) at (0, 0) {$#1$};
		\node [style=none] (11) at (3.25, -0.5) {$#3$};
		\node [style=none] (12) at (-3.25, -0.5) {$#2$};
		\node [style=none] (13) at (-2.5, -0.5) {};
		\node [style=none] (14) at (-0.75, -0.5) {};
		\node [style=none] (15) at (0.75, 0.5) {};
		\node [style=none] (16) at (2.5, 0.5) {};
		\node [style=none] (17) at (-2.5, 0.5) {};
		\node [style=none] (18) at (-0.75, 0.5) {};
		\node [style=none] (19) at (-3.25, 0.5) {$#4$};
		\node [style=none] (20) at (3.25, 0.5) {$#4$};
	\end{pgfonlayer}
	\begin{pgfonlayer}{edgelayer}
		\draw [in=180, out=0, looseness=1.25] (7.center) to (6.center);
		\draw [semithick] (3.center) to (8.center);
		\draw [semithick] (4.center) to (2.center);
		\draw [semithick] (0.center) to (1.center);
		\draw [semithick] (9.center) to (5.center);
		\draw [semithick] (1.center) to (9.center);
		\draw [semithick] (5.center) to (3.center);
		\draw [semithick] (8.center) to (4.center);
		\draw [semithick] (2.center) to (0.center);
		\draw [in=180, out=0, looseness=1.25] (13.center) to (14.center);
		\draw [in=180, out=0, looseness=1.25] (15.center) to (16.center);
		\draw [in=180, out=0, looseness=1.25] (17.center) to (18.center);
	\end{pgfonlayer}
\end{tikzpicture}
\tikzset{x=1em, y=1.5ex}
}

\newcommand{\traceform}[4]{
\tikzset{x=1em, y=2.1ex}
\begin{tikzpicture}
	\begin{pgfonlayer}{nodelayer}
		\node [style=none] (0) at (-0.75, 1.25) {};
		\node [style=none] (1) at (0.75, 1) {};
		\node [style=none] (2) at (-0.75, 1.75) {};
		\node [style=none] (3) at (-0.75, -0.25) {};
		\node [style=none] (4) at (0.75, -0.25) {};
		\node [style=none] (5) at (-0.75, -0.75) {};
		\node [style=none] (6) at (-0.75, 1) {};
		\node [style=none] (7) at (0.75, 1.25) {};
		\node [style=none] (8) at (4.5, 0) {};
		\node [style=none] (9) at (0.75, 0) {};
		\node [style=none] (10) at (0.75, -0.75) {};
		\node [style=none] (11) at (0.75, 1.75) {};
		\node [style=none] (12) at (0, 0.5) {$#1$};
		\node [style=none] (13) at (5.5, 0) {$#3$};
		\node [style=none] (14) at (2.25, 2.5) {};
		\node [style=none] (15) at (-1.5, 2.5) {};
		\node [style=none] (16) at (-4.1, 0) {$#2$};
		\node [style=none] (17) at (-3.25, 0) {};
		\node [style=none] (18) at (-0.75, 0) {};
		\node [style=reg] (19) at (1.75, 1) {\tiny $x$};
		\node [style=none] (20) at (2.25, 1) {};
		\node [style=none] (21) at (3.25, 2.75) {$#4$};
		\node [style=none] (22) at (-1.5, 1) {};
	\end{pgfonlayer}
	\begin{pgfonlayer}{edgelayer}
		\draw [in=180, out=0, looseness=1.25] (9.center) to (8.center);
		\draw [semithick] (4.center) to (10.center);
		\draw [semithick] (5.center) to (3.center);
		\draw [semithick] (0.center) to (2.center);
		\draw [semithick] (11.center) to (7.center);
		\draw [semithick] (2.center) to (11.center);
		\draw [semithick] (7.center) to (4.center);
		\draw [semithick] (10.center) to (5.center);
		\draw [semithick] (3.center) to (0.center);
		\draw (15.center) to (14.center);
		\draw [in=180, out=0, looseness=1.25] (17.center) to (18.center);
		\draw (1.center) to (19);
		\draw (19) to (20.center);
		\draw (6.center) to (22.center);
		\draw [bend right=90, looseness=1.75] (20.center) to (14.center);
		\draw [bend left=90, looseness=1.75] (22.center) to (15.center);
	\end{pgfonlayer}
\end{tikzpicture}
\tikzset{x=1em, y=1.5ex}
}

\newcommand{\diagbox}[3]{
\tikzset{x=1em, y=2.1ex}
\begin{tikzpicture}
	\begin{pgfonlayer}{nodelayer}
		\node [style=none] (0) at (-0.75, 0.5) {};
		\node [style=none] (1) at (-0.75, 1) {};
		\node [style=none] (2) at (-0.75, -0.5) {};
		\node [style=none] (3) at (0.75, -0.5) {};
		\node [style=none] (4) at (-0.75, -1) {};
		\node [style=none] (5) at (0.75, 0.5) {};
		\node [style=none] (6) at (2.5, 0) {};
		\node [style=none] (7) at (0.75, 0) {};
		\node [style=none] (8) at (0.75, -1) {};
		\node [style=none] (9) at (0.75, 1) {};
		\node [style=none] (10) at (0, 0) {$#1$};
		\node [style=none] (11) at (3.25, 0) {$#3$};
		\node [style=none] (12) at (-3.25, 0) {$#2$};
		\node [style=none] (13) at (-2.5, 0) {};
		\node [style=none] (14) at (-0.75, 0) {};
	\end{pgfonlayer}
	\begin{pgfonlayer}{edgelayer}
		\draw [in=180, out=0, looseness=1.25] (7.center) to (6.center);
		\draw [semithick] (3.center) to (8.center);
		\draw [semithick] (4.center) to (2.center);
		\draw [semithick] (0.center) to (1.center);
		\draw [semithick] (9.center) to (5.center);
		\draw [semithick] (1.center) to (9.center);
		\draw [semithick] (5.center) to (3.center);
		\draw [semithick] (8.center) to (4.center);
		\draw [semithick] (2.center) to (0.center);
		\draw [in=180, out=0, looseness=1.25] (13.center) to (14.center);
	\end{pgfonlayer}
\end{tikzpicture}
\tikzset{x=1em, y=1.5ex}
}

\section{Adding states to polyhedra}\label{app:state}
We have shown that $\ACDP$ with its associated equational theory $\AIHP{\field}$ provides a sound and complete calculus for polyhedra. In this section, we extend the calculus with a canonical notion of \emph{state}. Our development follows step-by-step the general recipe illustrated in~\cite[\S 4]{BHPSZ-popl19}. 

\subsection{The Calculus of Stateful Polyhedral Processes}\label{sec:CSFG}


We call $\CSFG$ the prop freely generated by~\eqref{eq:SFcalculusSyntax1}, \eqref{eq:SFcalculusSyntax2}, \eqref{eq:>}, \eqref{eq:perp} and the following.
\begin{equation}\label{eq:register}
\circuitXT
\end{equation}
Intuitively, the register $\circuitXT$ is a synchronous buffer holding a value $k\in \field$:  when it receives $l\in \field$ on the left port, it emits $k$ on the right one and stores $l$. To give a formal semantics to such behaviour we exploit a ``state bootstrapping'' technique that appears in several places in the literature, e.g. in the setting of cartesian bicategories~\cite{katis1997bicategories} and geometry of interaction~\cite{hoshino2014memoryful}.


\begin{definition}[Stateful processes~\cite{katis1997bicategories}]\label{def:stateful}
Let $\mathsf{T}$ be a prop. Define $\mathsf{St}(\mathsf{T})$ as the prop where:
\begin{itemize}
\item morphisms $n\to n$ are pairs $(s, c)$ where $s\in\mathbb{N}$ and $c\colon s+n \to s+m$ is a morphism of $\mathsf{T}$, quotiented by the smallest equivalence relation including every instance of
\[\diagstate{c}{n}{m}{s} \quad \sim \quad 
\tikzset{x=1em, y=2.1ex}
\InputIfFileExists{process-equiv.tikz}{}{\input{./tikz/process-equiv.tikz}}
\tikzset{x=1em, y=1.5ex}
\]
for a permutation $\sigma: s\to s$; the order is defined as $(s, c)\mysubeq{\mathsf{St}(\mathsf{T})} (s,d)$ if and only $c \mysubeq{\mathsf{T}} d$.
\item the composition of $(s,c) \colon n \to m$ and $(t,d)\colon m \to o$ is $(s+t,e)$ where $e$ is the arrow of $\mathsf{T}$ given by
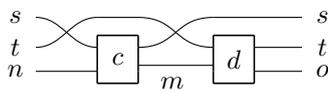
\[
\tikzset{x=1em, y=2.1ex}
\InputIfFileExists{process-composition.tikz}{}{\input{./tikz/process-composition.tikz}}
\tikzset{x=1em, y=1.5ex}
\]
\item the monoidal product of $(s_1,c_1)\colon n_1 \to m_1$ and  $(s_2,c_2)\colon n_2 \to m_2$ is $(s_1+s_2, e)$ where $e$ is given by
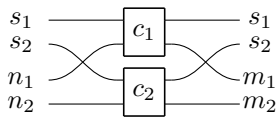
\[
\tikzset{x=1em, y=2.1ex}
\InputIfFileExists{process-tensor.tikz}{}{\input{./tikz/process-tensor.tikz}}
\tikzset{x=1em, y=1.5ex}
\]
\item the identity on $n$ is $(0,id_n)$ and the symmetry of $n,m$ is $(0,\sigma_{n,m})$.
\end{itemize}
\end{definition}


We use $\mathsf{St}(\RelX{\mathsf k})$ as our semantic domain: in an arrow $(s,R)\colon n \to m$, $s$ records the number of registers while $R\colon s+n \to s+m$ is a relation $R\subseteq \field^s \times \field^n \times \field^s \times \field^m$ containing quadruples $(u,l,v,r)$ representing transitions: $u$ and $v$ are the starting and arrival state (namely vectors in $\field^s$, holding a value in $\field$ to each of the $s$ registers), while $l$ and $r$ are vectors of values occurring on the left and the right ports. The equivalence relation $\sim$ ensures that registers remains anonymous: it equates arrows that only differ by a bijective relabelling of their lists of registers. This is, therefore, a syntactic form of equivalence similar in flavour to $\alpha$-equivalence, since it discards intentional details not relevant for the dynamics of processes.

We can now give the semantics of $\CSFG$ as the morphism $\dsemsO\colon \CSFG \to \mathsf{St}(\RelX{\mathsf k})$ defined:
\[\dsems{\circuitXT} = (1, \{(k,l,l,k)\mid l,k\in \field \}) \quad \text{ and } \quad  \dsems{o} = (0, \dsem{o})\]
for all generators $o$ in~\eqref{eq:SFcalculusSyntax1}, \eqref{eq:SFcalculusSyntax2}, \eqref{eq:>}, \eqref{eq:perp}. For instance $\dsems{\BcounitT}=(0,\{(\nullvec,k)\mid k\in \field\})$. The semantics of $\circuitX$ is the expected behaviour: from any state $k$ (the stored value), it makes a transition to state $l$ when $l$ is on the left port and $k$ is on the right. This can be restated as a structural operational semantics (sos) axiom $(\circuitXT, k) \dtrans{l}{k} (\circuitXT, l)$ where the labels above and under the arrow stand, respectively, for the values on the left and right ports.


Theorem 30 in~\cite{BHPSZ-popl19} ensures that no other data is needed for an axiomatisation: let $\SAIHP$ be the prop generated by~\eqref{eq:SFcalculusSyntax1}, \eqref{eq:SFcalculusSyntax2}, \eqref{eq:>}, \eqref{eq:perp}, \eqref{eq:register} and the axioms in Figures~\ref{fig:ih}, \ref{fig:axiom:ihp} and~\ref{fig:aihpaxioms}.

\begin{theorem}\label{thm:stateful}
For all $c,d$ in $\CSFG$, if
$\dsems{c} \subseteq \dsems{d}$ then $c \mysubeq{\SAIHP{}} d$. Moreover $\SAIHP \cong \mathsf{St}(\propCat{P}_{\mathsf k})$.
\end{theorem}
\begin{proof}[Proof of Theorem~\ref{thm:stateful}]
We make more clear the correspondence with~\cite{BHPSZ-popl19}. Considering the following diagram.
\[\xymatrix{{\CSFG} \ar@{->>}[r]^q \ar@(ur,ul)[rrrr]^{\dsemsO} & \SAIHP  \ar[r]^{F} & \mathsf{St}(\AIHP{\field}) \ar[r]^{\mathsf{St}(\cong)} & \mathsf{St}(\propCat{P}_{\mathsf k}) \ar@{>->}[r]^{\mathsf{St}(\iota)} &\mathsf{St}(\RelX{\mathsf k}) } \]
The morphism $\mathsf{St}(\iota)$ is just the obvious extension of the inclusion $\iota\colon \propCat{P}_{\mathsf k} \to \RelX{\mathsf k}$. Similarly, $\mathsf{St}(\cong)$ is the extension of the isomorphism shown in Corollary~\ref{cor:isopol}. The morphism $q$ is just the obvious quotient from $\CSFG$ to $\SAIHP$. The interesting part is provided by the morphism $F\colon \SAIHP \to \mathsf{St}(\AIHP{\field})$ defined in~\cite[\S 4.1]{BHPSZ-popl19}: take $\mathsf{T}$ as $\AIHP{\field}$ and $\mathsf{T}+\mathsf{X}$ as $\SAIHP$. By Theorem 30 in~\cite{BHPSZ-popl19}, since $\AIHP{\field}$ is compact closed, then $F$ is an isomorphism of props. To see that it is an isomorphism of ordered props, it is immediate to check that both $F$ and its inverse $G$ defined in~\cite[\S 4.1]{BHPSZ-popl19} preserves the order.
\end{proof}

We conclude by observing 
the semantics can be presented with intuitive sos rules.
Indeed, the same rules as in~\cite[\S 2]{BHPSZ-popl19}---interpreted over a field rather than the naturals---and:
\[\derivationRule{\scriptstyle x\mathrel{\geq} y}{(\greq, \bullet) \dtrans{x}{y} (\greq, \bullet)} \qquad (\One, \bullet) \dtrans{\bullet}{1} (\One, \bullet)\]
This diagrammatic language is, therefore, similar in flavour to traditional process calculi, and we call it the \emph{calculus of stateful polyhedral processes}. Theorem~\ref{thm:stateful} affirms that it expresses exactly the stateful polyhedral processes.

\subsection{Bounded Continuous Petri Nets}
\label{sec:petri}
Hereafter $\field$ is fixed to be the real numbers $\mathbb{R}$ and the set of non-negative reals is denoted by $\posreals = \{\,r\in\mathbb{R}\;|\;r\geq 0\,\}$.
A \emph{continuous} Petri net~\cite{DavidAlla10} differs from a (discrete) Petri net in that:
\begin{itemize}
\item markings are real valued -- that is, places hold a non-negative real number of tokens,
\item transitions can consume and produce non-negative real numbers of tokens,
\item transitions can be fired a non-negative real number amount of times -- for example a transition can be fired $0.5$ times, producing and consuming half the tokens.
\end{itemize}

\begin{definition}[Continuous Petri nets and their semantics]\label{def:Petri-net}
A Petri net $\mathcal{P}=(P,T,\pre{-},\post{-})$ consists of a finite set of places $P$, a finite set of transitions $T$, and functions
$\pre{-},\post{-}\from T\to \posreals^P$.
Given ${\mathbf{y}}, {\mathbf{z}}\in \posreals^P$, we write
${\mathbf{y}}\to{\mathbf{z}}$ if there exists ${\mathbf{t}}\in \posreals^T$
such that $\pre{{\mathbf{t}}}\leq {\mathbf{y}}$ and
${\mathbf{z}}={\mathbf{y}}-\pre{{\mathbf{t}}}+\post{{\mathbf{t}}}$,
where $\pre{{\mathbf{t}}}$ and $\post{{\mathbf{t}}}$ are the evident
liftings of $\pre{()}$ and $\post{()}$, e.g.\ $\pre{\mathbf{t}}(p) = \sum_{s\in T}\mathbf{t}(s)\cdot\pre{s}(p)$.
The (step) operational semantics of $\mathcal{P}$ is the relation $\osem{\mathcal{P}}=\{({\mathbf{y}}, {\mathbf{z}}) \mid {\mathbf{y}}\to{\mathbf{z}}\} \subseteq \posreals^P\times \posreals^P$.
\end{definition}

As for ordinary Petri nets, one can consider \emph{bounded nets}: each place has a maximum capacity $c\in \posreals\cup \{\top\}$: a place with capacity $\top$ is unbounded. The above definition is therefore extended with a boundary function $\mathbf{b} \in (\posreals\cup \{\top\})^P$ and the transition relation ${\mathbf{y}}\to{\mathbf{z}}$ is modified by additionally requiring that $\mathbf{y},\mathbf{z}  \leq \mathbf{b}$. Since $r \leq \top$ for all $r\in \posreals$, continuous Petri nets are instances of bounded continuous nets where every place is unbounded.

\medskip

To encode continuous Petri nets and their bounded variant as 
stateful polyhedral processes,
it is convenient to introduce syntactic sugar: the circuit below left is an adder that takes only positive values as inputs, the central circuit models a place, and the last one a transition.
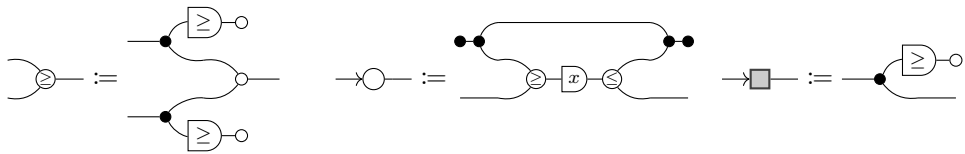
\[ \input{figures/petri/positive-add} \coloneqq \input{figures/petri/positive-add-2} \qquad \input{figures/petri/place} \coloneqq \input{figures/petri/place-2} \quad
\input{figures/petri/transition} \coloneqq \input{figures/petri/transition-2}\]
Observe that for $\input{figures/petri/place}$, it is essential the use of
$\input{figures/petri/positive-add}$ and and its opposite $\input{figures/petri/positive-add-op}$. Indeed, replacing them by ordinary adders $\WmultT$ and $\WcomultT$, would give as semantics the whole space $\mathbb{R}^2\times \mathbb{R}^2$, while as defined above
$\dsems{\input{figures/petri/place}}=(1,\{(m,i,m-o+i,o) \mid i,o,m \in \posreals,\, o \geq m\})$, modelling exactly the expected behaviour of a place.
In the diagrams below $\tscalar{c}$ is either a scalar $r\in \posreals$ or $\top= \BcounitT \BunitT$.
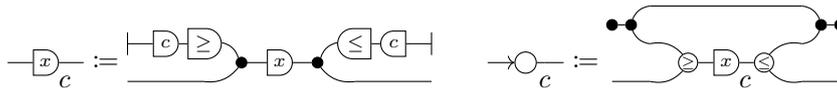
\[\begin{tikzpicture}
	\begin{pgfonlayer}{nodelayer}
		\node [style=reg] (0) at (0, 0) {$x$};
		\node [style=none] (1) at (-0.5, 0) {};
		\node [style=none] (2) at (0.5, 0) {};
		\node [style=none] (3) at (0.25, -0.25) {$c$};
	\end{pgfonlayer}
	\begin{pgfonlayer}{edgelayer}
		\draw (1.center) to (0);
		\draw (0) to (2.center);
	\end{pgfonlayer}
\end{tikzpicture} \coloneqq \input{figures/petri/cap-reg-2} \qquad \input{figures/petri/cap-place} \coloneqq \input{figures/petri/cap-place-2}\]
The leftmost diagram models a buffer with capacity $c$, while the rightmost a place with capacity $c$. Since $\top= \BcounitT \BunitT$, it holds that $\begin{tikzpicture}
	\begin{pgfonlayer}{nodelayer}
		\node [style=reg] (0) at (0, 0) {$x$};
		\node [style=none] (1) at (-0.5, 0) {};
		\node [style=none] (2) at (0.5, 0) {};
		\node [style=none] (3) at (0.25, -0.25) {\scriptsize $\top$};
	\end{pgfonlayer}
	\begin{pgfonlayer}{edgelayer}
		\draw (1.center) to (0);
		\draw (0) to (2.center);
	\end{pgfonlayer}
\end{tikzpicture} \myeq{\SAIHP} \circuitXT$ and $\input{figures/petri/cap-place-top} \myeq{\SAIHP} \input{figures/petri/place}$. 

\medskip


By choosing an ordering on places and transitions, the functions $\pre{-},\post{-}\from T\to \posreals^P$ can be regarded as $\posreals$-matrices of type $|T| \to |P|$ and thus can be encoded as $\FC$ circuits, hereafter denoted by respectively $W^-$ and $W^+$. The ordering on $P$ also makes the boundary function $b$ a vector $\begin{pmatrix*} c_1 \\ \vdots \\ c_{|P|} \end{pmatrix*}$ in $(\posreals\cup \{\top\})^{|P|}$: we write \input{figures/petri/cap-place-b} for \input{figures/petri/cap-place-b-2}.
Any bounded Continuous Petri net $\mathcal{P}$ can be encoded as the following circuit $d_{\mathcal{P}} \colon 0 \to 0$ in $\CSFG$.
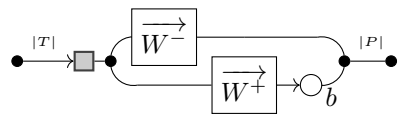
$$
\tikzset{x=1em, y=2.1ex}
\InputIfFileExists{encoding-Petri-a.tikz}{}{\input{./tikz/encoding-Petri-a.tikz}}
\tikzset{x=1em, y=1.5ex}
$$ 
%
It is easy to show that $\mathcal{P}$ and $d_\mathcal{P}$ have the same semantics.

\begin{proposition}\label{prop:encPT1}
For all bounded continuous Petri net $\mathcal{P}$, $\osem{\mathcal{P}} \sim \dsems{d_\mathcal{P}}$.
\end{proposition}
\begin{proof}[Proof of Proposition~\ref{prop:encPT1}]
In order to compute $\dsems{d_\mathcal{P}}$, it is convenient to cut $d_\mathcal{P}$ in three parts. The leftmost part of $d_\mathcal{P}$ has the following semantics
\[ \dsems{
\tikzset{x=1em, y=2.1ex}
\InputIfFileExists{net-part-1.tikz}{}{\input{./tikz/net-part-1.tikz}}
\tikzset{x=1em, y=1.5ex}
} = \quad (0,\{(\nullvec, \,\nullvec, \, \nullvec, \, \begin{pmatrix*} t \\ t \end{pmatrix*}) \mid t \in \posreals^{|T|}\})\]
The central part
\[\dsems{
\tikzset{x=1em, y=2.1ex}
\InputIfFileExists{net-part-2.tikz}{}{\input{./tikz/net-part-2.tikz}}
\tikzset{x=1em, y=1.5ex}
} = \quad (0\{(\nullvec, \begin{pmatrix*} x_1 \\ x_2 \end{pmatrix*}, \nullvec, \,\begin{pmatrix*} y_1 \\ y_2 \end{pmatrix*}) \mid W^-x_1=y_1, \, W^+x_2=y_2\})\]
By definition of $\dsemsO$, the composition of the two semantics above is the pair \[(0,\{(\nullvec, \,\nullvec, \, \nullvec, \, \begin{pmatrix*} y_1 \\ y_2 \end{pmatrix*}) \mid \exists t \in \posreals^{|T|} \text{ s.t. } W^-t=y_1, \, W^+t=y_2 \})\]

The right-most part
\[\dsem{
\tikzset{x=1em, y=2.1ex}
\InputIfFileExists{net-part-3.tikz}{}{\input{./tikz/net-part-3.tikz}}
\tikzset{x=1em, y=1.5ex}
} =\quad \{ (y,\begin{pmatrix*} i \\ o \end{pmatrix*}, y-o+i, \nullvec) \mid i,o,y \in \posreals^{|T|},\, o \geq y\}\]
The semantics of rightmost part is the pair
\[\dsems{
\tikzset{x=1em, y=2.1ex}
\InputIfFileExists{net-part-3c.tikz}{}{\input{./tikz/net-part-3c.tikz}}
\tikzset{x=1em, y=1.5ex}
} =\quad (|P|, \{ (y,\begin{pmatrix*} o \\ i \end{pmatrix*}, y-o+i ,\nullvec ) \mid i,o,y \in \posreals^{|P|},\, o \geq y, \, y-o+i\leq b, \, y\leq b\})\]
By composing everything we obtain $(|P|, \{ (y,\nullvec, y-i+o ,\nullvec ) \mid \exists t \in \posreals^{|T|} \text{ s.t. }  i,o,y \in \posreals^{|P|},\, o \geq y, \, y-o+i\leq b, \, y\leq b, \, W^-t=o, \, W^+t=i\})$
that is
$\dsems{d_\mathcal{P}} = (|P|,\{ (y,\nullvec, z ,\nullvec ) \mid \exists t \in \posreals^{|T|} \text{ s.t. } W^-t \geq y, \, y-W^-t+W^+t=z\leq b, \, y\leq b\})$.

Since the equivalence is stated modulo $\sim$, then it is safe to fix an ordering on $P$ and $T$. Thus, rather than considering $(y,z)\in \osem{P}$ as functions in $\posreals^{P}$ they can be regarded as vectors in $\posreals^{|P|}$. One can thus conclude by observing that $y \to z$ if and only if there exists $t \in \posreals^{|T|}$ such that  $W^-t \geq y$, $y-W^-t+W^+t=z$ and $y,z\leq b$.
\end{proof}



%

\section{Conclusions and Future Work}\label{sec:conc}
We have introduced the theories of polyhedral cones and the one of polyhedra. In other words, we have identified suitable sets of generators and axioms for which we proved completeness and expressivity. As side results, we get an inductive definition of the notion of polar cone, as well as an understanding of Weyl-Minkowski theorem as a normal form result.

As shown by Example \ref{ex:fnet}, the theory of polyhedra allows us to  represent networks with bounded resources, not expressible in $\IH{}$, and to manipulate them as symbolic expressions.

Indeed the passage from linear relations to polyhedra is a reflection of the fact that, operationally, we are able to consider several patterns of computations important in computer science, as opposed to purely linear patterns, traditionally studied in system/control theory.

For instance, as shown in \S\ref{app:state}, the addition to $\AIHP{}$ of a single generator, $\circuitXT$, directly gives us a concurrent extension of the signal flow calculus \cite{BaezErbele-CategoriesInControl,Bonchi2015}, introduced as a compositional account for linear dynamical systems, that is expressive enough to encode continuous Petri nets \cite{DavidAlla10}.
%



\bibliography{bibliography}

\begin{thebibliography}{10}

\bibitem{10.5555/137406}
Ravindra~K. Ahuja, Thomas~L. Magnanti, and James~B. Orlin.
\newblock {\em Network Flows: Theory, Algorithms, and Applications}.
\newblock Prentice-Hall, Inc., USA, 1993.

\bibitem{BaezErbele-CategoriesInControl}
John Baez and Jason Erbele.
\newblock Categories in control.
\newblock {\em Theory and Applications of Categories}, 30:836--881, 2015.

\bibitem{Baez2014}
John~C. Baez.
\newblock Network theory.
\newblock \url{http://math.ucr.edu/home/baez/networks/}, 2014.

\bibitem{baez2015compositional}
John~C Baez and Brendan Fong.
\newblock A compositional framework for passive linear networks.
\newblock {\em arXiv preprint arXiv:1504.05625}, 2015.

\bibitem{BonchiHPS17}
Filippo Bonchi, Joshua Holland, Dusko Pavlovic, and Pawe{\l} Soboci{\'{n}}ski.
\newblock Refinement for signal flow graphs.
\newblock In {\em 28th International Conference on Concurrency Theory, {CONCUR}
  2017, September 5-8, 2017, Berlin, Germany}, pages 24:1--24:16, 2017.
\newblock \href {https://doi.org/10.4230/LIPIcs.CONCUR.2017.24}
  {\path{doi:10.4230/LIPIcs.CONCUR.2017.24}}.

\bibitem{DBLP:journals/pacmpl/BonchiHPSZ19}
Filippo Bonchi, Joshua Holland, Robin Piedeleu, Pawel Sobocinski, and Fabio
  Zanasi.
\newblock Diagrammatic algebra: from linear to concurrent systems.
\newblock {\em Proc. {ACM} Program. Lang.}, 3({POPL}):25:1--25:28, 2019.
\newblock \href {https://doi.org/10.1145/3290338} {\path{doi:10.1145/3290338}}.

\bibitem{BHPSZ-popl19}
Filippo Bonchi, Joshua Holland, Robin Piedeleu, Pawe{\l} Soboci{\'n}ski, and
  Fabio Zanasi.
\newblock Diagrammatic algebra: from linear to concurrent systems.
\newblock {\em Proceedings of the 46th ACM SIGPLAN Symposium on Principles of
  Programming Languages (POPL)}, 3:1--28, 2019.

\bibitem{BonchiPSZ19}
Filippo Bonchi, Robin Piedeleu, Pawe{\l} Soboci{\'{n}}ski, and Fabio Zanasi.
\newblock Graphical affine algebra.
\newblock In {\em Proceedings of the 34th Annual {ACM/IEEE} Symposium on Logic
  in Computer Science (LICS)}, pages 1--12, 2019.

\bibitem{GCQ}
Filippo Bonchi, Jens Seeber, and Pawel Sobocinski.
\newblock {Graphical Conjunctive Queries}.
\newblock In Dan Ghica and Achim Jung, editors, {\em 27th EACSL Annual
  Conference on Computer Science Logic (CSL 2018)}, volume 119 of {\em Leibniz
  International Proceedings in Informatics (LIPIcs)}, pages 13:1--13:23,
  Dagstuhl, Germany, 2018. Schloss Dagstuhl--Leibniz-Zentrum fuer Informatik.
\newblock \href {https://doi.org/10.4230/LIPIcs.CSL.2018.13}
  {\path{doi:10.4230/LIPIcs.CSL.2018.13}}.

\bibitem{Bonchi2014b}
Filippo Bonchi, Pawe{\l} Soboci{\'n}ski, and Fabio Zanasi.
\newblock A categorical semantics of signal flow graphs.
\newblock In {\em Proceedings of the 25th International Conference on
  Concurrency Theory (CONCUR)}, pages 435--450. Springer, 2014.

\bibitem{Bonchi2015}
Filippo Bonchi, Pawel Sobocinski, and Fabio Zanasi.
\newblock Full abstraction for signal flow graphs.
\newblock In {\em Proceedings of the 42nd Annual ACM SIGPLAN Symposium on
  Principles of Programming Languages (POPL)}, pages 515--526, 2015.

\bibitem{BonchiSZ17}
Filippo Bonchi, Pawel Sobocinski, and Fabio Zanasi.
\newblock The calculus of signal flow diagrams {I:} linear relations on
  streams.
\newblock {\em Information and Computation}, 252:2--29, 2017.

\bibitem{Carboni1987}
Aurelio Carboni and R.~F.~C. Walters.
\newblock Cartesian bicategories {I}.
\newblock {\em Journal of Pure and Applied Algebra}, 49:11--32, 1987.

\bibitem{CoeckeDuncanZX2011}
Bob Coecke and Ross Duncan.
\newblock Interacting quantum observables: categorical algebra and
  diagrammatics.
\newblock {\em New Journal of Physics}, 13(4):043016, 2011.

\bibitem{Coecke2012}
Bob Coecke, Ross Duncan, Aleks Kissinger, and Quanlong Wang.
\newblock Strong complementarity and non-locality in categorical quantum
  mechanics.
\newblock In {\em LiCS 2012}, pages 245--254, 2012.

\bibitem{Coecke2017}
Bob Coecke and Aleks Kissinger.
\newblock {\em Picturing Quantum Processes - A first course in Quantum Theory
  and Diagrammatic Reasoning}.
\newblock Cambridge University Press, 2017.

\bibitem{DBLP:conf/popl/CousotC77}
Patrick Cousot and Radhia Cousot.
\newblock Abstract interpretation: {A} unified lattice model for static
  analysis of programs by construction or approximation of fixpoints.
\newblock In Robert~M. Graham, Michael~A. Harrison, and Ravi Sethi, editors,
  {\em Conference Record of the Fourth {ACM} Symposium on Principles of
  Programming Languages, Los Angeles, California, USA, January 1977}, pages
  238--252. {ACM}, 1977.
\newblock \href {https://doi.org/10.1145/512950.512973}
  {\path{doi:10.1145/512950.512973}}.

\bibitem{DavidAlla10}
Ren\'{e} David and Hassane Alla.
\newblock {\em Discrete, Continuous, and Hybrid Petri Nets}.
\newblock Springer, Berlin, 2 edition, 2010.
\newblock \href {https://doi.org/10.1007/978-3-642-10669-9}
  {\path{doi:10.1007/978-3-642-10669-9}}.

\bibitem{DBLP:journals/corr/abs-2009-06836}
Brendan Fong and David Spivak.
\newblock String diagrams for regular logic (extended abstract).
\newblock In John Baez and Bob Coecke, editors, {\em Applied Category Theory
  2019}, volume 323 of {\em Electronic Proceedings in Theoretical Computer
  Science}, page 196–229. Open Publishing Association, Sep 2020.
\newblock \href {https://doi.org/10.4204/eptcs.323.14}
  {\path{doi:10.4204/eptcs.323.14}}.

\bibitem{Ghica2016}
Dan~R Ghica and Achim Jung.
\newblock Categorical semantics of digital circuits.
\newblock In {\em Proceedings of the 16th Conference on Formal Methods in
  Computer-Aided Design (FMCAD)}, pages 41--48, 2016.

\bibitem{Haydon2020}
Nathan Haydon and Pawe{\l} Soboci\'{n}ski.
\newblock Compositional diagrammatic first-order logic.
\newblock In {\em 11th International Conference on the Theory and Application
  of Diagrams (DIAGRAMS 2020)}, 2020.

\bibitem{hoshino2014memoryful}
Naohiko Hoshino, Koko Muroya, and Ichiro Hasuo.
\newblock Memoryful geometry of interaction: from coalgebraic components to
  algebraic effects.
\newblock In {\em Proceedings of the Joint Meeting of the Twenty-Third EACSL
  Annual Conference on Computer Science Logic (CSL) and the Twenty-Ninth Annual
  ACM/IEEE Symposium on Logic in Computer Science (LICS)}, page~52. ACM, 2014.

\bibitem{JacobsZ18}
Bart Jacobs and Fabio Zanasi.
\newblock The logical essentials of bayesian reasoning.
\newblock {\em CoRR}, abs/1804.01193, 2018.
\newblock URL: \url{http://arxiv.org/abs/1804.01193}, \href
  {http://arxiv.org/abs/1804.01193} {\path{arXiv:1804.01193}}.

\bibitem{BaezCoya-propsnetworktheory}
Franciscus~Rebro John C.~Baez, Brandon~Coya.
\newblock Props in network theory.
\newblock {\em CoRR}, abs/1707.08321, 2017.
\newblock URL: \url{http://arxiv.org/abs/1707.08321}, \href
  {http://arxiv.org/abs/1707.08321} {\path{arXiv:1707.08321}}.

\bibitem{katis1997bicategories}
P~Katis, N~Sabadini, and RFC Walters.
\newblock Bicategories of processes.
\newblock {\em Journal of Pure and Applied Algebra}, 115(2):141--178, 1997.

\bibitem{Lack2004a}
Stephen Lack.
\newblock Composing {PROPs}.
\newblock {\em Theory and Application of Categories}, 13(9):147--163, 2004.

\bibitem{Lafont2003}
Yves Lafont.
\newblock Towards an algebraic theory of {B}oolean circuits.
\newblock {\em Journal of Pure and Applied Algebra}, 184(2--3):257--310, 2003.

\bibitem{MacLane1965}
Saunders {Mac Lane}.
\newblock Categorical algebra.
\newblock {\em Bulletin of the American Mathematical Society}, 71:40--106,
  1965.

\bibitem{mason1953feedback}
Samuel~J Mason.
\newblock {\em Feedback Theory: I. Some Properties of Signal Flow Graphs}.
\newblock MIT Research Laboratory of Electronics, 1953.

\bibitem{Piedeleu2021}
Robin Piedeleu and Fabio Zanasi.
\newblock A string diagrammatic axiomatisation of finite-state automata.
\newblock In {\em FoSSaCS 2021}, 2021.

\bibitem{Selinger2009}
Peter Selinger.
\newblock A survey of graphical languages for monoidal categories.
\newblock {\em Springer Lecture Notes in Physics}, 13(813):289--355, 2011.

\bibitem{ZanasiThesis}
Fabio Zanasi.
\newblock {\em Interacting Hopf Algebras: the theory of linear systems}.
\newblock PhD thesis, {E}cole Normale Sup\'{e}rieure de Lyon, 2015.

\end{thebibliography}

\appendix

\section{More details on the props $\propCat{PC}_{\mathsf k}$ and $\propCat{P}_{\mathsf k}$}\label{app:standardproperties}
In this appendix we show that monoidal product and composition as given in $\propCat{PC}_{\mathsf k}$ and $\propCat{P}_{\mathsf k}$ are well defined. For $\LinRel{\field}$, see e.g.~\cite{ZanasiThesis}. Our proof is based on four well-known operations on polyhedra and polyhedral cones. We illustrate the proof for polyhedral cones, as the one for polyhedra is completely analogous.

\begin{lemma}[Projection]
    \label{lm:polyproj}
    Let $C \subseteq \mathsf{k}^n$ be a polyhedral cone, then \[C' = \{ v \mid \exists x . \begin{pmatrix} v \\ x \end{pmatrix} \in C \}\] is a polyhedral cone.
\end{lemma}

\begin{lemma}[Intersection]
    \label{lm:polyintr}
    Let $C_1, C_2 \subseteq \mathsf{k}^n$ be polyhedral cones, then \[C_1 \cap C_2  = \{ v \mid v \in C_1, v \in C_2 \}\] is a polyhedral cone.
\end{lemma}


\begin{lemma}[Extension]
    \label{lm:polyext}
    Let $C \subseteq \mathsf{k}^n$ be a polyhedral cone, then \[ \tilde{C} = \{ \begin{pmatrix*} v \\ x \end{pmatrix*} \in \mathsf{k}^{n+1} \mid v \in C \} \] is a polyhedral cone.

We say that $\tilde{C}_n$ is the $n$-extension of $C$ to denote $n$ repeated applications of $\tilde{(\cdot)}$ to $C$.
\end{lemma}

\begin{lemma}[Permutation]
    \label{lm:polyperm}
    Let $C \subseteq \mathsf{k}^n$ be a polyhedral cone and $\pi \colon n \rightarrow n$ a permutation, then \[ \pi(C) = \{ v \mid \pi(v) \in C \} = \{ \pi(v) \mid v \in C \} \] is a polyhedral cone.
\end{lemma}

With this four operations, it is easy to show that composition and tensor product of two polyhedral cones are actually polyhedral cones.

\begin{lemma}[Tensor product]
    \label{lm:polytensor}
    Let $C_1 \subseteq \mathsf{k}^{n + m}$, $C_2 \subseteq \mathsf{k}^{p + l}$ be polyhedral cones, then \[C_1 \oplus C_2  = \{ (\begin{pmatrix*} v_1 \\ v_2 \end{pmatrix*}, \begin{pmatrix*} u_1 \\ u_2 \end{pmatrix*}) \in \mathsf{k}^{n+m} \times \mathsf{k}^{p+l} \mid (v_1, u_1) \in C_1, (v_2, u_2) \in C_2 \} \] is a polyhedral cone.
\end{lemma}
\begin{proof}
Let $C_1 = \{ x \in \mathsf{k}^{n + m} \mid A_1x \geq 0 \}$ and $C_2 = \{ x \in \mathsf{k}^{p + l} \mid A_2x \geq 0 \}$ for some matrices $A_1, A_2$. 

Then $C_1 \oplus C_2 = \{ x \in \mathsf{k}^{(n+m) + (p+l)} \mid \begin{pmatrix*} A_1 & 0 \\ 0 & A_2 \end{pmatrix*}x \geq 0 \}$
\end{proof}

\begin{lemma}[Sequential composition]
    \label{lm:polyseqc}
    Let $C_1 \subseteq \mathsf{k}^{n} \times \mathsf{k}^{m}$, $C_2 \subseteq \mathsf{k}^{m} \times \mathsf{k}^{p}$ be polyhedral cones, then \[C_1 \scolon C_2  = \{ (u, v) \in \mathsf{k}^{n} \times \mathsf{k}^{p} \mid \exists w \in \mathsf{k}^m . (u, w) \in C_1, (w, v) \in C_2 \} \] is a polyhedral cone.
\end{lemma}
\begin{proof}
Let $C_1 = \{ x \in \mathsf{k}^{n} \times \mathsf{k}^{m} \mid A_1x \geq 0 \}$ and $C_2 = \{ x \in \mathsf{k}^{m} \times \mathsf{k}^{p} \mid A_2x \geq 0 \}$ for some matrices $A_1, A_2$. 

Let $\tilde{C_1}_p \subseteq \mathsf{k}^n \times \mathsf{k}^m \times \mathsf{k}^p, \tilde{C_2}_n \subseteq \mathsf{k}^m \times \mathsf{k}^p \times \mathsf{k}^n$ and $\pi \colon m + p + n \rightarrow m + p + n$ which swaps the last $n$ coordinates with the first $m + p$ coordinates.

Then $C_1 \scolon C_2 = \{ (u, v) \mid \exists w . (u, w, v) \in \tilde{C_1}_p \wedge \pi(\tilde{C_2}_n) \}$ is a polyhedral cone.
\end{proof}

\section{Proofs of Section~\ref{sec:polyc}}

Before proving the soundness of the Fourier-Motzkin elimination within the theory $\IHP{\field}$, it is convenient to state a few derived laws.

\begin{lemma}
\input{figures/polyc/derived-laws}
\end{lemma}
\begin{proof}
\input{figures/polyc/derived-laws-proof}
\end{proof}

\begin{lemma}[Generalised spider]\label{lm:nmspider}
    \begin{equation*}
        \input{figures/polyc/gen-spider-def}
    \end{equation*}
\end{lemma}
\begin{proof}
\input{figures/polyc/gen-spider-proof}
\end{proof}

\subsection{An axiomatic proof of Fourier Motzkin Elimination}\label{app:FMelimination}
We have now all the tools to provide an axiomatic proof of Proposition~\ref{prop:FMElimitation}.

\begin{proof}[Proof of Proposition~\ref{prop:FMElimitation}]
Consider the polyhedral cone $C = \{ x \mid Ax \geq 0 \}$, which is represented by the following diagram of $\CDP$:
\input{figures/polyc/fm2.tex}
We assume that the circuit representing $A$ is in matrix form and it is made of the black comonoid structure followed by a circuit $A'$ and the white monoid structure.

Without loss of generalities, we consider the first $h$ rows of $A$ to have positive coefficients for $x_1$, the other $q$ rows to have negative coefficients for $x_1$ and the rest have a null coefficient.

\begin{remark}
    In the illustration of the algorithm that follows, all of the scalars $\scalarT$ will be such that $k \geq 0$. Negative scalars will be rewritten as $\begin{tikzpicture}
            \begin{pgfonlayer}{nodelayer}
                \node [style=antipode] (0) at (-0.25, 0) {};
                \node [style=reg] (1) at (0.5, 0) {$k$};
                \node [style=none] (2) at (-1, 0) {};
                \node [style=none] (3) at (1.25, 0) {};
            \end{pgfonlayer}
            \begin{pgfonlayer}{edgelayer}
                \draw (2.center) to (0);
                \draw (0) to (1);
                \draw (1) to (3.center);
            \end{pgfonlayer}
        \end{tikzpicture}
        $.
\end{remark}

Since the diagram is in matrix form, we can rewrite it to make explicit the coefficients multiplying $x_1$ in each row/inequality:

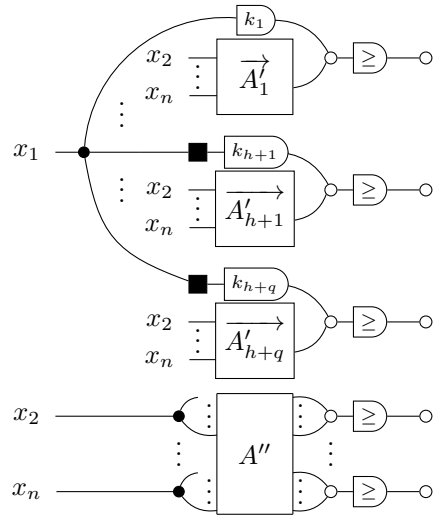
\input{figures/polyc/fm3}

\begin{remark}
    In the diagram above we abused some notation to better visualize it. In particular, notice that whenever two wires are labelled with the same $x_i$ they are linked together via the black comonoid structure.
\end{remark}

By virtue of the derived law~\eqref{eq:l1}, we can rewrite the diagram as:
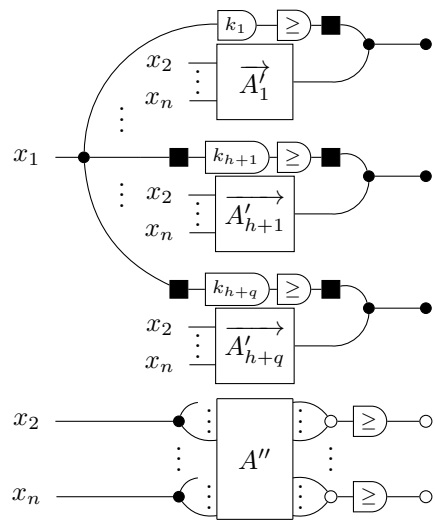
\input{figures/polyc/fm4}
Now, by applying the axiom (P5) the diagram can be further rewritten as:

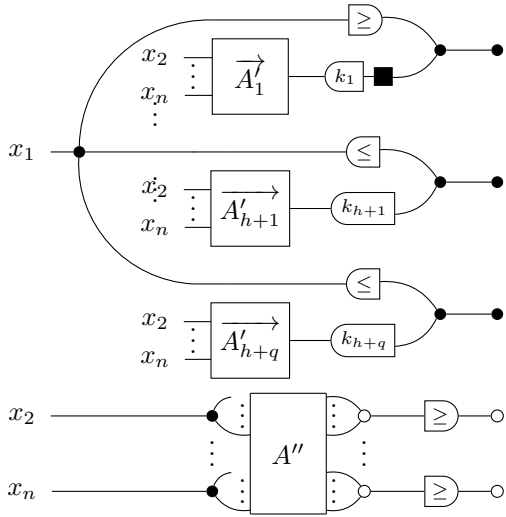
\input{figures/polyc/fm5}

By virtue of the compact closed structure, we \textit{bend} the wire linking each $A'_i$ as follows:
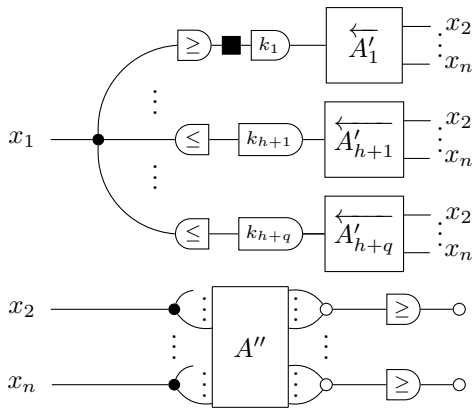
\input{figures/polyc/fm6}
 
Now we are ready to \textit{eliminate} the variable $x_1$ by \textit{closing} the corresponding wire with $\Bunit$:

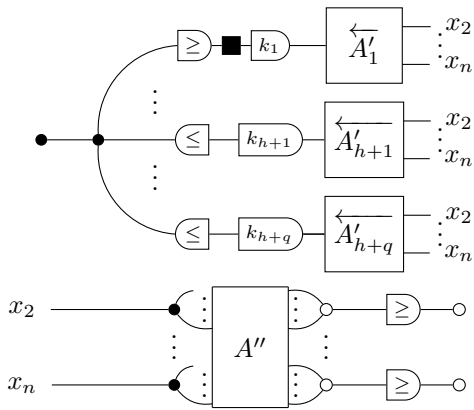
    \input{figures/polyc/fm7}
    The diagram results in a bunch of bent wires which we can bend again by virtue of the compact closed structure. In particular we focus on the last $q-h$ wires (i.e. those with the $\leq$ relation):
    \input{figures/polyc/fm8}
Now we can apply Lemma~\ref{lm:nmspider} to obtain the diagram in which each inequality has been combined in the very same way it happens in the Fourier-Motzkin algorithm
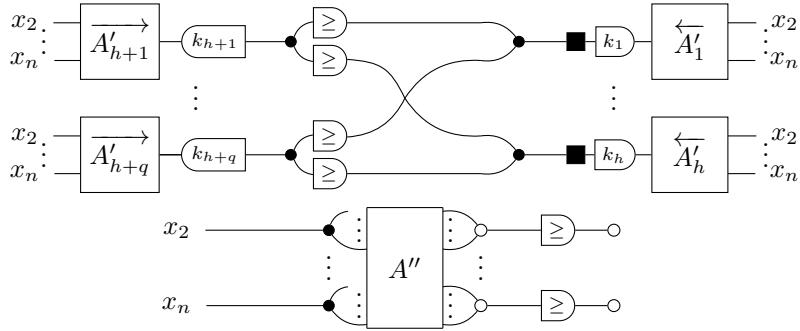
\input{figures/polyc/fm9}

By exploiting the derived law of $\IH{\mathsf k}$: $\begin{tikzpicture}
	\begin{pgfonlayer}{nodelayer}
		\node [style=none] (131) at (-1.25, 0.25) {};
		\node [style=none] (132) at (-1.25, -0.25) {};
		\node [style=black] (133) at (-0.75, 0) {};
		\node [style=small box] (176) at (0, 0) {$A$};
		\node [style=none] (177) at (0.75, 0) {};
	\end{pgfonlayer}
	\begin{pgfonlayer}{edgelayer}
		\draw [bend left] (131.center) to (133);
		\draw [bend right] (132.center) to (133);
		\draw (176) to (133);
		\draw (177.center) to (176);
	\end{pgfonlayer}
\end{tikzpicture}
\myeq{\IH{}}  \begin{tikzpicture}
	\begin{pgfonlayer}{nodelayer}
		\node [style=small box] (131) at (-0.5, 0.5) {$A$};
		\node [style=small box] (132) at (-0.5, -0.5) {$A$};
		\node [style=black] (133) at (0, 0) {};
		\node [style=none] (177) at (0.5, 0) {};
		\node [style=none] (178) at (-1, 0.5) {};
		\node [style=none] (179) at (-1, -0.5) {};
	\end{pgfonlayer}
	\begin{pgfonlayer}{edgelayer}
		\draw [bend left] (131) to (133);
		\draw [bend right] (132) to (133);
		\draw (133) to (177.center);
		\draw (178.center) to (131);
		\draw (179.center) to (132);
	\end{pgfonlayer}
\end{tikzpicture}$
we can \textit{replicate} each $A_i$, obtaining the following diagram:
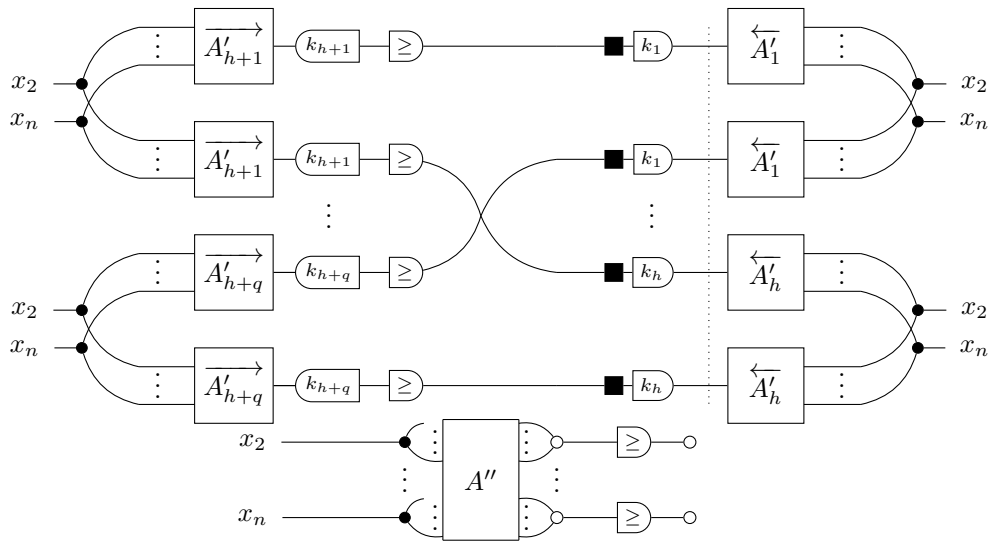
\input{figures/polyc/fm10}

Exploiting again the compact closed structure, we perform a bending of the wires at the point marked by the dotted line. 

What we obtain is the diagram representing $hq$ inequalities:
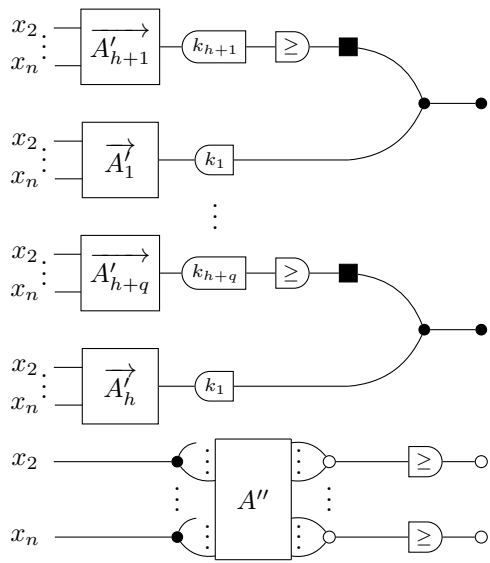
\input{figures/polyc/fm11}
Finally, by applying again the law~\eqref{eq:l1} we obtain the desired result 
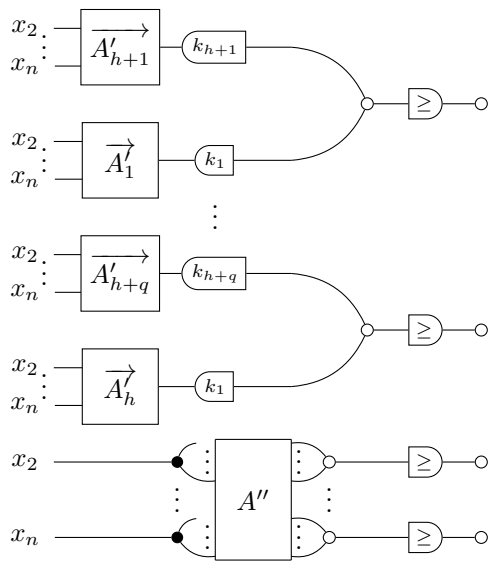
\input{figures/polyc/fm12}
\end{proof}

\subsection{Proof of the First Normal Form Theorem}\label{app:polynf}
The Fourier Motzkin elimination is essential to prove Theorem~\ref{th:polynf}.
\begin{proof}[Proof of Theorem~\ref{th:polynf}]
\input{figures/polyc/pnf-proof}
\end{proof}

\subsection{Proofs for the polar operator}
In this appendix we report the proof of Proposition~\ref{th:polarpolar}. First we prove the third point, then the second and finally the first, which require some more work.

\begin{proof}[Proof of Proposition~\ref{th:polarpolar}.3]
    The proof goes by induction on the structure of $\FC$.
    For each generator $c$ of $\FC$, it is trivial by definition that $c^\circ$ is in $\FCop$.
    The inductive cases, are also trivial. For instance, $(c;d)^\circ$ is by definition $c^\circ ; d^\circ$. By inductive hypothesis both $c^\circ$ and $d^\circ$ are in $\FCop$, so $c^\circ ; d^\circ$ is in $\FCop$.
\end{proof}

\begin{proof}[Proof of Proposition~\ref{th:polarpolar}.2]
    \input{figures/polyc/polarpolar-proof}
\end{proof}

\begin{proof}[Proof of Proposition~\ref{th:polarpolar}.1]
It is enough to check that for all equations $l=r$ and inequations $l\subseteq r$ in Figures~\ref{fig:ih} and~\ref{fig:axiom:ihp}, it holds that $l^\circ \myeq{\IHP{}} r^\circ$ and $r^\circ \mysubeq{\IHP{}} l^\circ$. For the axioms in Figures~\ref{fig:ih}, this is completely trivial. For those in Figures~\ref{fig:axiom:ihp} the only challenging case is the spider axiom. We report its proof below.

    \input{figures/polyc/whitespider-proof}
\end{proof}

\subsection{Proof of the Completeness Theorem}\label{app:completenessproofs}
In this appendix we illustrate the last steps to prove the Theorem~\ref{th:compl2}.

\begin{definition}
    An arrow $c \colon n \rightarrow m$ of $\FC$ is in \textit{non-negative matrix form} if it is in matrix form and all of the scalars $\scalar$ appearing 
    in it are such that $k \geq 0$.
\end{definition}

\begin{lemma}
    \label{lm:movegeq}
    For each arrow $c \colon n \rightarrow m$ of $\FC$ in non-negative matrix form
    \[ \input{figures/polyc/movegeq-def} \]
\end{lemma}
\begin{proof}
    \input{figures/polyc/movegeq-proof}
\end{proof}

\begin{lemma} \label{lm:movegeq2}
    For each arrow $c \colon n \rightarrow m$ of $\FC$ in non-negative matrix form
    \[ \input{figures/polyc/movegeq2-def} \]
\end{lemma}
\begin{proof}
    It follows from Lemma~\ref{lm:movegeq} and from the following derived law of $\IH{}$
    \[ \input{figures/polyc/movegeq2-proof} \]
\end{proof}

\begin{lemma}
    \label{lm:cones}
    Given two arrows $n \xrightarrow{A} m$ and $l \xrightarrow{B} m$ of $\Mat_{\mathsf k}$ such that \[\mathsf{cone}(V_A) \subseteq \mathsf{cone}(V_B) \] where $V_A$ and $V_B$ are the (finite) sets of vectors 
    consisting, respectively, of the columns of $A$ and $B$, there exists a matrix $C$ with non-negative coefficients, such that the diagram
    \[
    \begin{tikzcd}
    n \arrow[rd, "A"] \arrow[d, "C"'] &   \\
    l \arrow[r, "B"]                  & m
    \end{tikzcd} \]
    commutes in $\Mat_{\mathsf k}$.
\end{lemma}
\begin{proof}
    Notice that $V_A \subseteq \mathsf{cone}(V_A)$ since each vector $a \in V_A$ can be obtained as a conical combination of the others where $a$ is multiplied by $1$ and the others by $0$.

    This means that if $\mathsf{cone}(V_A) \subseteq \mathsf{cone}(V_B)$ then $V_A \subseteq \mathsf{cone}(V_B)$.
    
    Now, if $V_A \subseteq \mathsf{cone}(V_B)$, each $a_i \in V_A$ can be obtained as \[a_i = \alpha_1^ib_1 + ... + \alpha_n^ib_n\], where $b_i \in V_B$ and each $\alpha_j^i \geq 0$.

    In other words, the matrix $A$ can be obtained as a multiplication of $B$ by a matrix $C$ with non-negative coefficients:

    \[ A = B \times C \]

    where the columns $c_i$ of $C$ are given by the coefficients $\alpha_1^i \ldots \alpha_n^i$, i.e. $c_i = \begin{pmatrix*} \alpha_1^i \\ \myvdots \\ \alpha_n^i \end{pmatrix*}$.
\end{proof}

Now we have all of the right tools to prove that whenever two finitely generated cones (and equivalently, polyhedral cones) are (denotationally) equal, then also their string diagrammatic representation as arrows of $\CDP$ are equal modulo the axioms of $\IHP{}$.

\begin{lemma}
    \label{th:compl1}
    Given two arrows $A, B \colon n \rightarrow m$ of $\FC$, such that \[\mathsf{cone}(V_A) \subseteq \mathsf{cone}(V_B) \] where $V_A$ and $V_B$ are the (finite) sets of vectors 
    consisting, respectively, of the columns of (the matrices denoting) $A$ and $B$, the following holds: 
    \[  \input{figures/polyc/compl1-def} \]
\end{lemma}
\begin{proof}
    \input{figures/polyc/compl1-proof}
\end{proof}

With the above lemma and the finitely generated normal form (Theorem~\ref{co:fingennf}) it is easy to prove completeness.

\begin{proof}[Proof of Theorem~\ref{th:compl2}]
    By Theorem~\ref{co:fingennf}, there exist $c', d'$ in finitely generated normal form, such that $c \myeq{\IHP{}} c'$ and $d \myeq{\IHP{}} d'$. 

    Thus $\dsemp{c'} \subseteq \dsemp{d'}$ and we can apply Lemma~\ref{th:compl1}.
\end{proof}

\subsection{Proof for expressiveness}

\begin{proof}[Proof of Proposition~\ref{th:ihpiso}]
\input{figures/poly/ihpiso-proof}
%
For the viceversa we proceed by induction on $\CDP$. For the base cases, we have to consider the generators of $\CDP$. The semantics of all the generators in~\eqref{eq:SFcalculusSyntax1}, \eqref{eq:SFcalculusSyntax2} are linear relations and thus, polyhedral cones. For $\greq$
observe that
    \begin{align*}
        \dsemp{\greq} &=  \\
        &= \{ (x,y) \mid x, y \in \mathsf{k}^n, x \geq y \} \\
        &= \{ (x,y) \mid x, y \in \mathsf{k}^n, x - y \geq 0 \} \\
        &= \{ (x,y) \mid x, y \in \mathsf{k}^n, \begin{pmatrix} 1 & -1 \end{pmatrix} \begin{pmatrix} x \\ y \end{pmatrix} \geq 0 \}
    \end{align*}
    which is a polyhedral cone by definition.

For the inductive case, we consider $c \scolon d$ since the proof for  $c \oplus d$ is completely analogous. By definition of $\dsemO$, we have that $\dsem{c \scolon d} = \dsem{c} \scolon \dsem{d}$. By inductive hypothesis both $\dsem{c}$ and $\dsem{d}$ are polyhedral cones. Since the composition of polyhedral cones is again a polyhedral cone, we can conclude that $\dsem{c \scolon d}$ is a polyhedral cone.
\end{proof}


%
%
%
%
%
%
%
\section{More details on the props $\IHP{\field}$ and $\AIHP{\field}$}
In this appendix we investigate the categorical properties of $\IHP{\field}$ and $\AIHP{\field}$ that we omitted in the main text for lack of space. More precisely, we are going to show that $\IHP{\field}$ and $\AIHP{\field}$ form, respectively, an Abelian bicategory and a Cartesian bicategory in the sense of~\cite{Carboni1987}.

We first recall from~\cite{BonchiHPS17} the following result.
\begin{lemma}\label{lemma:adjoint} The following holds in $\IH{\field}$.
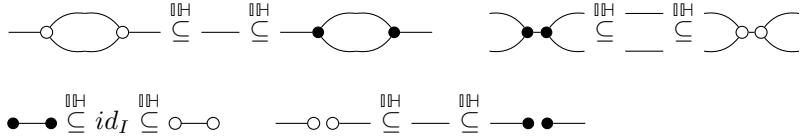
\begin{equation*}
    \begin{tikzpicture}
	\begin{pgfonlayer}{nodelayer}
		\node [style=none] (0) at (0, -0.25) {};
		\node [style=none] (1) at (0, 0.25) {};
		\node [style=white] (2) at (0.5, 0) {};
		\node [style=none] (3) at (1, 0) {};
		\node [style=white] (4) at (-0.5, 0) {};
		\node [style=none] (5) at (-1, 0) {};
	\end{pgfonlayer}
	\begin{pgfonlayer}{edgelayer}
		\draw [bend left] (4) to (1.center);
		\draw [bend left] (1.center) to (2);
		\draw [bend right] (0.center) to (2);
		\draw [bend left] (0.center) to (4);
		\draw (2) to (3.center);
		\draw (5.center) to (4);
	\end{pgfonlayer}
\end{tikzpicture}
\mysubeq{\IH{}}
\begin{tikzpicture}
	\begin{pgfonlayer}{nodelayer}
		\node [style=none] (3) at (0.25, 0) {};
		\node [style=none] (5) at (-0.25, 0) {};
	\end{pgfonlayer}
	\begin{pgfonlayer}{edgelayer}
		\draw (5.center) to (3.center);
	\end{pgfonlayer}
\end{tikzpicture}
\mysubeq{\IH{}}
\begin{tikzpicture}
	\begin{pgfonlayer}{nodelayer}
		\node [style=none] (0) at (0, -0.25) {};
		\node [style=none] (1) at (0, 0.25) {};
		\node [style=black] (2) at (0.5, 0) {};
		\node [style=none] (3) at (1, 0) {};
		\node [style=black] (4) at (-0.5, 0) {};
		\node [style=none] (5) at (-1, 0) {};
	\end{pgfonlayer}
	\begin{pgfonlayer}{edgelayer}
		\draw [bend left] (4) to (1.center);
		\draw [bend left] (1.center) to (2);
		\draw [bend right] (0.center) to (2);
		\draw [bend left] (0.center) to (4);
		\draw (2) to (3.center);
		\draw (5.center) to (4);
	\end{pgfonlayer}
\end{tikzpicture}
\qquad
\begin{tikzpicture}
	\begin{pgfonlayer}{nodelayer}
		\node [style=black] (3) at (-0.25, 0) {};
		\node [style=none] (5) at (-0.75, -0.25) {};
		\node [style=none] (6) at (-0.75, 0.25) {};
		\node [style=black] (7) at (0, 0) {};
		\node [style=none] (8) at (0.5, 0.25) {};
		\node [style=none] (9) at (0.5, -0.25) {};
	\end{pgfonlayer}
	\begin{pgfonlayer}{edgelayer}
		\draw [bend right] (5.center) to (3);
		\draw (3) to (7);
		\draw [bend right] (7) to (9.center);
		\draw [bend left] (7) to (8.center);
		\draw [bend right] (3) to (6.center);
	\end{pgfonlayer}
\end{tikzpicture}
\mysubeq{\IH{}}
\begin{tikzpicture}
	\begin{pgfonlayer}{nodelayer}
		\node [style=none] (5) at (-0.25, -0.25) {};
		\node [style=none] (6) at (-0.25, 0.25) {};
		\node [style=none] (8) at (0.25, 0.25) {};
		\node [style=none] (9) at (0.25, -0.25) {};
	\end{pgfonlayer}
	\begin{pgfonlayer}{edgelayer}
		\draw (6.center) to (8.center);
		\draw (9.center) to (5.center);
	\end{pgfonlayer}
\end{tikzpicture}
\mysubeq{\IH{}}
\begin{tikzpicture}
	\begin{pgfonlayer}{nodelayer}
		\node [style=white] (3) at (-0.25, 0) {};
		\node [style=none] (5) at (-0.75, -0.25) {};
		\node [style=none] (6) at (-0.75, 0.25) {};
		\node [style=white] (7) at (0, 0) {};
		\node [style=none] (8) at (0.5, 0.25) {};
		\node [style=none] (9) at (0.5, -0.25) {};
	\end{pgfonlayer}
	\begin{pgfonlayer}{edgelayer}
		\draw [bend right] (5.center) to (3);
		\draw (3) to (7);
		\draw [bend right] (7) to (9.center);
		\draw [bend left] (7) to (8.center);
		\draw [bend right] (3) to (6.center);
	\end{pgfonlayer}
\end{tikzpicture}
\end{equation*}
\begin{equation*}
\begin{tikzpicture}
	\begin{pgfonlayer}{nodelayer}
		\node [style=black] (6) at (-0.25, 0) {};
		\node [style=black] (8) at (0.25, 0) {};
	\end{pgfonlayer}
	\begin{pgfonlayer}{edgelayer}
		\draw (6.center) to (8.center);
	\end{pgfonlayer}
\end{tikzpicture}
\mysubeq{\IH{}}
id_I
\mysubeq{\IH{}}
\begin{tikzpicture}
	\begin{pgfonlayer}{nodelayer}
		\node [style=white] (6) at (-0.25, 0) {};
		\node [style=white] (8) at (0.25, 0) {};
	\end{pgfonlayer}
	\begin{pgfonlayer}{edgelayer}
		\draw (6.center) to (8.center);
	\end{pgfonlayer}
\end{tikzpicture}
\qquad
\begin{tikzpicture}
	\begin{pgfonlayer}{nodelayer}
		\node [style=none] (6) at (-0.75, 0) {};
		\node [style=white] (8) at (-0.25, 0) {};
		\node [style=white] (9) at (0, 0) {};
		\node [style=none] (10) at (0.5, 0) {};
	\end{pgfonlayer}
	\begin{pgfonlayer}{edgelayer}
		\draw (6.center) to (8);
		\draw (9) to (10.center);
	\end{pgfonlayer}
\end{tikzpicture}
\mysubeq{\IH{}}
\begin{tikzpicture}
	\begin{pgfonlayer}{nodelayer}
		\node [style=none] (3) at (0.25, 0) {};
		\node [style=none] (5) at (-0.25, 0) {};
	\end{pgfonlayer}
	\begin{pgfonlayer}{edgelayer}
		\draw (5.center) to (3.center);
	\end{pgfonlayer}
\end{tikzpicture}
\mysubeq{\IH{}}
\begin{tikzpicture}
	\begin{pgfonlayer}{nodelayer}
		\node [style=none] (6) at (-0.75, 0) {};
		\node [style=black] (8) at (-0.25, 0) {};
		\node [style=black] (9) at (0, 0) {};
		\node [style=none] (10) at (0.5, 0) {};
	\end{pgfonlayer}
	\begin{pgfonlayer}{edgelayer}
		\draw (6.center) to (8);
		\draw (9) to (10.center);
	\end{pgfonlayer}
\end{tikzpicture}
\end{equation*}
\end{lemma}

Observe that the above laws also holds in both $\IHP{\field}$ and $\AIHP{\field}$ as they contains all the axioms of $\IH{\field}$.

\begin{lemma}\label{lemma:laxcomonoid}
For all arrows $c\colon n\to m$ of $\IHP{\field}$, the following holds
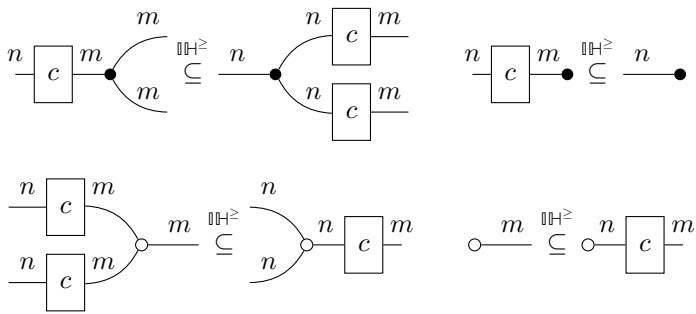
\begin{equation*}
\begin{tikzpicture}
	\begin{pgfonlayer}{nodelayer}
		\node [style=black] (10) at (0.5, 0) {};
		\node [style=small box] (11) at (-0.25, 0) {$c$};
		\node [style=none] (12) at (1.25, 0.5) {};
		\node [style=none] (13) at (1.25, -0.5) {};
		\node [style=none] (14) at (-0.75, 0) {};
		\node [style=none] (15) at (-0.75, 0.25) {$n$};
		\node [style=none] (16) at (0.25, 0.25) {$m$};
		\node [style=none] (17) at (1, 0.75) {$m$};
		\node [style=none] (18) at (1, -0.25) {$m$};
	\end{pgfonlayer}
	\begin{pgfonlayer}{edgelayer}
		\draw (14.center) to (11);
		\draw (11) to (10);
		\draw [bend left] (10) to (12.center);
		\draw [bend right] (10) to (13.center);
	\end{pgfonlayer}
\end{tikzpicture}
\mysubeq{\IHP{}}
\begin{tikzpicture}
	\begin{pgfonlayer}{nodelayer}
		\node [style=black] (10) at (0, 0) {};
		\node [style=small box] (12) at (1, 0.5) {$c$};
		\node [style=small box] (13) at (1, -0.5) {$c$};
		\node [style=none] (14) at (-0.75, 0) {};
		\node [style=none] (15) at (-0.5, 0.25) {$n$};
		\node [style=none] (16) at (0.5, 0.75) {$n$};
		\node [style=none] (17) at (1.5, 0.75) {$m$};
		\node [style=none] (18) at (1.5, -0.25) {$m$};
		\node [style=none] (19) at (1.75, 0.5) {};
		\node [style=none] (20) at (1.75, -0.5) {};
		\node [style=none] (21) at (0.5, -0.25) {$n$};
	\end{pgfonlayer}
	\begin{pgfonlayer}{edgelayer}
		\draw [bend left] (10) to (12);
		\draw [bend right] (10) to (13);
		\draw (14.center) to (10);
		\draw (12) to (19.center);
		\draw (20.center) to (13);
	\end{pgfonlayer}
\end{tikzpicture}
\qquad
\begin{tikzpicture}
	\begin{pgfonlayer}{nodelayer}
		\node [style=black] (10) at (0.5, 0) {};
		\node [style=small box] (11) at (-0.25, 0) {$c$};
		\node [style=none] (14) at (-0.75, 0) {};
		\node [style=none] (15) at (-0.75, 0.25) {$n$};
		\node [style=none] (16) at (0.25, 0.25) {$m$};
	\end{pgfonlayer}
	\begin{pgfonlayer}{edgelayer}
		\draw (14.center) to (11);
		\draw (11) to (10);
	\end{pgfonlayer}
\end{tikzpicture}
\mysubeq{\IHP{}}
\begin{tikzpicture}
	\begin{pgfonlayer}{nodelayer}
		\node [style=black] (10) at (0.5, 0) {};
		\node [style=none] (14) at (-0.25, 0) {};
		\node [style=none] (15) at (0, 0.25) {$n$};
	\end{pgfonlayer}
	\begin{pgfonlayer}{edgelayer}
		\draw (14.center) to (10);
	\end{pgfonlayer}
\end{tikzpicture}
\end{equation*}
\begin{equation*}
\begin{tikzpicture}
	\begin{pgfonlayer}{nodelayer}
		\node [style=white] (10) at (1, 0) {};
		\node [style=small box] (12) at (0, 0.5) {$c$};
		\node [style=small box] (13) at (0, -0.5) {$c$};
		\node [style=none] (14) at (1.75, 0) {};
		\node [style=none] (15) at (1.5, 0.25) {$m$};
		\node [style=none] (16) at (0.5, 0.75) {$m$};
		\node [style=none] (17) at (-0.5, 0.75) {$n$};
		\node [style=none] (18) at (-0.5, -0.25) {$n$};
		\node [style=none] (19) at (-0.75, 0.5) {};
		\node [style=none] (20) at (-0.75, -0.5) {};
		\node [style=none] (21) at (0.5, -0.25) {$m$};
	\end{pgfonlayer}
	\begin{pgfonlayer}{edgelayer}
		\draw [bend right] (10) to (12);
		\draw [bend left] (10) to (13);
		\draw (14.center) to (10);
		\draw (12) to (19.center);
		\draw (20.center) to (13);
	\end{pgfonlayer}
\end{tikzpicture}
\mysubeq{\IHP{}}
    \begin{tikzpicture}
	\begin{pgfonlayer}{nodelayer}
		\node [style=white] (10) at (0, 0) {};
		\node [style=small box] (11) at (0.75, 0) {$c$};
		\node [style=none] (12) at (-0.75, 0.5) {};
		\node [style=none] (13) at (-0.75, -0.5) {};
		\node [style=none] (14) at (1.25, 0) {};
		\node [style=none] (15) at (1.25, 0.25) {$m$};
		\node [style=none] (16) at (0.25, 0.25) {$n$};
		\node [style=none] (17) at (-0.5, 0.75) {$n$};
		\node [style=none] (18) at (-0.5, -0.25) {$n$};
	\end{pgfonlayer}
	\begin{pgfonlayer}{edgelayer}
		\draw (14.center) to (11);
		\draw (11) to (10);
		\draw [bend right] (10) to (12.center);
		\draw [bend left] (10) to (13.center);
	\end{pgfonlayer}
\end{tikzpicture}
\qquad
\begin{tikzpicture}
	\begin{pgfonlayer}{nodelayer}
		\node [style=white] (10) at (-0.25, 0) {};
		\node [style=none] (14) at (0.5, 0) {};
		\node [style=none] (15) at (0.25, 0.25) {$m$};
	\end{pgfonlayer}
	\begin{pgfonlayer}{edgelayer}
		\draw (14.center) to (10);
	\end{pgfonlayer}
\end{tikzpicture}
\mysubeq{\IHP{}}
\begin{tikzpicture}
	\begin{pgfonlayer}{nodelayer}
		\node [style=white] (10) at (-0.75, 0) {};
		\node [style=small box] (11) at (0, 0) {$c$};
		\node [style=none] (14) at (0.5, 0) {};
		\node [style=none] (15) at (0.5, 0.25) {$m$};
		\node [style=none] (16) at (-0.5, 0.25) {$n$};
	\end{pgfonlayer}
	\begin{pgfonlayer}{edgelayer}
		\draw (14.center) to (11);
		\draw (11) to (10);
	\end{pgfonlayer}
\end{tikzpicture}
\end{equation*}
\end{lemma}
\begin{proof}
The proof goes by induction on $\CDP$. 
As base cases, we have to consider all the generators. For the generators of $\CD$, the properties hold by~\cite[\S 6]{BonchiHPS17}. For $\greq$, the properties follows immediately by the axioms (P1), (P2), (P3) and (P4). The inductive cases are trivial.
\end{proof}

From the two lemmas above, it immediately follows that $\IHP{\field}$ is an Abelian bicategory.

\begin{proposition}
$\IHP{\field}$ is an Abelian bicategory.
\end{proposition}
\begin{proof}
It is enough to check that all the axioms of the definition in~\cite{Carboni1987} holds. The comonoid is provided by the black structure; the monoid by the white structure. The fact that they have right adjoints is just Lemma~\ref{lemma:adjoint}. The axioms $\bullet - fr$ and $\circ - fr$ ensure that the black and the white structure form Frobenius algebras.
Finally, Lemma~\ref{lemma:laxcomonoid} guarantees that every arrow is both a lax comonoid and a lax monoid homomorphism.
\end{proof}

From this fact, it immediately follows (see e.g.~\cite[\S 7]{BonchiHPS17}) that every set of arrows $n \to m$ in $\IHP{\field}$ is not just oredered but, actually, $\mysubeq{\IHP{}}$ forms a lattice: meet, top, join and bottom are illustrated below.
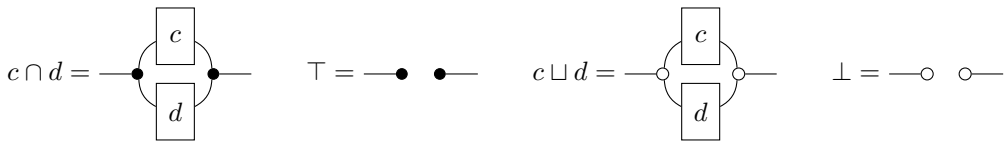
\begin{equation*}
c \cap d = \begin{tikzpicture}
    \begin{pgfonlayer}{nodelayer}
		\node [style=black] (10) at (-0.5, 0) {};
		\node [style=small box] (11) at (0, 0.5) {$c$};
		\node [style=black] (14) at (0.5, 0) {};
		\node [style=small box] (15) at (0, -0.5) {$d$};
		\node [style=none] (16) at (-1, 0) {};
		\node [style=none] (17) at (1, 0) {};
	\end{pgfonlayer}
	\begin{pgfonlayer}{edgelayer}
		\draw [bend right] (14) to (11);
		\draw [bend right] (11) to (10);
		\draw [bend right] (10) to (15);
		\draw [bend right] (15) to (14);
		\draw (16.center) to (10);
		\draw (17.center) to (14);
	\end{pgfonlayer}
\end{tikzpicture}
\qquad
    \top = \begin{tikzpicture}
	\begin{pgfonlayer}{nodelayer}
		\node [style=black] (10) at (-0.5, 0) {};
		\node [style=black] (14) at (0, 0) {};
		\node [style=none] (16) at (-1, 0) {};
		\node [style=none] (17) at (0.5, 0) {};
	\end{pgfonlayer}
	\begin{pgfonlayer}{edgelayer}
		\draw (16.center) to (10);
		\draw (17.center) to (14);
	\end{pgfonlayer}
\end{tikzpicture}
\qquad
c \sqcup d = \begin{tikzpicture}
    \begin{pgfonlayer}{nodelayer}
		\node [style=white] (10) at (-0.5, 0) {};
		\node [style=small box] (11) at (0, 0.5) {$c$};
		\node [style=white] (14) at (0.5, 0) {};
		\node [style=small box] (15) at (0, -0.5) {$d$};
		\node [style=none] (16) at (-1, 0) {};
		\node [style=none] (17) at (1, 0) {};
	\end{pgfonlayer}
	\begin{pgfonlayer}{edgelayer}
		\draw [bend right] (14) to (11);
		\draw [bend right] (11) to (10);
		\draw [bend right] (10) to (15);
		\draw [bend right] (15) to (14);
		\draw (16.center) to (10);
		\draw (17.center) to (14);
	\end{pgfonlayer}
\end{tikzpicture}
\qquad
\bot = \begin{tikzpicture}
	\begin{pgfonlayer}{nodelayer}
		\node [style=white] (10) at (-0.5, 0) {};
		\node [style=white] (14) at (0, 0) {};
		\node [style=none] (16) at (-1, 0) {};
		\node [style=none] (17) at (0.5, 0) {};
	\end{pgfonlayer}
	\begin{pgfonlayer}{edgelayer}
		\draw (16.center) to (10);
		\draw (17.center) to (14);
	\end{pgfonlayer}
\end{tikzpicture}
\end{equation*}

Unfortunately, the last inequality of Lemma~\ref {lemma:laxcomonoid} does not hold for $\One$ and thus not all arrows of $\AIHP{\field}$ are lax monoid homomorphisms. However, they are still lax comonoid homomorphisms.

\begin{lemma}\label{lemma:laxcomonoidA}
For all arrows $c\colon n\to m$ of $\AIHP{\field}$, the following holds
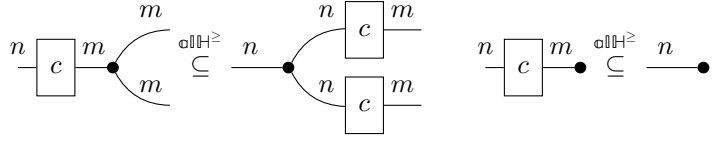
\begin{equation*}
\begin{tikzpicture}
	\begin{pgfonlayer}{nodelayer}
		\node [style=black] (10) at (0.5, 0) {};
		\node [style=small box] (11) at (-0.25, 0) {$c$};
		\node [style=none] (12) at (1.25, 0.5) {};
		\node [style=none] (13) at (1.25, -0.5) {};
		\node [style=none] (14) at (-0.75, 0) {};
		\node [style=none] (15) at (-0.75, 0.25) {$n$};
		\node [style=none] (16) at (0.25, 0.25) {$m$};
		\node [style=none] (17) at (1, 0.75) {$m$};
		\node [style=none] (18) at (1, -0.25) {$m$};
	\end{pgfonlayer}
	\begin{pgfonlayer}{edgelayer}
		\draw (14.center) to (11);
		\draw (11) to (10);
		\draw [bend left] (10) to (12.center);
		\draw [bend right] (10) to (13.center);
	\end{pgfonlayer}
\end{tikzpicture}
\mysubeq{\AIHP{}}
\begin{tikzpicture}
	\begin{pgfonlayer}{nodelayer}
		\node [style=black] (10) at (0, 0) {};
		\node [style=small box] (12) at (1, 0.5) {$c$};
		\node [style=small box] (13) at (1, -0.5) {$c$};
		\node [style=none] (14) at (-0.75, 0) {};
		\node [style=none] (15) at (-0.5, 0.25) {$n$};
		\node [style=none] (16) at (0.5, 0.75) {$n$};
		\node [style=none] (17) at (1.5, 0.75) {$m$};
		\node [style=none] (18) at (1.5, -0.25) {$m$};
		\node [style=none] (19) at (1.75, 0.5) {};
		\node [style=none] (20) at (1.75, -0.5) {};
		\node [style=none] (21) at (0.5, -0.25) {$n$};
	\end{pgfonlayer}
	\begin{pgfonlayer}{edgelayer}
		\draw [bend left] (10) to (12);
		\draw [bend right] (10) to (13);
		\draw (14.center) to (10);
		\draw (12) to (19.center);
		\draw (20.center) to (13);
	\end{pgfonlayer}
\end{tikzpicture}
\qquad
\begin{tikzpicture}
	\begin{pgfonlayer}{nodelayer}
		\node [style=black] (10) at (0.5, 0) {};
		\node [style=small box] (11) at (-0.25, 0) {$c$};
		\node [style=none] (14) at (-0.75, 0) {};
		\node [style=none] (15) at (-0.75, 0.25) {$n$};
		\node [style=none] (16) at (0.25, 0.25) {$m$};
	\end{pgfonlayer}
	\begin{pgfonlayer}{edgelayer}
		\draw (14.center) to (11);
		\draw (11) to (10);
	\end{pgfonlayer}
\end{tikzpicture}
\mysubeq{\AIHP{}}
\begin{tikzpicture}
	\begin{pgfonlayer}{nodelayer}
		\node [style=black] (10) at (0.5, 0) {};
		\node [style=none] (14) at (-0.25, 0) {};
		\node [style=none] (15) at (0, 0.25) {$n$};
	\end{pgfonlayer}
	\begin{pgfonlayer}{edgelayer}
		\draw (14.center) to (10);
	\end{pgfonlayer}
\end{tikzpicture}
\end{equation*}
\end{lemma}
\begin{proof}
Following the proof of Lemma~\ref{lemma:laxcomonoid}, we have only to prove that the properties hold for $\One$. This is immediate by the axioms $dup$ and $del$.
\end{proof}

\begin{proposition}
$\IHP{\field}$ is a Cartesian bicategory.
\end{proposition}
\begin{proof}
It is enough to check that all the axioms of the definition in~\cite{Carboni1987} holds. The comonoid is provided by the black structure and the adjoints by Lemma~\ref{lemma:adjoint}. The axioms $\bullet - fr$ ensure that the black structure form Frobenius algebras.
Finally, Lemma~\ref{lemma:laxcomonoidA} guarantees that every arrow is a lax comonoid homomorphism.
\end{proof}

From this fact follows that every set of arrows $n \to m$ in $\AIHP{\field}$ is not just ordered but, actually, $\mysubeq{\AIHP{}}$ forms a meet semi-lattice, where meet and top are defined like above. This means that
\begin{equation}\label{eq:meetsemilattice}
    \begin{tikzpicture}
	\begin{pgfonlayer}{nodelayer}
		\node [style=small box] (21) at (0, 0) {$d$};
		\node [style=none] (23) at (-0.75, 0) {};
		\node [style=none] (24) at (0.75, 0) {};
	\end{pgfonlayer}
	\begin{pgfonlayer}{edgelayer}
		\draw (21) to (24.center);
		\draw (23.center) to (21);
	\end{pgfonlayer}
\end{tikzpicture}
\mysubeq{\AIHP{}}
\begin{tikzpicture}
	\begin{pgfonlayer}{nodelayer}
		\node [style=small box] (21) at (0, 0) {$c$};
		\node [style=none] (23) at (-0.75, 0) {};
		\node [style=none] (24) at (0.75, 0) {};
	\end{pgfonlayer}
	\begin{pgfonlayer}{edgelayer}
		\draw (21) to (24.center);
		\draw (23.center) to (21);
	\end{pgfonlayer}
\end{tikzpicture}
\Longleftrightarrow
\begin{tikzpicture}
	\begin{pgfonlayer}{nodelayer}
		\node [style=small box] (21) at (0, -0.5) {$c$};
		\node [style=small box] (25) at (0, 0.5) {$d$};
		\node [style=black] (26) at (-0.75, 0) {};
		\node [style=black] (27) at (0.75, 0) {};
		\node [style=none] (28) at (1.25, 0) {};
		\node [style=none] (29) at (-1.25, 0) {};
	\end{pgfonlayer}
	\begin{pgfonlayer}{edgelayer}
		\draw [bend left] (25) to (27);
		\draw [bend left] (26) to (25);
		\draw [bend right] (26) to (21);
		\draw [bend right] (21) to (27);
		\draw (27) to (28.center);
		\draw (26) to (29.center);
	\end{pgfonlayer}
\end{tikzpicture}
\myeq{\AIHP{}}
\begin{tikzpicture}
	\begin{pgfonlayer}{nodelayer}
		\node [style=small box] (21) at (0, 0) {$d$};
		\node [style=none] (23) at (-0.75, 0) {};
		\node [style=none] (24) at (0.75, 0) {};
	\end{pgfonlayer}
	\begin{pgfonlayer}{edgelayer}
		\draw (21) to (24.center);
		\draw (23.center) to (21);
	\end{pgfonlayer}
\end{tikzpicture}
\end{equation}


\section{Proofs of Section~\ref{sec:poly}}\label{app:poly}
In this appendix, we illustrate the proofs of Section~\ref{sec:poly}. We begin with a few properties of an arbitrary ordered field $\field$ and some additional lemmas.

\begin{lemma}\label{lm:arch}
$\forall \alpha \in \field . \ \exists \gamma \in \field . \ \gamma \geq \alpha$
\end{lemma}
\begin{proof}
We know that $1 \in \field$ and $1 \geq 0$. By the ordered field axioms we get
\begin{align*}
    &1 \geq 0 \\
    \implies \ &0 + 1 \geq 0 \\
    \implies \ &\alpha - \alpha + 1 \geq 0 \\
    \implies \ &\alpha + 1 \geq \alpha.
\end{align*}
Therefore, $\gamma \geq \alpha$ exists and it is equal to $\alpha + 1$.
\end{proof}

\begin{corollary}\label{co:arch}
$\lnot \exists \alpha \in \field . \ \forall \gamma \in \field . \ \gamma < \alpha$
\end{corollary}

\begin{lemma}\label{lm:neg}
Let $b, c \in \field$. Then the following holds
\[ \forall \alpha \in \field . \ \alpha > 0, \ \alpha b \leq c \implies b \leq 0 \]
\end{lemma}
\begin{proof}
Suppose $b > 0$, then it must be $c > 0$.

Thus $\alpha \leq \frac{c}{b}$ for any $\alpha > 0$. 

Note that $0 \leq \frac{c}{b}$ since both $c$ and $b$ are $< 0$, and thus $\forall \beta < 0 \ \beta \leq \frac{c}{b}$.

Therefore $\forall \gamma . \ \gamma \leq \frac{c}{b}$, which contradicts Corollary~\ref{co:arch}.
\end{proof}

\begin{corollary}\label{co:negvec}
Let $v, w \in \field^n$. Then the following holds
\[ \forall \alpha . \ \alpha > 0, \ \alpha v \leq w \implies v \leq 0 \]
\end{corollary}
\begin{proof}
It is sufficient to notice that Lemma~\ref{lm:neg} holds for all inequalities $\alpha v_i \leq w_i$ for $i=1\ldots n$.
\end{proof}


\begin{proof}[Proof of Lemma \ref{lm:nf2}]
Since $\One$ has arity 0, by naturality of symmetry we can always write $c$ as 
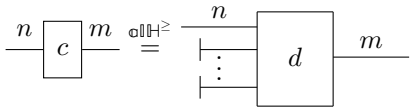
\begin{equation*}
    \begin{tikzpicture}
	\begin{pgfonlayer}{nodelayer}
		\node [style=small box] (21) at (0, 0) {$c$};
		\node [style=none] (23) at (-0.75, 0) {};
		\node [style=none] (24) at (0.75, 0) {};
		\node [style=none] (25) at (0.5, 0.25) {$m$};
		\node [style=none] (26) at (-0.5, 0.25) {$n$};
	\end{pgfonlayer}
	\begin{pgfonlayer}{edgelayer}
		\draw (21) to (24.center);
		\draw (23.center) to (21);
	\end{pgfonlayer}
\end{tikzpicture}
\myeq{\AIHP{}}
\begin{tikzpicture}
	\begin{pgfonlayer}{nodelayer}
		\node [style=none] (24) at (1.5, -0.1) {};
		\node [style=none] (25) at (1, 0.1) {$m$};
		\node [style=none] (26) at (-1, 0.5) {$n$};
		\node [style=none] (27) at (-1.5, 0.3) {};
		\node [style=none] (28) at (-1.25, 0) {};
		\node [style=none] (29) at (-1.25, 0.15) {};
		\node [style=none] (31) at (-0.5, 0.3) {};
		\node [style=none] (32) at (-0.5, 0) {};
		\node [style=none] (33) at (-1.25, -0.5) {};
		\node [style=none] (34) at (-1.25, -0.65) {};
		\node [style=none] (35) at (-0.5, -0.5) {};
		\node [style=none] (36) at (-1, -0.25) {$\myvdots$};
		\node [style=none] (37) at (-1.25, -0.35) {};
		\node [style=none] (38) at (-1.25, -0.15) {};
		\node [style=none] (39) at (-0.5, 0.5) {};
		\node [style=none] (40) at (0.5, 0.5) {};
		\node [style=none] (41) at (0.5, -0.75) {};
		\node [style=none] (42) at (-0.5, -0.75) {};
		\node [style=none] (43) at (0.5, -0.1) {};
		\node [style=none] (44) at (0, -0.1) {$d$};
	\end{pgfonlayer}
	\begin{pgfonlayer}{edgelayer}
		\draw (27.center) to (31.center);
		\draw (28.center) to (32.center);
		\draw (33.center) to (35.center);
		\draw (29.center) to (38.center);
		\draw (37.center) to (34.center);
		\draw (39.center) to (40.center);
		\draw (40.center) to (41.center);
		\draw (41.center) to (42.center);
		\draw (42.center) to (39.center);
		\draw (43.center) to (24.center);
	\end{pgfonlayer}
\end{tikzpicture}
\end{equation*}
for some $d$ in $\IHP{\mathsf k}$. 
By the axiom (dup), we can group all the $\One$ in order to obtain 
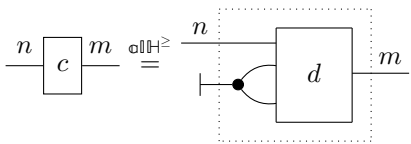
\begin{equation*}
\begin{tikzpicture}
	\begin{pgfonlayer}{nodelayer}
		\node [style=small box] (21) at (0, 0) {$c$};
		\node [style=none] (23) at (-0.75, 0) {};
		\node [style=none] (24) at (0.75, 0) {};
		\node [style=none] (25) at (0.5, 0.25) {$m$};
		\node [style=none] (26) at (-0.5, 0.25) {$n$};
	\end{pgfonlayer}
	\begin{pgfonlayer}{edgelayer}
		\draw (21) to (24.center);
		\draw (23.center) to (21);
	\end{pgfonlayer}
\end{tikzpicture}
\myeq{\AIHP{}}
\begin{tikzpicture}
	\begin{pgfonlayer}{nodelayer}
		\node [style=none] (24) at (1.25, -0.1) {};
		\node [style=none] (25) at (1, 0.1) {$m$};
		\node [style=none] (26) at (-1.5, 0.5) {$n$};
		\node [style=none] (27) at (-1.75, 0.3) {};
		\node [style=none] (31) at (-0.5, 0.3) {};
		\node [style=none] (32) at (-0.5, 0) {};
		\node [style=none] (33) at (-1.5, -0.25) {};
		\node [style=none] (34) at (-1.5, -0.4) {};
		\node [style=none] (35) at (-0.5, -0.5) {};
		\node [style=none] (37) at (-1.5, -0.1) {};
		\node [style=none] (39) at (-0.5, 0.5) {};
		\node [style=none] (40) at (0.5, 0.5) {};
		\node [style=none] (41) at (0.5, -0.75) {};
		\node [style=none] (42) at (-0.5, -0.75) {};
		\node [style=none] (43) at (0.5, -0.1) {};
		\node [style=none] (44) at (0, -0.1) {$d$};
		\node [style=black] (45) at (-1, -0.25) {};
		\node [style=none] (46) at (-1.25, 0.75) {};
		\node [style=none] (47) at (-1.25, -1) {};
		\node [style=none] (48) at (0.75, -1) {};
		\node [style=none] (49) at (0.75, 0.75) {};
	\end{pgfonlayer}
	\begin{pgfonlayer}{edgelayer}
		\draw (27.center) to (31.center);
		\draw (37.center) to (34.center);
		\draw (39.center) to (40.center);
		\draw (40.center) to (41.center);
		\draw (41.center) to (42.center);
		\draw (42.center) to (39.center);
		\draw (43.center) to (24.center);
		\draw [bend left] (45) to (32.center);
		\draw [bend right] (45) to (35.center);
		\draw (45) to (33.center);
		\draw [dotted] (46.center) to (49.center);
		\draw [dotted] (49.center) to (48.center);
		\draw [dotted] (48.center) to (47.center);
		\draw [dotted] (47.center) to (46.center);
	\end{pgfonlayer}
\end{tikzpicture}
\end{equation*}
The diagram in the dotted box is the arrow $c'$.
\end{proof}

\begin{proof}[Proof of Theorem \ref{th:nf}]
By Lemma~\ref{lm:nf2}
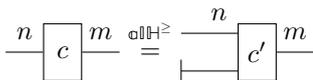
\begin{equation*}
\begin{tikzpicture}
	\begin{pgfonlayer}{nodelayer}
		\node [style=small box] (21) at (0, 0) {$c$};
		\node [style=none] (23) at (-0.75, 0) {};
		\node [style=none] (24) at (0.75, 0) {};
		\node [style=none] (25) at (0.5, 0.25) {$m$};
		\node [style=none] (26) at (-0.5, 0.25) {$n$};
	\end{pgfonlayer}
	\begin{pgfonlayer}{edgelayer}
		\draw (21) to (24.center);
		\draw (23.center) to (21);
	\end{pgfonlayer}
\end{tikzpicture}
\myeq{\AIHP{}}
\begin{tikzpicture}
	\begin{pgfonlayer}{nodelayer}
		\node [style=small box] (21) at (0, 0) {$c'$};
		\node [style=none] (24) at (0.75, 0) {};
		\node [style=none] (25) at (0.5, 0.25) {$m$};
		\node [style=none] (26) at (-0.5, 0.5) {$n$};
		\node [style=none] (27) at (-1, 0.25) {};
		\node [style=none] (28) at (-1, -0.25) {};
		\node [style=none] (31) at (-0.25, 0.25) {};
		\node [style=none] (32) at (-0.25, -0.25) {};
		\node [style=none] (33) at (-1, -0.1) {};
		\node [style=none] (34) at (-1, -0.4) {};
	\end{pgfonlayer}
	\begin{pgfonlayer}{edgelayer}
		\draw (21) to (24.center);
		\draw (27.center) to (31.center);
		\draw (28.center) to (32.center);
		\draw (33.center) to (34.center);
	\end{pgfonlayer}
\end{tikzpicture}
\end{equation*} with $c'$ an arrow of $\IHP{\mathsf k}$.

By the first normal form theorem of $\IHP{\mathsf k}$, there exists an  \\ arrow $D \colon m+n+1 \rightarrow o$ of $\FC$ such that 
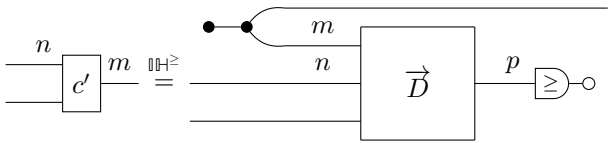
\begin{equation*}
\begin{tikzpicture}
	\begin{pgfonlayer}{nodelayer}
		\node [style=small box] (21) at (0, 0) {$c'$};
		\node [style=none] (24) at (0.75, 0) {};
		\node [style=none] (25) at (0.5, 0.25) {$m$};
		\node [style=none] (26) at (-0.5, 0.5) {$n$};
		\node [style=none] (27) at (-1, 0.25) {};
		\node [style=none] (28) at (-1, -0.25) {};
		\node [style=none] (31) at (-0.25, 0.25) {};
		\node [style=none] (32) at (-0.25, -0.25) {};
	\end{pgfonlayer}
	\begin{pgfonlayer}{edgelayer}
		\draw (21) to (24.center);
		\draw (27.center) to (31.center);
		\draw (28.center) to (32.center);
	\end{pgfonlayer}
\end{tikzpicture}
    \myeq{\IHP{}}
\begin{tikzpicture}
	\begin{pgfonlayer}{nodelayer}
		\node [style=reg] (3) at (1.75, 0) {$\geq$};
		\node [style=white] (4) at (2.25, 0) {};
		\node [style=none] (9) at (-1.75, 0.5) {};
		\node [style=none] (10) at (-3, 0) {};
		\node [style=none] (11) at (-1.25, 0.25) {$n$};
		\node [style=none] (12) at (-1.25, 0.75) {$m$};
		\node [style=none] (13) at (-0.75, 0.5) {};
		\node [style=none] (14) at (-0.75, 0) {};
		\node [style=black] (15) at (-2.25, 0.75) {};
		\node [style=black] (16) at (-2.75, 0.75) {};
		\node [style=none] (17) at (-1.75, 1) {};
		\node [style=none] (18) at (2.5, 1) {};
		\node [style=none] (19) at (-3, -0.5) {};
		\node [style=none] (20) at (-0.75, -0.5) {};
		\node [style=none] (21) at (-0.75, 0.75) {};
		\node [style=none] (22) at (0.75, 0.75) {};
		\node [style=none] (23) at (0.75, -0.75) {};
		\node [style=none] (24) at (-0.75, -0.75) {};
		\node [style=none] (25) at (0.75, 0) {};
		\node [style=none] (26) at (0, 0) {$\overrightarrow{D}$};
		\node [style=none] (27) at (1.25, 0.25) {$p$};
	\end{pgfonlayer}
	\begin{pgfonlayer}{edgelayer}
		\draw (3) to (4);
		\draw (10.center) to (14.center);
		\draw (9.center) to (13.center);
		\draw (17.center) to (18.center);
		\draw [bend right=330] (9.center) to (15);
		\draw [bend left] (15) to (17.center);
		\draw (15) to (16);
		\draw (19.center) to (20.center);
		\draw (21.center) to (24.center);
		\draw (24.center) to (23.center);
		\draw (21.center) to (22.center);
		\draw (22.center) to (23.center);
		\draw (25.center) to (3);
	\end{pgfonlayer}
\end{tikzpicture}
\end{equation*}
Since $D$ is a matrix, it can be rewritten as 
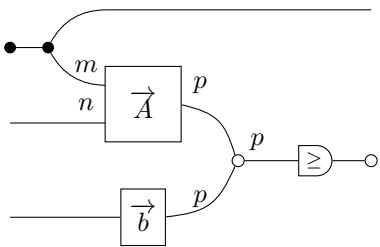
\begin{equation*}
\begin{tikzpicture}
	\begin{pgfonlayer}{nodelayer}
		\node [style=square] (0) at (0, 0.75) {$\overrightarrow{A}$};
		\node [style=small box] (1) at (0, -0.75) {$\overrightarrow{b}$};
		\node [style=white] (2) at (1.25, 0) {};
		\node [style=reg] (3) at (2.25, 0) {$\geq$};
		\node [style=white] (4) at (3, 0) {};
		\node [style=none] (6) at (-1.75, -0.75) {};
		\node [style=none] (9) at (-0.5, 1) {};
		\node [style=none] (10) at (-1.75, 0.5) {};
		\node [style=none] (11) at (-0.75, 0.75) {$n$};
		\node [style=none] (12) at (-0.75, 1.25) {$m$};
		\node [style=none] (13) at (-0.25, 1) {};
		\node [style=none] (14) at (-0.25, 0.5) {};
		\node [style=black] (15) at (-1.25, 1.5) {};
		\node [style=black] (16) at (-1.75, 1.5) {};
		\node [style=none] (17) at (-0.5, 2) {};
		\node [style=none] (18) at (3, 2) {};
		\node [style=none] (19) at (1.5, 0.25) {$p$};
		\node [style=none] (20) at (0.75, 1) {$p$};
		\node [style=none] (21) at (0.75, -0.5) {$p$};
	\end{pgfonlayer}
	\begin{pgfonlayer}{edgelayer}
		\draw [bend left, looseness=1.25] (0) to (2);
		\draw [bend right, looseness=1.25] (1) to (2);
		\draw (2) to (3);
		\draw (3) to (4);
		\draw (6.center) to (1);
		\draw (10.center) to (14.center);
		\draw (9.center) to (13.center);
		\draw (17.center) to (18.center);
		\draw [bend right=330] (9.center) to (15);
		\draw [bend left] (15) to (17.center);
		\draw (15) to (16);
	\end{pgfonlayer}
\end{tikzpicture}
\end{equation*}
and thus the thesis.
\end{proof}

\begin{proof}[Proof of Lemma~\ref{th:polyinc}]
We fix $P_1 = \{ x \in \field^n \mid Ax \leq b \}$ and $P_2= \{ x \in \field^n \mid Cx \leq d \}$
and their \textit{homogenisations}
\[ P_1^H = \{ \begin{pmatrix} x \\ y\end{pmatrix} \in \field^{n+1} \mid \begin{pmatrix} A & -b \\ 0 & -1\end{pmatrix}\begin{pmatrix} x \\ y\end{pmatrix} \leq 0 \} \quad
P_2^H = \{ \begin{pmatrix} x \\ y\end{pmatrix} \in \field^{n+1} \mid \begin{pmatrix} C & -d \\ 0 & -1\end{pmatrix}\begin{pmatrix} x \\ y\end{pmatrix} \leq 0 \}\]

1. First we prove that $P_1 \subseteq P_2 \implies P_1^H \subseteq P_2^H$.

Take $\vect{x \\ y} \in P_1^H$. If $y > 0$ then $\vect{\frac{x}{y} \\ 1} \in P_1^H$ if and only if $\vect{\frac{x}{y}} \in P_1$, and thus $\vect{\frac{x}{y}} \in P_2$ and $\vect{\frac{x}{y} \\ 1} \in P_2^H$. Finally, multiplying the vector by $y$ we get that $\vect{x \\ y} \in P_2^H$.

If $y = 0$, then $\vect{x} \in \{ x \in \field^n \mid Ax \leq 0\}$. We pick a vector $x' \in P_1$ and note that $x' + \lambda x$ is also in $P_1$ and thus in $P_2$, for all $\lambda \geq 0$.

Since $x' + \lambda x \in P_2$, we have that $Cx' + C\lambda x \leq d$ and $C\lambda x \leq d - Cx'$, which by Corollary~\ref{co:negvec} implies that $Cx \leq 0$. In other words, $\vect{x} \in \{ x \in \field^n \mid Cx \leq 0\}$ and thus $\vect{x \\ y} \in P_2^H$.

2. The other direction $P_1^H \subseteq P_2^H \implies P_1 \subseteq P_2$ follows from the fact that $P_1 = \{ x \in \field^n \mid \vect{x \\ 1} \in P_1^H \}$ and $P_2 = \{ x \in \field^n \mid \vect{x \\ 1} \in P_2^H \}$.
\end{proof}

\begin{proof}[Proof of Theorem~\ref{th:compl}]
By Theorem~\ref{th:nf}, 
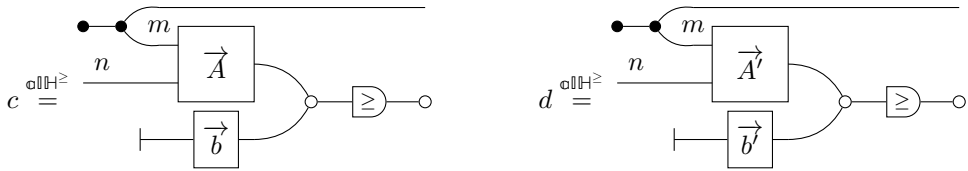
\begin{equation*}
    c \myeq{\AIHP{}} \input{figures/poly/nfpolyhedron}
\qquad \qquad
    d \myeq{\AIHP{}} \input{figures/poly/nfpolyhedron2}
\end{equation*}
So $\dsem{c} = \{ x \in \field^{n+m} \mid Ax + b \geq 0\}$ and $\dsem{d} = \{ x \in \field^{n+m} \mid A'x + b' \geq 0\}$. Since $\dsem{c} \subseteq \dsem{d}$, by Lemma~\ref{th:polyinc}, it holds that
\begin{equation*}
    \dsem{c^H} = \{ \vect{x \\ y} \mid Ax + by \geq 0, \ y \geq 0\} \subseteq \{ \vect{x \\ y} \mid A'x + b'y \geq 0, \ y \geq 0\} = \dsem{d^H}
\end{equation*}
Diagrammatically, 
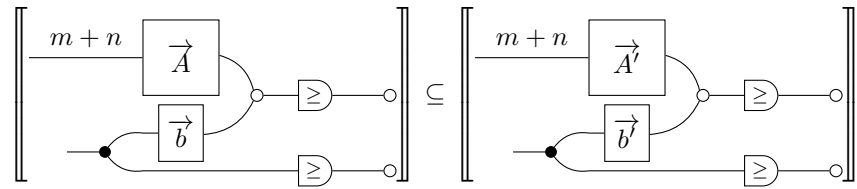
\input{figures/poly/compl1}
Note that $c^H$ and $d^H$ are arrows of $\CDP$ (polyhedral cones), thus, by completeness of $\IHP{\mathsf k}$ we have that 
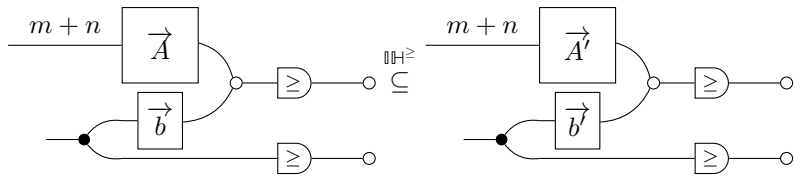
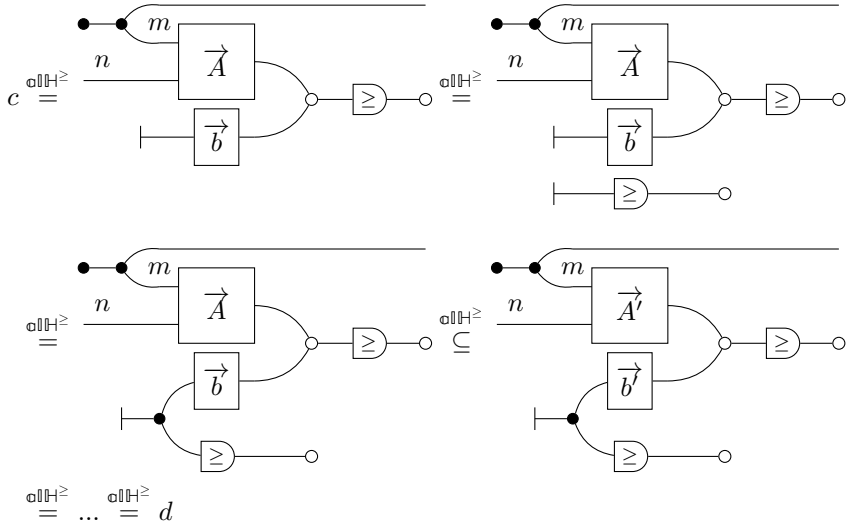
\input{figures/poly/compl2}
Thus,
\input{figures/poly/compl3}
\end{proof}

\begin{lemma}
\begin{equation*}
\begin{tikzpicture}
	\begin{pgfonlayer}{nodelayer}
		\node [style=none] (0) at (-0.25, 0) {};
		\node [style=white] (1) at (0.25, 0) {};
		\node [style=none] (2) at (-0.25, 0.15) {};
		\node [style=none] (3) at (-0.25, -0.15) {};
	\end{pgfonlayer}
	\begin{pgfonlayer}{edgelayer}
		\draw (2.center) to (3.center);
		\draw (1) to (0.center);
	\end{pgfonlayer}
\end{tikzpicture}
 \myeq{AP2} 
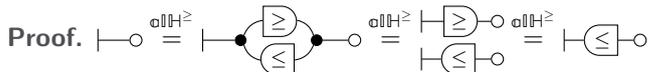
\begin{tikzpicture}
	\begin{pgfonlayer}{nodelayer}
		\node [style=coreg] (7) at (0, 0) {$\leq$};
		\node [style=none] (8) at (-0.5, 0) {};
		\node [style=none] (9) at (-0.5, 0.15) {};
		\node [style=none] (10) at (-0.5, -0.15) {};
		\node [style=white] (11) at (0.5, 0) {};
	\end{pgfonlayer}
	\begin{pgfonlayer}{edgelayer}
		\draw (9.center) to (10.center);
		\draw (7) to (8.center);
		\draw (11) to (7);
	\end{pgfonlayer}
\end{tikzpicture}
\end{equation*}
\end{lemma}
\begin{proof}
\input{figures/poly/ap2-proof}
\end{proof}

\begin{lemma}\label{lm:inc}
Let $c \colon n \rightarrow m$ be an arrow of $\ACDP$, then the following holds
\begin{equation*}
\input{figures/poly/inc-def}
\end{equation*}
\end{lemma}
\begin{proof}
\input{figures/poly/inc-proof}
\end{proof}

\begin{lemma}\label{lm:zero}
Let $c \colon 0 \rightarrow 0$ be an arrow of $\ACDP$. Then 
\begin{equation*}
    \input{figures/poly/zero-def}
\end{equation*}
\end{lemma}
\begin{proof} 
\input{figures/poly/zero-proof}
\end{proof}

\begin{proof}[Proof of Theorem~\ref{th:acompl2}]
\input{figures/poly/compl2-proof}
\end{proof}

\begin{proof}[Proof of Proposition~\ref{th:aihpiso}]
For the first statement is enough to check that the semantics of the circuit in~\eqref{eq:polynf} is exactly $\{ (x,y) \in \field^n \times \field^m \mid A\begin{pmatrix*} x \\ y \end{pmatrix*} + b \geq 0 \}$. For the vice versa, we proceed by induction: as in the proof of Proposition~\ref{th:ihpiso}, the inductive cases are trivial. For the base case, we only have to check $\One$, since we have already seen that the semantics of the other generators are polyhedral cones and thus a polyhedra.
For $\One$, observe that $\{(\bullet, y) \in \field^0 \times \field^1 \mid  \begin{pmatrix*} -1 \\ 1 \end{pmatrix*} y + \begin{pmatrix*} 1 \\ -1 \end{pmatrix*} \geq 0 \}= \{(\bullet, 1)\} = \dsem{\One}$.
\end{proof}

\section{Proof of Example~\ref{ex:fnet}}\label{app:flow}
\begin{proof}[Proof of Derived law~\ref{eq:dlaw1}]
Consider the semantics of both left and right hand side
\begin{equation*}
    \dsem{
\tikzset{x=1em, y=2.1ex}
\InputIfFileExists{flownet/law1-1.tikz}{}{\input{./tikz/flownet/law1-1.tikz}}
\tikzset{x=1em, y=1.5ex}
} = \{(\vect{a \\ b}, c) \mid a \leq k, b \leq q, a + b = c \leq l\}
\end{equation*}
\begin{equation*}
    \dsem{
\tikzset{x=1em, y=2.1ex}
\InputIfFileExists{flownet/law1-2.tikz}{}{\input{./tikz/flownet/law1-2.tikz}}
\tikzset{x=1em, y=1.5ex}
} = \{(\vect{a \\ b}, c) \mid a \leq k, b \leq q, a + b = c\}
\end{equation*}
and the hypothesis that $k + q \leq l$. The inclusion $\dsem{
\tikzset{x=1em, y=2.1ex}
\InputIfFileExists{flownet/law1-2.tikz}{}{\input{./tikz/flownet/law1-2.tikz}}
\tikzset{x=1em, y=1.5ex}
} \subseteq \dsem{
\tikzset{x=1em, y=2.1ex}
\InputIfFileExists{flownet/law1-1.tikz}{}{\input{./tikz/flownet/law1-1.tikz}}
\tikzset{x=1em, y=1.5ex}
}$ holds by means of the ordered fields axioms, i.e. $a \leq k, b \leq q \implies a + b \leq k + q \leq l$. The other inclusion $\dsem{
\tikzset{x=1em, y=2.1ex}
\InputIfFileExists{flownet/law1-2.tikz}{}{\input{./tikz/flownet/law1-2.tikz}}
\tikzset{x=1em, y=1.5ex}
} \subseteq \dsem{
\tikzset{x=1em, y=2.1ex}
\InputIfFileExists{flownet/law1-1.tikz}{}{\input{./tikz/flownet/law1-1.tikz}}
\tikzset{x=1em, y=1.5ex}
}$ is proved as follows: suppose $(\vect{a \\ b}, c) \not\in \dsem{
\tikzset{x=1em, y=2.1ex}
\InputIfFileExists{flownet/law1-2.tikz}{}{\input{./tikz/flownet/law1-2.tikz}}
\tikzset{x=1em, y=1.5ex}
}$, then $a > k \vee b > q \vee a + b \neq c$ and thus $(\vect{a \\ b}, c) \not\in \dsem{
\tikzset{x=1em, y=2.1ex}
\InputIfFileExists{flownet/law1-1.tikz}{}{\input{./tikz/flownet/law1-1.tikz}}
\tikzset{x=1em, y=1.5ex}
}$. Finally, since $\dsem{
\tikzset{x=1em, y=2.1ex}
\InputIfFileExists{flownet/law1-1.tikz}{}{\input{./tikz/flownet/law1-1.tikz}}
\tikzset{x=1em, y=1.5ex}
} = \dsem{
\tikzset{x=1em, y=2.1ex}
\InputIfFileExists{flownet/law1-2.tikz}{}{\input{./tikz/flownet/law1-2.tikz}}
\tikzset{x=1em, y=1.5ex}
}$, by Theorem~\ref{co:compl}, the equality holds.
\end{proof}
\begin{proof}[Proof of Derived law~\ref{eq:dlaw2}]
We prove the two inclusions separately
\begin{align*}
    
\tikzset{x=1em, y=2.1ex}
\begin{tikzpicture}
	\begin{pgfonlayer}{nodelayer}
		\node [style=none] (0) at (-1.5, 0) {};
		\node [style=edge] (1) at (0, 0) {$k+l$};
		\node [style=none] (2) at (1.5, 0) {};
	\end{pgfonlayer}
	\begin{pgfonlayer}{edgelayer}
		\draw (0.center) to (1.center);
		\draw (1.center) to (2.center);
	\end{pgfonlayer}
\end{tikzpicture}
}
\tikzset{x=1em, y=1.5ex}
 = &
\tikzset{x=1em, y=2.1ex}
\InputIfFileExists{dlaw2/dir1/step1.tikz}{}{\input{./tikz/dlaw2/dir1/step1.tikz}}
\tikzset{x=1em, y=1.5ex}
 \myeq{\AIHP{}} 
\tikzset{x=1em, y=2.1ex}
\InputIfFileExists{dlaw2/dir1/step2.tikz}{}{\input{./tikz/dlaw2/dir1/step2.tikz}}
\tikzset{x=1em, y=1.5ex}
 \myeq{\AIHP{}} 
\tikzset{x=1em, y=2.1ex}
\InputIfFileExists{dlaw2/dir1/step3.tikz}{}{\input{./tikz/dlaw2/dir1/step3.tikz}}
\tikzset{x=1em, y=1.5ex}
 \myeq{\AIHP{}} \\
    &
\tikzset{x=1em, y=2.1ex}
\InputIfFileExists{dlaw2/dir1/step4.tikz}{}{\input{./tikz/dlaw2/dir1/step4.tikz}}
\tikzset{x=1em, y=1.5ex}
 \myeq{\AIHP{}} 
\tikzset{x=1em, y=2.1ex}
\InputIfFileExists{dlaw2/dir1/step5.tikz}{}{\input{./tikz/dlaw2/dir1/step5.tikz}}
\tikzset{x=1em, y=1.5ex}
 \myeq{\AIHP{}} 
\tikzset{x=1em, y=2.1ex}
\InputIfFileExists{dlaw2/dir1/step6.tikz}{}{\input{./tikz/dlaw2/dir1/step6.tikz}}
\tikzset{x=1em, y=1.5ex}
 \myeq{\AIHP{}} \\
    &
\tikzset{x=1em, y=2.1ex}
\InputIfFileExists{dlaw2/dir1/step7.tikz}{}{\input{./tikz/dlaw2/dir1/step7.tikz}}
\tikzset{x=1em, y=1.5ex}
 \mysupeq{\AIHP{}} 
\tikzset{x=1em, y=2.1ex}
\InputIfFileExists{dlaw2/dir1/step8.tikz}{}{\input{./tikz/dlaw2/dir1/step8.tikz}}
\tikzset{x=1em, y=1.5ex}
 \myeq{\AIHP{}} 
\tikzset{x=1em, y=2.1ex}
\InputIfFileExists{flownet/law2-1.tikz}{}{\input{./tikz/flownet/law2-1.tikz}}
\tikzset{x=1em, y=1.5ex}

\end{align*}
\begin{align*}
    
\tikzset{x=1em, y=2.1ex}
\InputIfFileExists{flownet/law2-1.tikz}{}{\input{./tikz/flownet/law2-1.tikz}}
\tikzset{x=1em, y=1.5ex}
 = &
\tikzset{x=1em, y=2.1ex}
\InputIfFileExists{dlaw2/dir2/step1.tikz}{}{\input{./tikz/dlaw2/dir2/step1.tikz}}
\tikzset{x=1em, y=1.5ex}
 \mysupeq{\AIHP{}} 
\tikzset{x=1em, y=2.1ex}
\InputIfFileExists{dlaw2/dir2/step2.tikz}{}{\input{./tikz/dlaw2/dir2/step2.tikz}}
\tikzset{x=1em, y=1.5ex}
 \myeq{\AIHP{}} 
\tikzset{x=1em, y=2.1ex}
\InputIfFileExists{dlaw2/dir2/step3.tikz}{}{\input{./tikz/dlaw2/dir2/step3.tikz}}
\tikzset{x=1em, y=1.5ex}
 \myeq{\AIHP{}} \\
    &
\tikzset{x=1em, y=2.1ex}
\InputIfFileExists{dlaw2/dir2/step4.tikz}{}{\input{./tikz/dlaw2/dir2/step4.tikz}}
\tikzset{x=1em, y=1.5ex}
 \myeq{\AIHP{}} 
\tikzset{x=1em, y=2.1ex}
\InputIfFileExists{dlaw2/dir2/step5.tikz}{}{\input{./tikz/dlaw2/dir2/step5.tikz}}
\tikzset{x=1em, y=1.5ex}
 \myeq{\AIHP{}} \\ 
    &
\tikzset{x=1em, y=2.1ex}
\InputIfFileExists{dlaw2/dir2/step6.tikz}{}{\input{./tikz/dlaw2/dir2/step6.tikz}}
\tikzset{x=1em, y=1.5ex}
 \myeq{\AIHP{}} 
\tikzset{x=1em, y=2.1ex}
\begin{tikzpicture}
	\begin{pgfonlayer}{nodelayer}
		\node [style=none] (0) at (-1.5, 0) {};
		\node [style=edge] (1) at (0, 0) {$k+l$};
		\node [style=none] (2) at (1.5, 0) {};
	\end{pgfonlayer}
	\begin{pgfonlayer}{edgelayer}
		\draw (0.center) to (1.center);
		\draw (1.center) to (2.center);
	\end{pgfonlayer}
\end{tikzpicture}
}
\tikzset{x=1em, y=1.5ex}

\end{align*}
\end{proof}

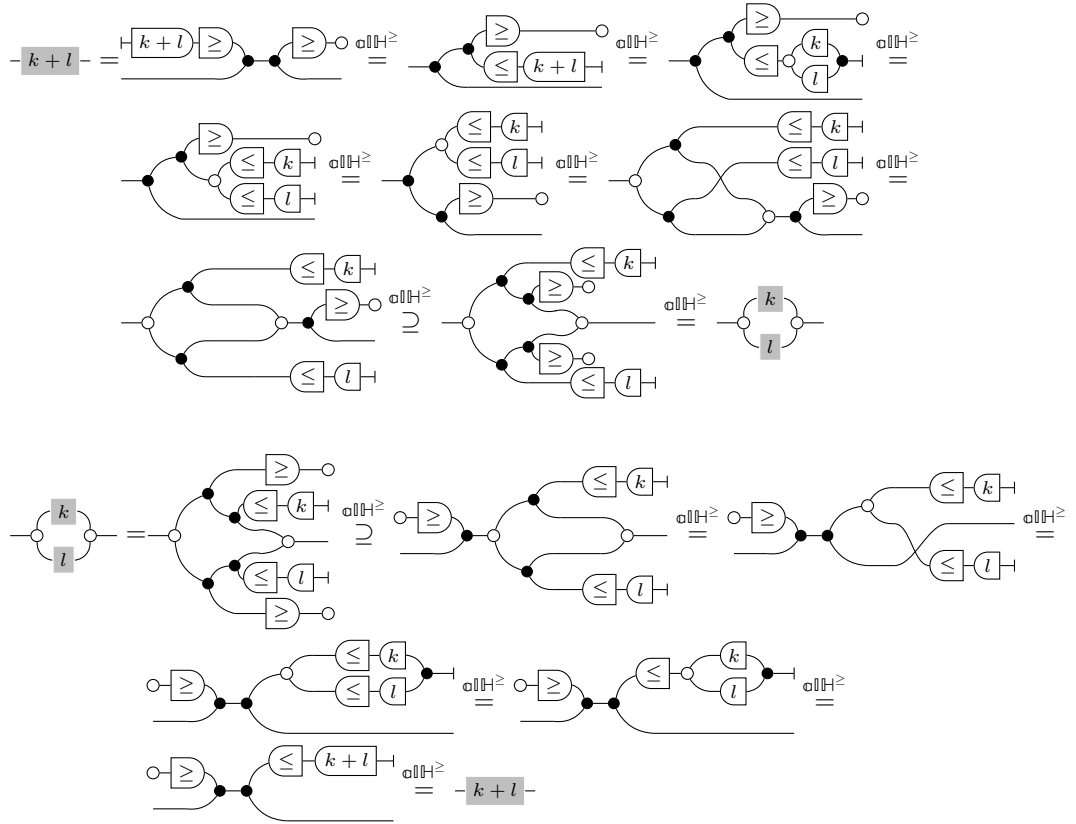
\end{document}